\documentclass[12pt,reqno]{article}
\pdfoutput=1
\usepackage[a4paper,width=170mm,top=25mm,bottom=25mm]{geometry}
\usepackage{hyperref}
\hypersetup{
    colorlinks=false,                       % false: boxed links; true: colored links
    linkcolor=red,                          % color of internal links
    citecolor=blue                          % color of links to bibliography
}
\usepackage{setspace}
\usepackage{fancyhdr}
\usepackage{amsmath,amsthm,amsfonts,amssymb,braket,mathtools,mathrsfs}
\usepackage{url}
\usepackage{bbm, dsfont}
\usepackage{bbold}
\usepackage{enumerate}
\usepackage{multirow}
\usepackage[all]{xy}
\usepackage{color}
\usepackage{graphicx}
\usepackage{caption}
\usepackage{subcaption}
\usepackage{import}
\usepackage{braket}

\numberwithin{equation}{section}
\theoremstyle{plain}
\newtheorem{theorem}{Theorem}[section]
\newtheorem{corollary}[theorem]{Corollary}
\newtheorem{lemma}[theorem]{Lemma}
\newtheorem{proposition}[theorem]{Proposition}
\newtheorem{main theorem}{Main Theorem}
\newtheorem{mainthm}{Theorem}
\theoremstyle{definition}

\newtheorem{remark}[theorem]{Remark}

\newtheorem{definition}[theorem]{Definition}

\newtheorem*{acknowledgements}{Acknowledgements}
\newtheorem*{conflictofinterest}{Conflict of interest statement}
\newtheorem*{dataavailability}{Data availability statement}

\usepackage{tikz}
\usepackage{xy}
\usepackage{tikz-cd}
\usepackage{xcolor}
\usetikzlibrary{decorations.pathmorphing,decorations.text,decorations.markings}

\tikzset{snake/.style={decorate, decoration=snake}}
\usetikzlibrary {arrows.meta,bending,positioning,shapes.geometric}

\newcommand{\ThinLine}{\draw[line width=0.8pt]}

\tikzset{block/.style={rectangle, draw, fill=blue!20, text width=5em, 
    text centered, rounded corners, minimum height=4em},
    cloud/.style={draw, ellipse,fill=red!20, node distance=3cm, minimum height=2em},
    line/.style={draw, -latex'},
    Cbox/.style = {circle, draw, thick, fill=white, opaque}, %% circle box
    midarrow/.style={
        decoration={
            markings,
            mark=at position 0.5 with {\arrow{>}}
        },
        postaction={decorate}
    }
    }

\newcommand{\End}{\text{End}}

\newcommand{\Hom}{\text{Hom}}

\newcommand{\Tr}{\mathrm{Tr}}
\newcommand{\tr}{\mathrm{tr}}
\newcommand{\range}{\mathrm{Range}}
\newcommand{\Irr}{\mathrm{Irr}}

\NewDocumentCommand{\tens}{t_}
 {%
  \IfBooleanTF{#1}
   {\tensop}
   {\otimes}%
 }
\NewDocumentCommand{\tensop}{m}
 {%
  \mathbin{\mathop{\otimes}\displaylimits_{#1}}%
 }

\begin{document}
%%%%%%%%%%%%%%%%% TITLE PAGE %%%%%%%%%%%%%%%
\title{Frustration-Free Models for Topological Holography:\\
Fusion Spin Chains as Boundary Algebras}
\author{
  Zhengwei Liu$^{1,2,3}$\thanks{liuzhengwei@mail.tsinghua.edu.cn}, Zishuo Zhao$^{3}$\thanks{zishuozhao0602@gmail.com} \\
  \\
  \small $^{1}$Department of Mathematics, Tsinghua University, Beijing 100084, China \\
  \small $^{2}$Yanqi Lake Beijing Institute of Mathematical Sciences and Applications, Beijing 100407, China\\
  \small $^{3}$Yau Mathematical Sciences Center, Tsinghua University, Beijing 100084, China
}
\date{\today}
\maketitle
%%%%%%%%%%%%%%%%%%%%%%%%%%%%%%%%%%%%%%%%
\begin{abstract}
    We introduce a class of $2d$ quantum spin systems based on a unitary spherical fusion category with a chosen object. 
    These models have frustration-free, reflection positive, but generally non-commuting local interactions, generalizing the Levin--Wen string-net models. 
    Based on the theory developed in \cite{LiuZhao2025RPTO}, we explicitly realize the Temperley--Lieb--Jones algebra, and more generally fusion spin chains as their boundary algebras. 
\end{abstract}
%%%%%%%%%%%%%% CONTENTS %%%%%%%%%%%%%%%%%%%%
\tableofcontents
%%%%%%%%%%%%%%%%%%%%%%%%%%%%%%%%%%%%%%%%%%%
\section{Introduction}

Topological orders are commonly studied through the ground states of gapped local Hamiltonians. 
In two dimensions, these phases are characterized by topology-dependent ground-state degeneracy, anyonic excitations, and non-trivial boundary effects. 
For paradigmatic microscopic models such as Kitaev's quantum-double models \cite{Kitaev2003} and Levin--Wen string-net models \cite{LevinWen2005}, the commutativity of local terms makes the above features explicitly accessible. 
Beyond the commuting projector ones, frustration-free Hamiltonians form a broader class that are intimately related to general gapped phases \cite{Hastings2006,MichalakisZwolak2013stability}. 
Nevertheless, frustration-free models are more difficult to analyze, due to the non-commutativity between local terms. 

In \cite{LiuZhao2025RPTO}, we introduced reflection positivity \cite{OScombined,GlimmJaffe1987,FILS1978} as a tool for studying topological orders of frustration-free Hamiltonians.
We proved that for a reflection-positive, frustration-free local Hamiltonian on a sphere, the nondegeneracy of the global ground state is equivalent to local indistinguishability of ground states on the disk.
We also showed that Osterwalder--Schrader (OS) reconstruction applied to local ground states produces the boundary algebra. 
Predicted by topological holography, these algebras capture the symmetries described by the bulk theory \cite{KongWenZheng2017BoundaryBulk,JiWeni2019NoninvertibleAnomalies,kong2020AlgebraicHigherSymmetrya}. 

In this paper, we supply this framework with a systematic construction of new $2d$ lattice models.
The construction takes a unitary fusion category $\mathcal{C}$ with a chosen object $\rho\in \mathcal{C}$ as the input. 
The resulting models possess frustration-free, reflection positive, but generally non-commuting local interactions. 
Taking $\mathcal{C}$ to be the Fibonacci category with simple objects $\{\mathbb{1},\tau\}$, with $\tau$ being the chosen object, we obtain a qubit system whose boundary algebra being the Temperley--Lieb--Jones algebra \cite{TemperleyLieb1971,Jones1983} at loop parameter $\frac{1+\sqrt{5}}{2}$. 

Our construction generalizes Levin--Wen string-net models in a way that edges in one direction of the honeycomb lattice are colored by the chosen object $\rho$. 
We do not require $\rho$ to contain every simple object. 
In particular, $\rho$ can be simple, meaning that edges colored by it carry \emph{no} degrees of freedom. 
The local interaction is determined by a set of endomorphisms of tensor powers of $\rho$, which generates the full tensor tower in a suitable sense. 
The resulting interaction may stretch through multiple plaquettes. 
Taking $\rho$ to be the direct sum of all simples recovers the original string-net models. 

From the recent point of views, Temperley--Lieb--Jones algebra fits into spin chains with categorical/MPO symmetries \cite{feiguin2007AnyonChain,aasenfendlymong2020,LootensDelcampVerstraete2024,Jones2024DHRBimodules,LuChatterjeeTantivasadakarn2026}. 
A large class is termed fusion spin chains. 
These spin chains have local operators being endomorphisms of tensor powers of a fixed object in a fusion category, with inclusion maps given by tensoring with identity morphisms \cite{Jones2024DHRBimodules}. 
Through the LTO axioms \cite{JNPW2025,chuah2024boundary,jones2025HolographyBulkboundaryLocal}, the boundary algebras of Levin--Wen string-net models can be identified with the fusion spin chain associated with the object $\oplus_{x\in\Irr(\mathcal{C})}x$. 
The work \cite[Example 2.5]{jones2025HolographyBulkboundaryLocal} suggests that an analogous realization for an arbitrary strong tensor generator $X$ exists, without providing the local interactions. 
Here we prove that OS reconstruction, applied to the model constructed with input category $\mathcal{C}$ and object $\rho$, realizes fusion spin chains determined by $\rho$ as the boundary algebras. 
In particular, fusion spin chains defined by strong tensor generating objects arise as boundary algebras: 

\begin{mainthm}[Theorem \ref{thm:: boundary algebra determined by tensor generator}]
Given $X$ a self-dual strong tensor generating object in a unitary fusion category $\mathcal{C}$. 
We construct a $2$D quantum spin system with frustration-free, reflection positive local interaction, such that the boundary algebra obtained by OS reconstruction is bounded spread isomorphic to the fusion spin chain $\mathcal{A}(\mathcal{C},X)$. 
\end{mainthm}

The models introduced in this paper raise several questions. 
Naturally, one would like to know if these models are in the same phase as Levin--Wen string-net models with the same input fusion category. 
In order to represent a stable phase of matter, one should prove the existence of a uniform gap above the ground state, and verify the LTQO condition \cite{MichalakisZwolak2013stability}, which we suspect to hold approximately for a generic chosen object. 
A related problem is whether there exist strictly local commuting projector Hamiltonians that exhibit the same set of local ground states. 
Alternatively, one consider the superselection theory, as quasi-local automorphisms preserve the fusion and braiding of the superselection sectors \cite{ogata2022DerivationBraidedCtensor}. 
Recent developments in characterizing Haag duality \cite{ogata2025HaagDuality,NaaijkensPenneysWallick2026,LiuVanLuijk2026} might be helpful in proving the approximate Haag duality. 
We anticipate that the theory of DHR-bimodules \cite{Jones2024DHRBimodules,EvansJones2025} will play a crucial role in constructing superselection sectors. 

The paper is organized as follows: in Section \ref{Sec:: two examples}, we illustrate the Fibonacci model of a qubit system; in Section \ref{Sec:: Non-isotropic String-net Model}, we introduce the construction in its generic form; in Section \ref{Sec:: Solving Local Ground States}, we solve the local ground states of the model for a family of exhausting regions, proving that the interaction is frustration-free; in Section \ref{Sec:: Fusion Spin Chains from OS Reconstruction}, we compute the boundary algebra of the models and verify that these algebras are isomorphic to fusion spin chains; finally, in Section \ref{Sec:: Examples}, we present some examples of the model in detail, including a qubit system based on the Fibonacci category, Levin--Wen string-net model, and the generic construction using strong tensor generating objects. 

\begin{acknowledgements}
All authors were supported by Beijing Natural Science Foundation (Grant No. Z221100002722017). 
Zhengwei Liu was supported by Beijing Natural Science Foundation Key Program (Grant No. Z220002). \\
AI disclosure: LLMs were used only for language improvements and drawing figures. 
The authors take full responsibility for the content of the paper. 
\end{acknowledgements}

\section{Trailer: A Qubit System}\label{Sec:: two examples}

Consider a qubit system on the square lattice in which each vertex carries a Hilbert space $\mathbb{C}^2$ with orthonormal basis $\{\ket{0}, \ket{1}\}$. 
The model consists of two types of interactions: one for each horizontal edge, and the other for each rectangle (Fig. \ref{figure:: 2-qubit Fib}). 

\begin{figure}[ht]
    \centering
    \begin{tikzpicture}[line cap=round]
    % Horizontal edges (darker and wider)
    \foreach \j in {0,...,3}
    \draw[gray!80, line width=0.8pt] (0,\j) -- (5,\j);
    % Vertical edges (lighter and thinner)
    \foreach \i in {0,...,5}
    \draw[gray!45, line width=0.5pt] (\i,0) -- (\i,3);
    % Outward half-edges
    \foreach \j in {0,...,3} {
    \draw[gray!80, line width=0.8pt] (0,\j) -- (-0.3,\j);
    \draw[gray!80, line width=0.8pt] (5,\j) -- (5.3,\j);
    }
    \foreach \i in {0,...,5} {
    \draw[gray!45, line width=0.5pt] (\i,0) -- (\i,-0.3);
    \draw[gray!45, line width=0.5pt] (\i,3) -- (\i,3.3);
    }
  % Two-qubit interaction
  \draw[red!55, line width=1.5pt]
    (0.5,1) ellipse [x radius=0.62, y radius=0.2];
  % Eight-qubit interaction
  \draw[blue!55, line width=1.5pt, rounded corners=4pt]
    (0.78,0.78) rectangle (4.22,2.22);
  % Vertices
  \foreach \i in {0,...,5}
    \foreach \j in {0,...,3}
      \fill[black] (\i,\j) circle[radius=1.6pt];
    \node at (2.5,2.5) {$L$};
    \node at (2.5,0.5) {$U$};
\end{tikzpicture}
    \caption{two types of interactions of the model: one for every horizontal edge $e$, and one for $2\times 4$ rectangle $R = L\cup U$.}
    \label{figure:: 2-qubit Fib}
\end{figure}
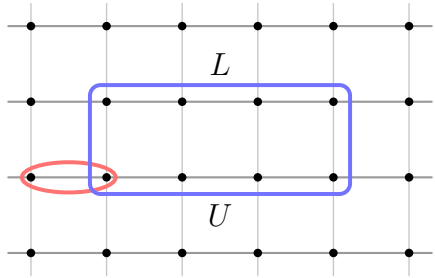

For a horizontal edge $e$ connecting vertices $v,w$, we define $P_e = \ket{1_v1_w}\bra{1_v1_w}$, which punishes the appearance of consecutive $\ket{1}$ in each row. 
To define the block interaction, we introduce the following operators. 
Consider the unitary operator $F = \begin{bmatrix}
    -\frac{1}{\phi} & \frac{1}{\sqrt{\phi}}\\
    \frac{1}{\sqrt{\phi}} & \frac{1}{\phi}
\end{bmatrix}$ where $\phi = \frac{1+\sqrt{5}}{2}$. 
Define $n = (\mathrm{I} - Z)/2$ and $\hat{n} = FnF^\dagger$. 
Then for a set of four consecutive qubits in a row, we have the local operator:
\begin{equation}\label{eqn:: Jones projections on qubits}
    \begin{aligned}
        T &= \phi\, n_1(\mathrm{I}-n_2)\hat{n}_3(\mathrm{I}-n_4) + \phi(\mathrm{I}-n_1)\hat{n}_2(\mathrm{I}-n_3)n_4 \\
        & + \phi^2 (\mathrm{I}-n_1)(\mathrm{I} - P_{23})(\mathrm{I}-\hat{n}_2)(1-\hat{n}_3)(\mathrm{I} - P_{23})(\mathrm{I}-n_4). 
    \end{aligned}
\end{equation} 
Note that $T$ is a four qubit controlled operator, with its action on $2,3$ controlled by $1,4$. 
The interaction between two sets of qubits in adjacent rows $L$ and $U$ in a rectangle $R$ is
\begin{equation}
    F_R = \phi(T_U + T_L) -2T_LT_U. 
\end{equation}
The Hamiltonian is thus defined as
\begin{equation}
    H = \sum_{e} P_e + \sum_{R} F_R. 
\end{equation}
Here the second sum ranges over the rectangles of the form $R_{m,n}=R_0+m(2,0)+n(1,1)$, where $R_0$ is a fixed rectangle. 
Thus, for each adjacent pair of rows, the four-qubit windows have the parity inherited from the honeycomb plaquettes.

The operator $T$ has the following interesting properties: if we number the vertices in one row by integers from left to right, and denote by $T_{x}$ the operator $T$ acting on the qubits $x,x+1,x+2,x+3$, then on the ground state space of $\sum_{e}P_e$, we have 
\begin{equation}
    \begin{aligned}
        T^2_x &= \phi T_x;\\
         T_xT_{x+2}T_{x} = T_x,&\quad T_{x+2}T_{x}T_{x+2} = T_{x+2}, 
    \end{aligned}
\end{equation}
which is precisely the Temperley-Lieb relation. 
Moreover, one may verify that 
\begin{equation}
   [T_x,T_{x+j}]=0,\quad j\neq 2,-2. 
\end{equation}
In other words, restricted to the ground state subspace of $\sum_{e}P_e$, the sets $\{T_{1},T_{3},T_{5},\dots\}$ and $\{T_{2},T_{4},T_{6},\dots\}$ generate two copies of Temperley--Lieb algebra commuting with each other. 
This feature will be essential for analyzing the ground state and constructing the boundary algebra. 

\section{Non-isotropic String-net Model}\label{Sec:: Non-isotropic String-net Model}

\subsection{Preliminaries}

Let $\mathcal{C}$ be a (strict) spherical unitary fusion category. 
We denote the set of objects in $\mathcal{C}$ by $\mathrm{Obj}(\mathcal{C})$ and a set of representatives for the isomorphism classes of simple objects by $\Irr(\mathcal{C})$. 
We use the shorthand notation $ab$ to denote the tensor product of objects $a$ and $b$.  
The morphism space from $a$ to $b$ is denoted by $\Hom(a,b)$, and $\End(a)  =\Hom(a,a)$. 
The identity morphism of $a$ is denoted by $\mathrm{Id}_a$. 
For an object $a$, its dual is denoted by $\overline{a}$, with $a$ identified with its double dual: $\overline{\overline{a}} = a$. 
We denote by $\tr_{\mathcal{C}}$ the categorical trace on $\mathcal{C}$, and by $d_a = \tr_{\mathcal{C}}(\mathrm{Id}_a)$ the quantum dimension of $a$. 
We use the standard graphical calculus of unitary fusion categories, with diagrams read from top to bottom. 
We represent the unit object $\mathbb{1}$ by a dashed line or simply empty space. 
Strands with upward-pointing arrows represent dual objects. 

All inner products in this paper are anti-linear in the first entry. 
In the following, we'll often indicate the support of an local operator using subscripts, and we omit the tensoring with identities out side the support when doing so. 

\subsection{Local Interaction}

Let $\mathcal{C}$ be a unitary spherical fusion category (UFC).
The model constructed below has local degrees of freedom determined by a chosen object
\begin{equation}
    \rho = \bigoplus_{a\in \Irr(\mathcal{C})} n_a\cdot a, \quad n_a\in \mathbb{N}_{\geq 0}. 
\end{equation}
For simplicity, we primarily consider the case in which $\rho$ is self-dual. 
By replacing tensor powers of $\rho$ with alternating tensor products of $\rho$ and $\bar{\rho}$, the results remain valid for the non-self-dual case. 

Consider the following honeycomb lattice $\Gamma$ in $\mathbb{R}^2$, where the vertical edges are colored by the chosen object $\rho$ (orange edges in Fig. \ref{fig:: colored honeycomb lattice}). 
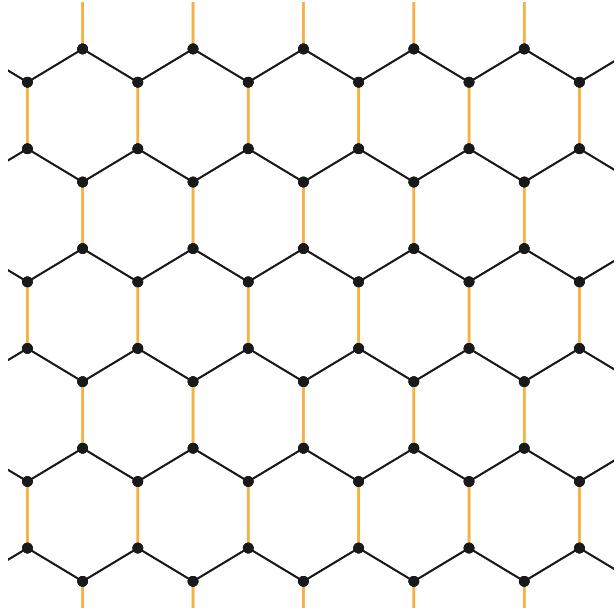
\begin{figure}[ht]
    \centering
    \definecolor{bondorange}{RGB}{245,177,67}
    \definecolor{bondblack}{RGB}{25,25,25}
     \begin{tikzpicture}[
        x=1cm,
        y=1cm,
        line cap=round,
        line join=round
    ]
        \path[use as bounding box] (0,0) rectangle (8,8);
        \clip (0,0) rectangle (8,8);

        % Shift the whole honeycomb pattern to the right.
        \begin{scope}[xshift=0.25cm]

            \def\dx{0.73}
            \def\dy{0.44}

            % Orange vertical bonds
            \foreach \cy/\shift in {
                -0.10/0,
                 1.22/0.73,
                 2.54/0,
                 3.86/0.73,
                 5.18/0,
                 6.50/0.73,
                 7.82/0,
                 9.14/0.73
            }{
                \foreach \xb in {
                    -2.92,-1.46,0,1.46,2.92,4.38,
                    5.84,7.30,8.76,10.22
                }{
                    \pgfmathsetmacro{\cx}{\xb+\shift}

                    \draw[
                        bondorange,
                        line width=0.95pt
                    ]
                        ({\cx-\dx},{\cy-\dy})
                        --
                        ({\cx-\dx},{\cy+\dy});

                    \draw[
                        bondorange,
                        line width=0.95pt
                    ]
                        ({\cx+\dx},{\cy-\dy})
                        --
                        ({\cx+\dx},{\cy+\dy});
                }
            }

            % Black diagonal bonds
            \foreach \cy/\shift in {
                -0.10/0,
                 1.22/0.73,
                 2.54/0,
                 3.86/0.73,
                 5.18/0,
                 6.50/0.73,
                 7.82/0,
                 9.14/0.73
            }{
                \foreach \xb in {
                    -2.92,-1.46,0,1.46,2.92,4.38,
                    5.84,7.30,8.76,10.22
                }{
                    \pgfmathsetmacro{\cx}{\xb+\shift}

                    \draw[
                        bondblack,
                        line width=0.70pt
                    ]
                        ({\cx-\dx},{\cy+\dy})
                        --
                        (\cx,{\cy+2*\dy})
                        --
                        ({\cx+\dx},{\cy+\dy});

                    \draw[
                        bondblack,
                        line width=0.70pt
                    ]
                        ({\cx-\dx},{\cy-\dy})
                        --
                        (\cx,{\cy-2*\dy})
                        --
                        ({\cx+\dx},{\cy-\dy});
                }
            }

            % Black vertices
            \foreach \cy/\shift in {
                -0.10/0,
                 1.22/0.73,
                 2.54/0,
                 3.86/0.73,
                 5.18/0,
                 6.50/0.73,
                 7.82/0,
                 9.14/0.73
            }{
                \foreach \xb in {
                    -2.92,-1.46,0,1.46,2.92,4.38,
                    5.84,7.30,8.76,10.22
                }{
                    \pgfmathsetmacro{\cx}{\xb+\shift}

                    \foreach \vx/\vy in {
                        {\cx-\dx}/{\cy+\dy},
                        {\cx}/{\cy+2*\dy},
                        {\cx+\dx}/{\cy+\dy},
                        {\cx+\dx}/{\cy-\dy},
                        {\cx}/{\cy-2*\dy},
                        {\cx-\dx}/{\cy-\dy}
                    }{
                        \fill[bondblack]
                            (\vx,\vy) circle[radius=2pt];
                    }
                }
            }

        \end{scope}
    \end{tikzpicture}
    \caption{The colored honeycomb lattice $\Gamma$: 
    each vertical edge is colored by the chosen object $\rho$ in $\mathcal{C}$.}
    \label{fig:: colored honeycomb lattice}
\end{figure}
We choose to assign local Hilbert spaces to vertices only. 
The local Hilbert space assigned to each vertex is obtained by interpreting the uncolored edges as the object $\displaystyle \gamma = \bigoplus_{a\in \Irr(\mathcal{C})} a$, the colored edge as the chosen object $\rho$, and the vertex as a morphism in the corresponding morphism space. 
Specifically, there are two types of local Hilbert spaces:
\begin{equation}
    \begin{aligned}
        \vcenter{\hbox{\begin{tikzpicture}
    \ThinLine[orange] (0:0) -- (-90:0.5);
    \ThinLine (0:0) -- (150:0.5);
    \ThinLine (0:0) -- (30:0.5);
    \draw[fill=black] (0:0) ellipse (0.05 and 0.05);
    \node at (-0.2,-0.25) {$v$};
\end{tikzpicture}}}: \bigoplus_{a,b\in \Irr(\mathcal{C})}\Hom(ab,\rho)\quad 
\vcenter{\hbox{\begin{tikzpicture}
    \ThinLine[orange] (90:0.5) -- (0:0);
    \ThinLine (0:0) -- (-150:0.5);
    \ThinLine (0:0) -- (-30:0.5);
    \draw[fill=black] (0:0) ellipse (0.05 and 0.05);
    \node at (-0.2,0.25) {$v$};
\end{tikzpicture}}}: \bigoplus_{a,b\in \Irr(\mathcal{C})}\Hom(\rho,a b). 
    \end{aligned}
\end{equation}
To define the inner product on the local Hilbert space at a vertex $\vcenter{\hbox{\begin{tikzpicture}
    \ThinLine[orange] (0:0) -- (-90:0.5);
    \ThinLine (0:0) -- (150:0.5);
    \ThinLine (0:0) -- (30:0.5);
    \draw[fill=black] (0:0) ellipse (0.05 and 0.05);
    \node at (-0.2,-0.25) {$v$};
\end{tikzpicture}}}$, denote by $q_c$ the projection in $\End(\rho)$ onto the summand $n_c\cdot c$, where $q_c=0$ if $n_c=0$. 
Define the following positive operator in $\End(\rho)$
\begin{equation}
    \mathrm{D}_{\rho} = \sum_{c\in \Irr(\mathcal{C})}\frac{1}{\sqrt{d_c}}q_c.
\end{equation}
Then for simple objects $a,b\in \Irr(\mathcal{C})$, the inner product on the summand $\Hom(ab,\rho)$ is
\begin{equation}
    \braket{\eta|\xi} = \frac{1}{\sqrt{d_a d_b}}\tr_{\mathcal{C}}(\eta^* \mathrm{D}_\rho \xi),\quad  \eta,\xi\in \Hom(ab,\rho).
\end{equation}
The inner product on the other type of local Hilbert space is defined similarly. 

By a region, we mean a $2d$ submanifold of $\mathbb{R}^2$ whose boundary intersects edges of $\Gamma$ transversely. 
A region $\Lambda$ is \emph{finite} if the set $V(\Lambda)$ of vertices of $\Gamma$ in its interior is finite; let $E(\Lambda)$ denote the set of edges whose endpoints both lie in $V(\Lambda)$. 
For a finite region $\Lambda$, the local Hilbert space is defined as 
\begin{equation}
    \mathcal{H}_{\Lambda} = \bigotimes_{v\in V(\Lambda)}\mathcal{H}_v. 
\end{equation}
We call a tensor product state $\bigotimes_{v\in V(\Lambda)} \ket{\xi_{v}}$ a configuration on $\Lambda$ if every $\xi_v$ is chosen from a summand of $\mathcal{H}_v$. 

There are two types of local interactions in the model, with the edge interactions paralleling those of string-net models. 
For an un-colored edge $e$ with endpoints $v,w$, define $Q_e$ as the orthogonal projection on $\mathcal{H}_v\otimes \mathcal{H}_w$ onto the subspace spanned by configurations that assign the same simple object to $e$.  
The interaction associated with each un-colored edge is therefore
\begin{equation}
    \varPhi(e) = \mathrm{I} - Q_e. 
\end{equation}
For a finite region $\Lambda$ containing consecutive vertices in a row, we define the common image of the projections $Q_e$ over $e\in E(\Lambda)$, which is the ground state subspace of the Hamiltonian $\sum_{e\in E(\Lambda)}\mathrm{I} - Q_e$, to be the \emph{string-net subspace} of $\mathcal{H}_{\Lambda}$, denoted as $\mathcal{H}^{\mathrm{st.n.}}_{\Lambda}$. 

The second type of interaction is assigned to regions consisting of consecutive plaquettes. 
To define these interactions, we recall the identification of string-net subspaces with morphism spaces in $\mathcal{C}$. 
Consider a region $\Lambda\subset \mathbb{R}^2$ that contains $2l$ consecutive vertices in a row for some $l\geq 1$: 
\begin{equation}\label{fig:: 1st form of Lambda}
    \vcenter{\hbox{\begin{tikzpicture}
        \ThinLine[orange] (0:0) -- (-90:0.5);
    \ThinLine (0:0) -- (150:0.5);
    \ThinLine (0:0) -- (30:0.5);
    \draw[fill=black] (0:0) ellipse (0.05 and 0.05);
    \node at (0,0.3) {$v_1$};
    \begin{scope}[shift={(0.866,0.5)}]
        \ThinLine[orange] (0:0) -- (90:0.5);
    \ThinLine (0:0) -- (-150:0.5);
    \ThinLine (0:0) -- (-30:0.5);
    \draw[fill=black] (0:0) ellipse (0.05 and 0.05);
    \node at (0,-0.3) {$v_2$};
    \end{scope}
    \node at (1.73,0.25) {$\cdots$};
    \begin{scope}[shift={(2.596,0)}]
        \ThinLine[orange] (0:0) -- (-90:0.5);
    \ThinLine (0:0) -- (150:0.5);
    \ThinLine (0:0) -- (30:0.5);
    \draw[fill=black] (0:0) ellipse (0.05 and 0.05);
    \node at (0,0.3) {$v_{2l-1}$};
    \begin{scope}[shift={(0.866,0.5)}]
        \ThinLine[orange] (0:0) -- (90:0.5);
    \ThinLine (0:0) -- (-150:0.5);
    \ThinLine (0:0) -- (-30:0.5);
    \draw[fill=black] (0:0) ellipse (0.05 and 0.05);
    \node at (0,-0.3) {$v_{2l}$};
    \end{scope}
    \end{scope}
    \begin{scope}[overlay]
    % Mask only the small region immediately outside the parallelogram
    \fill[
        white,
        even odd rule,
        rounded corners=5pt
    ]
        (-0.56,-0.56) rectangle (3.96,1.06)
        (-0.5,-0.35) --
        ( 3.58,-0.35) --
        ( 3.88, 0.85) --
        (-0.2, 0.85) -- cycle;
    % Draw the parallelogram over the truncated edge endpoints
    \draw[
        densely dashed,
        rounded corners=5pt,
        line width=0.5pt
    ]
        (-0.5,-0.35) --
        ( 3.58,-0.35) --
        ( 3.88, 0.85) --
        (-0.2, 0.85) -- cycle;
\end{scope}
    \end{tikzpicture}}}
\end{equation}
Given a configuration $\ket{\xi} = \bigotimes^{2l}_{i=1}\ket{\xi_i} \in \mathcal{H}_{\Lambda}$, we define $\mathrm{eval}_{\Lambda}(\xi) \in \Hom(\gamma\rho^{l},\rho^{l}\gamma)$ by connecting the uncolored edges in $\Lambda$:
\begin{equation}
    \mathrm{eval}_{\Lambda}: \vcenter{\hbox{\begin{tikzpicture}
        \ThinLine[orange] (0:0) -- (-90:0.5);
    \ThinLine (0:0) -- (150:0.5);
    \ThinLine (0:0) -- (30:0.5);
    \draw[fill=black] (0:0) ellipse (0.05 and 0.05);
    \node at (0,0.375) {$\xi_1$};
    \node at (0.65,0.375) {$\otimes$};
    \begin{scope}[shift={(1.299,0.75)}]
        \ThinLine[orange] (0:0) -- (90:0.5);
    \ThinLine (0:0) -- (-150:0.5);
    \ThinLine (0:0) -- (-30:0.5);
    \draw[fill=black] (0:0) ellipse (0.05 and 0.05);
    \node at (0,-0.375) {$\xi_2$};
    \end{scope}
    \node at (2.25,0.25) {$\cdots$};
    \begin{scope}[shift={(3.25,0)}]
        \ThinLine[orange] (0:0) -- (-90:0.5);
    \ThinLine (0:0) -- (150:0.5);
    \ThinLine (0:0) -- (30:0.5);
    \draw[fill=black] (0:0) ellipse (0.05 and 0.05);
    \node at (0,0.425) {$\xi_{2l-1}$};
    \node at (0.65,0.375) {$\otimes$};
    \begin{scope}[shift={(1.299,0.75)}]
        \ThinLine[orange] (0:0) -- (90:0.5);
    \ThinLine (0:0) -- (-150:0.5);
    \ThinLine (0:0) -- (-30:0.5);
    \draw[fill=black] (0:0) ellipse (0.05 and 0.05);
    \node at (0,-0.375) {$\xi_{2l}$};
    \end{scope}
    \end{scope}
    \end{tikzpicture}}} \longmapsto \vcenter{\hbox{\begin{tikzpicture}
        \ThinLine[orange] (0:0) -- (-90:0.5);
    \ThinLine (0:0) -- (150:0.5);
    \ThinLine (0:0) -- (30:0.5);
    \draw[fill=black] (0:0) ellipse (0.05 and 0.05);
    \node at (0,0.35) {$\xi_1$};
    \begin{scope}[shift={(0.866,0.5)}]
        \ThinLine[orange] (0:0) -- (90:0.5);
    \ThinLine (0:0) -- (-150:0.5);
    \ThinLine (0:0) -- (-30:0.5);
    \draw[fill=black] (0:0) ellipse (0.05 and 0.05);
    \node at (0,-0.35) {$\xi_2$};
    \end{scope}
    \node at (1.73,0.25) {$\cdots$};
    \begin{scope}[shift={(2.596,0)}]
        \ThinLine[orange] (0:0) -- (-90:0.5);
    \ThinLine (0:0) -- (150:0.5);
    \ThinLine (0:0) -- (30:0.5);
    \draw[fill=black] (0:0) ellipse (0.05 and 0.05);
    \node at (0,0.425) {$\xi_{2l-1}$};
    \begin{scope}[shift={(0.866,0.5)}]
        \ThinLine[orange] (0:0) -- (90:0.5);
    \ThinLine (0:0) -- (-150:0.5);
    \ThinLine (0:0) -- (-30:0.5);
    \draw[fill=black] (0:0) ellipse (0.05 and 0.05);
    \node at (0,-0.35) {$\xi_{2l}$};
    \end{scope}
    \end{scope}
    \end{tikzpicture}}}
\end{equation}
We thus obtain a linear map $\mathrm{eval}_{\Lambda}:\mathcal{H}_{\Lambda}\rightarrow \Hom(\gamma\rho^{l},\rho^{l}\gamma)$. 

In what follows, we equip $\Hom(\gamma\rho^{l},\rho^{l}\gamma)$ with the following inner product. 
Write $\mathrm{D}_{\rho^l}:=\mathrm{D}_{\rho}^{\tens l}$. 
For objects $a,b\in \Irr(\mathcal{C})$ and $f,g\in \Hom(a\rho^{l},\rho^{l}b)$, we define
\begin{equation}\label{eqn:: inner product on Hom}
    \braket{f,g} = \frac{1}{\sqrt{d_a d_b}}\tr_{\mathcal{C}}\left( f^* (\mathrm{D}_{\rho}^{\tens l}\otimes \mathrm{Id}_b) g (\mathrm{Id}_a\otimes \mathrm{D}_{\rho}^{\tens l}) \right).
\end{equation}

\begin{lemma}\label{lemma:: eval in a single row}
    The map $\mathrm{eval}_{\Lambda}: \mathcal{H}_{\Lambda} \rightarrow \Hom(\gamma \rho^l, \rho^l\gamma)$ is surjective and satisfies
    \begin{equation}
        \braket{\mathrm{eval}_{\Lambda}(\eta),\mathrm{eval}_{\Lambda}(\xi)} = \braket{\eta|\prod_{e\in E(\Lambda)}Q_e|\xi},\quad \eta,\xi\in \mathcal{H}_{\Lambda}.
    \end{equation}
\end{lemma}
\begin{proof}
    Identifying $\Hom(\gamma \rho^l, \rho^l\gamma)$ with $\Hom(\rho^l\gamma \rho^l,\gamma)$ by Frobenius reciprocity, we see that the image of $\mathrm{eval}_{\Lambda}$ is spanned by fusion trees of the form
    \begin{equation}
        \vcenter{\hbox{\begin{tikzpicture}
            \ThinLine (0,1) -- (0,0.5);
            \ThinLine[orange] (-0.25,1) -- (0,0.5);
            \draw[fill=black] (0,0.5) ellipse (0.05 and 0.05);
            \ThinLine (0,0.5) -- (0,0);
            \ThinLine[orange] (0.5,1) -- (0,0);
            \draw[fill=black] (0,0) ellipse (0.05 and 0.05);
            \ThinLine (0,0) -- (0,-0.15);
            \node at (0,-0.3) {$\vdots$};
            \ThinLine (0,-0.45) -- (0,-0.6);
            \ThinLine[orange] (-0.8,1) -- (0,-0.6);
            \draw[fill=black] (0,-0.6) ellipse (0.05 and 0.05);
            \ThinLine (0,-0.6) -- (0,-1.1);
            \ThinLine[orange] (1.05,1) -- (0,-1.1);
            \draw[fill=black] (0,-1.1) ellipse (0.05 and 0.05);
            \ThinLine (0,-1.1) -- (0,-1.6);
        \end{tikzpicture}}}
    \end{equation}
    with each uncolored edge labeled by a simple object in $\Irr(\mathcal{C})$. 
    By semisimplicity of $\mathcal{C}$, every morphism in $\Hom(\rho^l\gamma \rho^l,\gamma)$ is a linear combination of such fusion trees, and this proves the surjectivity. 
    To prove the asserted identity, first note that both sides evaluate to zero if the edge labels of $\ket{\eta}$ or $\ket{\xi}$ are incompatible. 
    Now take $a_1,\dots,a_{2l+1}\in \Irr(\mathcal{C})$, and choose 
    \begin{equation}
        \ket{\xi} = \bigotimes^{2l}_{i=1}\ket{\xi_i}, \ket{\eta} = \bigotimes^{2l}_{i=1}\ket{\eta_i}\in \mathcal{H}_{\Lambda}
    \end{equation}
    such that $\xi_{2k-1},\eta_{2k-1} \in \Hom(a_{2k-1}a_{2k},\rho)$ and $\xi_{2k},\eta_{2k} \in \Hom(\rho,a_{2k}a_{2k+1})$ for $1\leq k\leq l$. 
    Then we have 
    \begin{equation}
        \begin{aligned}
            & \Braket{\mathrm{eval}_{\Lambda}\left( \bigotimes^{2l}_{i=1} \eta_{i}\right)|\mathrm{eval}_{\Lambda}\left( \bigotimes^{2l}_{i=1} \xi_{i}\right)} = \frac{1}{\sqrt{d_{a_1} d_{a_{2l+1}}}}\, \tr_{\mathcal{C}}\left( \vcenter{\hbox{\begin{tikzpicture}
                \ThinLine[orange] (0:0) -- (-90:0.5);
            \ThinLine (0:0) -- (150:0.5);
            \ThinLine (0:0) -- (30:0.5);
            \draw[fill=black] (0:0) ellipse (0.05 and 0.05);
            \node at (0,0.35) {$\xi_1$};
            \begin{scope}[shift={(0.866,0.5)}]
                \ThinLine[orange] (0:0) -- (90:1.25);
            \ThinLine (0:0) -- (-150:0.5);
            \ThinLine (0:0) -- (-30:0.5);
            \draw[fill=black] (0:0) ellipse (0.05 and 0.05);
            \node at (0,-0.35) {$\xi_2$};
            \end{scope}
            \node at (1.73,0.25) {$\cdots$};
            \begin{scope}[shift={(2.596,0)}]
                \ThinLine[orange] (0:0) -- (-90:0.5);
            \ThinLine (0:0) -- (150:0.5);
            \ThinLine (0:0) -- (30:0.5);
            \draw[fill=black] (0:0) ellipse (0.05 and 0.05);
            \node at (0,0.425) {$\xi_{2l-1}$};
            \begin{scope}[shift={(0.866,0.5)}]
                \ThinLine[orange] (0:0) -- (90:1.25);
            \ThinLine (0:0) -- (-150:0.5);
            \ThinLine (0:0) -- (-30:0.5);
            \draw[fill=black] (0:0) ellipse (0.05 and 0.05);
            \node at (0,-0.35) {$\xi_{2l}$};
            \end{scope}
            \end{scope}
            \begin{scope}[shift={(0,-1.5)},yscale=-1]
                    \ThinLine[orange] (0:0) -- (-90:0.5);
            \ThinLine (0:0) -- (150:0.5);
            \ThinLine (0:0) -- (30:0.5);
            \draw[fill=black] (0:0) ellipse (0.05 and 0.05);
            \node at (0,0.35) {$\eta^*_1$};
            \begin{scope}[shift={(0.866,0.5)}]
                \ThinLine[orange] (0:0) -- (90:0.5);
            \ThinLine (0:0) -- (-150:0.5);
            \ThinLine (0:0) -- (-30:0.5);
            \draw[fill=black] (0:0) ellipse (0.05 and 0.05);
            \node at (0,-0.35) {$\eta^*_2$};
            \end{scope}
            \node at (1.73,0.25) {$\cdots$};
            \begin{scope}[shift={(2.596,0)}]
                \ThinLine[orange] (0:0) -- (-90:0.5);
            \ThinLine (0:0) -- (150:0.5);
            \ThinLine (0:0) -- (30:0.5);
            \draw[fill=black] (0:0) ellipse (0.05 and 0.05);
            \node at (0,0.35) {$\eta^*_{2l-1}$};
            \begin{scope}[shift={(0.866,0.5)}]
                \ThinLine[orange] (0:0) -- (90:0.5);
            \ThinLine (0:0) -- (-150:0.5);
            \ThinLine (0:0) -- (-30:0.5);
            \draw[fill=black] (0:0) ellipse (0.05 and 0.05);
            \node at (0,-0.35) {$\eta^*_{2l}$};
            \end{scope}
            \end{scope}
            \end{scope}
            \node [draw, minimum width=0.25cm, minimum height=0.25cm, fill = white] at (0,-0.75) {\scriptsize $\mathrm{D}_\rho$};
            \node [draw, minimum width=0.25cm, minimum height=0.25cm, fill = white] at (2.596,-0.75) {\scriptsize $\mathrm{D}_\rho$};
            \node [draw, minimum width=0.25cm, minimum height=0.25cm, fill = white] at (0.866,1.25) {\scriptsize $\mathrm{D}_\rho$};
            \node [draw, minimum width=0.25cm, minimum height=0.25cm, fill = white] at (3.462,1.25) {\scriptsize $\mathrm{D}_\rho$};
            \end{tikzpicture}}} \right)\\
            &= \frac{1}{\sqrt{d_{a_1} d_{a_{2l+1}}}}\frac{1}{d_{a_2}}\, \tr_{\mathcal{C}}\left( \vcenter{\hbox{\begin{tikzpicture}
                \ThinLine[orange] (0:0) -- (-90:0.5);
            \ThinLine (0:0) -- (150:0.5);
            \ThinLine (0:0) -- (30:0.5);
            \draw[fill=black] (0:0) ellipse (0.05 and 0.05);
            \node at (0,0.35) {$\xi_1$};
            \begin{scope}[shift={(0,-1)},yscale=-1]
                \ThinLine[orange] (0:0) -- (-90:0.5);
            \ThinLine (0:0) -- (150:0.5);
            \ThinLine (0:0) -- (30:0.5);
            \draw[fill=black] (0:0) ellipse (0.05 and 0.05);
            \node at (0,0.35) {$\eta^*_1$};
            \end{scope}
            \node [draw, minimum width=0.25cm, minimum height=0.25cm, fill = white] at (0,-0.5) {\scriptsize $\mathrm{D}_\rho$};
            \end{tikzpicture}}} \right) \tr_{\mathcal{C}}\left( \vcenter{\hbox{\begin{tikzpicture}
            \begin{scope}[shift={(0.866,0.5)}]
                \ThinLine[orange] (0:0) -- (90:1.25);
            \ThinLine (0:0) -- (-150:0.5);
            \ThinLine (0:0) -- (-30:0.5);
            \draw[fill=black] (0:0) ellipse (0.05 and 0.05);
            \node at (0,-0.35) {$\xi_2$};
            \end{scope}
            \node at (1.73,0.25) {$\cdots$};
            \begin{scope}[shift={(2.596,0)}]
                \ThinLine[orange] (0:0) -- (-90:0.5);
            \ThinLine (0:0) -- (150:0.5);
            \ThinLine (0:0) -- (30:0.5);
            \draw[fill=black] (0:0) ellipse (0.05 and 0.05);
            \node at (0,0.425) {$\xi_{2l-1}$};
            \begin{scope}[shift={(0.866,0.5)}]
                \ThinLine[orange] (0:0) -- (90:1.25);
            \ThinLine (0:0) -- (-150:0.5);
            \ThinLine (0:0) -- (-30:0.5);
            \draw[fill=black] (0:0) ellipse (0.05 and 0.05);
            \node at (0,-0.35) {$\xi_{2l}$};
            \end{scope}
            \end{scope}
            \begin{scope}[shift={(0,-1.5)},yscale=-1]
            \begin{scope}[shift={(0.866,0.5)}]
                \ThinLine[orange] (0:0) -- (90:0.5);
            \ThinLine (0:0) -- (-150:0.5);
            \ThinLine (0:0) -- (-30:0.5);
            \draw[fill=black] (0:0) ellipse (0.05 and 0.05);
            \node at (0,-0.35) {$\eta^*_2$};
            \end{scope}
            \node at (1.73,0.25) {$\cdots$};
            \begin{scope}[shift={(2.596,0)}]
                \ThinLine[orange] (0:0) -- (-90:0.5);
            \ThinLine (0:0) -- (150:0.5);
            \ThinLine (0:0) -- (30:0.5);
            \draw[fill=black] (0:0) ellipse (0.05 and 0.05);
            \node at (0,0.35) {$\eta^*_{2l-1}$};
            \begin{scope}[shift={(0.866,0.5)}]
                \ThinLine[orange] (0:0) -- (90:0.5);
            \ThinLine (0:0) -- (-150:0.5);
            \ThinLine (0:0) -- (-30:0.5);
            \draw[fill=black] (0:0) ellipse (0.05 and 0.05);
            \node at (0,-0.35) {$\eta^*_{2l}$};
            \end{scope}
            \end{scope}
            \end{scope}
            \node [draw, minimum width=0.25cm, minimum height=0.25cm, fill = white] at (2.596,-0.75) {\scriptsize $\mathrm{D}_\rho$};
            \node [draw, minimum width=0.25cm, minimum height=0.25cm, fill = white] at (0.866,1.25) {\scriptsize $\mathrm{D}_\rho$};
            \node [draw, minimum width=0.25cm, minimum height=0.25cm, fill = white] at (3.462,1.25) {\scriptsize $\mathrm{D}_\rho$};
            \end{tikzpicture}}} \right). 
        \end{aligned}
    \end{equation}
    In the last step we used the following identity: 
    \begin{equation}
        \vcenter{\hbox{
\begin{tikzpicture}[x=0.75pt,y=0.75pt,yscale=-1,xscale=1,every path/.style={line width=1pt},every node/.style={font=\small}]
%Straight Lines [id:da16113389880320583] 
\draw    (77.85,62) -- (97.85,82) ;
%Straight Lines [id:da8028522241837003] 
\draw    (97.85,82) -- (117.85,62) ;
%Straight Lines [id:da010259258749105493] 
\draw[orange]    (97.85,82) -- (97.85,102) ;
\draw[fill=black] (97.85,82) ellipse (1.5 and 1.5);
%Straight Lines [id:da44263313631044565] 
\draw    (77.85,142) -- (97.85,122) ;
%Straight Lines [id:da26463439435718594] 
\draw    (97.85,122) -- (117.85,142) ;
%Straight Lines [id:da9879392152779412] 
\draw[orange]    (97.85,122) -- (97.85,102) ;
\draw[fill=black] (97.85,122) ellipse (1.5 and 1.5);
%Curve Lines [id:da2329477327898768] 
\draw    (77.85,62) .. controls (62.25,54.89) and (61.85,149.29) .. (77.85,142) ;
% Text Node
\draw (83,79.6) node [anchor=north west][inner sep=0.75pt]  {$\xi $};
% Text Node
\draw (83,113.2) node [anchor=north west][inner sep=0.75pt]  {$\eta^{*}$};
% Text Node
\draw (117,46.8) node [anchor=north west][inner sep=0.75pt]  {$c$};
% Text Node
\draw (116.25,143.4) node [anchor=north west][inner sep=0.75pt]  {$c$};
% Text Node
\draw (101.05,99) node [anchor=north west][inner sep=0.75pt]  {$\rho$};
\end{tikzpicture}
}} = \frac{1}{d_c}\tr_{\mathcal{C}}(\eta^* \xi)\cdot \vcenter{\hbox{
    \begin{tikzpicture}
        \draw[line width=1pt] (0,1) -- (0,-1);
        \node at (0.5,0) {$c$};
    \end{tikzpicture}
}},\quad c\in \Irr(\mathcal{C}). 
    \end{equation}
Since $ \frac{1}{\sqrt{d_{a_1} d_{a_{2l+1}}}}\frac{1}{d_{a_2}} = \frac{1}{\sqrt{d_{a_1} d_{a_2}}}\frac{1}{\sqrt{d_{a_2} d_{a_{2l+1}}}}$, 
by repeating the above step we get: 
\begin{equation}
    \begin{aligned}
        &\braket{ \eta_1|\xi_1}\, \frac{1}{\sqrt{d_{a_2}d_{a_{2l+1}}}}\, \tr_{\mathcal{C}}\left( \vcenter{\hbox{\begin{tikzpicture}
            \begin{scope}[shift={(0.866,0.5)}]
                \ThinLine[orange] (0:0) -- (90:1.25);
            \ThinLine (0:0) -- (-150:0.5);
            \ThinLine (0:0) -- (-30:0.5);
            \draw[fill=black] (0:0) ellipse (0.05 and 0.05);
            \node at (0,-0.35) {$\xi_2$};
            \end{scope}
            \node at (2.6,0.25) {$\cdots$};
            \begin{scope}[shift={(3.466,0)}]
                \ThinLine[orange] (0:0) -- (-90:0.5);
            \ThinLine (0:0) -- (150:0.5);
            \ThinLine (0:0) -- (30:0.5);
            \draw[fill=black] (0:0) ellipse (0.05 and 0.05);
            \node at (0,0.425) {$\xi_{2l-1}$};
            \begin{scope}[shift={(0.866,0.5)}]
                \ThinLine[orange] (0:0) -- (90:1.25);
            \ThinLine (0:0) -- (-150:0.5);
            \ThinLine (0:0) -- (-30:0.5);
            \draw[fill=black] (0:0) ellipse (0.05 and 0.05);
            \node at (0,-0.35) {$\xi_{2l}$};
            \end{scope}
            \end{scope}
            \begin{scope}[shift={(0,-1.5)},yscale=-1]
            \begin{scope}[shift={(0.866,0.5)}]
                \ThinLine[orange] (0:0) -- (90:0.5);
            \ThinLine (0:0) -- (-150:0.5);
            \ThinLine (0:0) -- (-30:0.5);
            \draw[fill=black] (0:0) ellipse (0.05 and 0.05);
            \node at (0,-0.35) {$\eta^*_2$};
            \end{scope}
            \node at (2.6,0.25) {$\cdots$};
            \begin{scope}[shift={(3.466,0)}]
                \ThinLine[orange] (0:0) -- (-90:0.5);
            \ThinLine (0:0) -- (150:0.5);
            \ThinLine (0:0) -- (30:0.5);
            \draw[fill=black] (0:0) ellipse (0.05 and 0.05);
            \node at (0,0.35) {$\eta^*_{2l-1}$};
            \begin{scope}[shift={(0.866,0.5)}]
                \ThinLine[orange] (0:0) -- (90:0.5);
            \ThinLine (0:0) -- (-150:0.5);
            \ThinLine (0:0) -- (-30:0.5);
            \draw[fill=black] (0:0) ellipse (0.05 and 0.05);
            \node at (0,-0.35) {$\eta^*_{2l}$};
            \end{scope}
            \end{scope}
            \end{scope}
            \node [draw, minimum width=0.25cm, minimum height=0.25cm, fill = white] at (3.462,-0.75) {\scriptsize $\mathrm{D}_\rho$};
            \node [draw, minimum width=0.25cm, minimum height=0.25cm, fill = white] at (0.866,1.25) {\scriptsize $\mathrm{D}_\rho$};
            \node [draw, minimum width=0.25cm, minimum height=0.25cm, fill = white] at (4.328,1.25) {\scriptsize $\mathrm{D}_\rho$};
            \begin{scope}[shift={(1.732,0)}]
                \ThinLine[orange] (0:0) -- (-90:0.5);
            \ThinLine (0:0) -- (150:0.5);
            \ThinLine (0:0) -- (30:0.5);
            \draw[fill=black] (0:0) ellipse (0.05 and 0.05);
            \node at (0,0.35) {$\xi_3$};
            \begin{scope}[shift={(0,-1.5)},yscale=-1]
                \ThinLine[orange] (0:0) -- (-90:0.5);
            \ThinLine (0:0) -- (150:0.5);
            \ThinLine (0:0) -- (30:0.5);
            \draw[fill=black] (0:0) ellipse (0.05 and 0.05);
            \node at (0,0.35) {$\eta^*_3$};
            \end{scope}
            \node [draw, minimum width=0.25cm, minimum height=0.25cm, fill = white] at (0,-0.75) {\scriptsize $\mathrm{D}_\rho$};
            \end{scope}
            \end{tikzpicture}}} \right)
    \end{aligned}
\end{equation}
\newpage
\begin{equation}
    = \braket{\eta_1|\xi_1}\braket{\eta_2|\xi_2}\frac{1}{\sqrt{d_{a_3} d_{a_{2l+1}}} }\, \tr_{\mathcal{C}}\left( \vcenter{\hbox{\begin{tikzpicture}
            \node at (2.6,0.25) {$\cdots$};
            \begin{scope}[shift={(3.466,0)}]
                \ThinLine[orange] (0:0) -- (-90:0.5);
            \ThinLine (0:0) -- (150:0.5);
            \ThinLine (0:0) -- (30:0.5);
            \draw[fill=black] (0:0) ellipse (0.05 and 0.05);
            \node at (0,0.425) {$\xi_{2l-1}$};
            \begin{scope}[shift={(0.866,0.5)}]
                \ThinLine[orange] (0:0) -- (90:1.25);
            \ThinLine (0:0) -- (-150:0.5);
            \ThinLine (0:0) -- (-30:0.5);
            \draw[fill=black] (0:0) ellipse (0.05 and 0.05);
            \node at (0,-0.35) {$\xi_{2l}$};
            \end{scope}
            \end{scope}
            \begin{scope}[shift={(0,-1.5)},yscale=-1]
            \node at (2.6,0.25) {$\cdots$};
            \begin{scope}[shift={(3.466,0)}]
                \ThinLine[orange] (0:0) -- (-90:0.5);
            \ThinLine (0:0) -- (150:0.5);
            \ThinLine (0:0) -- (30:0.5);
            \draw[fill=black] (0:0) ellipse (0.05 and 0.05);
            \node at (0,0.35) {$\eta^*_{2l-1}$};
            \begin{scope}[shift={(0.866,0.5)}]
                \ThinLine[orange] (0:0) -- (90:0.5);
            \ThinLine (0:0) -- (-150:0.5);
            \ThinLine (0:0) -- (-30:0.5);
            \draw[fill=black] (0:0) ellipse (0.05 and 0.05);
            \node at (0,-0.35) {$\eta^*_{2l}$};
            \end{scope}
            \end{scope}
            \end{scope}
            \node [draw, minimum width=0.25cm, minimum height=0.25cm, fill = white] at (3.462,-0.75) {\scriptsize $\mathrm{D}_\rho$};
            \node [draw, minimum width=0.25cm, minimum height=0.25cm, fill = white] at (4.328,1.25) {\scriptsize $\mathrm{D}_\rho$};
            \begin{scope}[shift={(1.732,0)}]
                \ThinLine[orange] (0:0) -- (-90:0.5);
            \ThinLine (0:0) -- (150:0.5);
            \ThinLine (0:0) -- (30:0.5);
            \draw[fill=black] (0:0) ellipse (0.05 and 0.05);
            \node at (0,0.35) {$\xi_3$};
            \begin{scope}[shift={(0,-1.5)},yscale=-1]
                \ThinLine[orange] (0:0) -- (-90:0.5);
            \ThinLine (0:0) -- (150:0.5);
            \ThinLine (0:0) -- (30:0.5);
            \draw[fill=black] (0:0) ellipse (0.05 and 0.05);
            \node at (0,0.35) {$\eta^*_3$};
            \end{scope}
            \node [draw, minimum width=0.25cm, minimum height=0.25cm, fill = white] at (0,-0.75) {\scriptsize $\mathrm{D}_\rho$};
            \end{scope}
            \end{tikzpicture}}} \right) = \cdots = \prod^{2l}_{i=1} \braket{\eta_i|\xi_i}. 
\end{equation}
By linearity of $\mathrm{eval}$, this proves the claim. 
\end{proof}

Following the same procedure, we define the evaluation map $\mathrm{eval}_{\Lambda}$ for a region $\Lambda$ in one of the following forms:
\begin{equation}\label{eqn:: 2nd form of Lambda}
    \vcenter{\hbox{\begin{tikzpicture}
        \begin{scope}[yscale=-1]
            \ThinLine[orange] (0:0) -- (-90:0.5);
    \ThinLine (0:0) -- (150:0.5);
    \ThinLine (0:0) -- (30:0.5);
    \draw[fill=black] (0:0) ellipse (0.05 and 0.05);
    \node at (0,0.3) {$v_1$};
    \begin{scope}[shift={(0.866,0.5)}]
        \ThinLine[orange] (0:0) -- (90:0.5);
    \ThinLine (0:0) -- (-150:0.5);
    \ThinLine (0:0) -- (-30:0.5);
    \draw[fill=black] (0:0) ellipse (0.05 and 0.05);
    \node at (0,-0.3) {$v_2$};
    \end{scope}
    \node at (1.73,0.25) {$\cdots$};
    \begin{scope}[shift={(2.596,0)}]
        \ThinLine[orange] (0:0) -- (-90:0.5);
    \ThinLine (0:0) -- (150:0.5);
    \ThinLine (0:0) -- (30:0.5);
    \draw[fill=black] (0:0) ellipse (0.05 and 0.05);
    \node at (0,0.4) {$v_{2l-1}$};
    \begin{scope}[shift={(0.866,0.5)}]
        \ThinLine[orange] (0:0) -- (90:0.5);
    \ThinLine (0:0) -- (-150:0.5);
    \ThinLine (0:0) -- (-30:0.5);
    \draw[fill=black] (0:0) ellipse (0.05 and 0.05);
    \node at (0,-0.3) {$v_{2l}$};
    \end{scope}
    \end{scope}
    \end{scope}
    \begin{scope}[overlay]
            \fill[
                white,
                even odd rule,
                rounded corners=5pt
            ]
                (-0.56,-1.06) rectangle (3.96,0.56)
                (-0.2,-0.85) --
                ( 3.88,-0.85) --
                ( 3.58, 0.35) --
                (-0.5, 0.35) -- cycle;
            \draw[
                densely dashed,
                rounded corners=5pt,
                line width=0.5pt
            ]
                (-0.2,-0.85) --
                ( 3.88,-0.85) --
                ( 3.58, 0.35) --
                (-0.5, 0.35) -- cycle;
        \end{scope}
    \end{tikzpicture}}}\quad 
    \vcenter{\hbox{\begin{tikzpicture}
        \begin{scope}[yscale=-1]
    \begin{scope}[shift={(0.866,0.5)}]
        \ThinLine[orange] (0:0) -- (90:0.5);
    \ThinLine (0:0) -- (-150:0.5);
    \ThinLine (0:0) -- (-30:0.5);
    \draw[fill=black] (0:0) ellipse (0.05 and 0.05);
    \end{scope}
    \node at (1.73,0.25) {$\cdots$};
    \begin{scope}[shift={(2.596,0)}]
        \ThinLine[orange] (0:0) -- (-90:0.5);
    \ThinLine (0:0) -- (150:0.5);
    \ThinLine (0:0) -- (30:0.5);
    \draw[fill=black] (0:0) ellipse (0.05 and 0.05);
    \node at (0,0.4) {$v_{2l-2}$};
    \begin{scope}[shift={(0.866,0.5)}]
        \ThinLine[orange] (0:0) -- (90:0.5);
    \ThinLine (0:0) -- (-150:0.5);
    \ThinLine (0:0) -- (-30:0.5);
    \draw[fill=black] (0:0) ellipse (0.05 and 0.05);
    \end{scope}
    \end{scope}
    \end{scope}
    \begin{scope}[overlay]
    % Mask the portions of the edges outside the trapezoid
    \fill[
        white,
        even odd rule,
        rounded corners=5pt
    ]
        (0.14,-1.06) rectangle (4.19,0.56)
        (0.35,-0.85) --
        (3.98,-0.85) --
        (3.287,0.35) --
        (1.043,0.35) -- cycle;
    % Rounded, densely dashed trapezoidal boundary
    \draw[
        densely dashed,
        rounded corners=5pt,
        line width=0.5pt
    ]
        (0.35,-0.85) --
        (3.98,-0.85) --
        (3.287,0.35) --
        (1.043,0.35) -- cycle;
    % Restore labels covered by the mask
    \node at (0.866,-0.2) {$v_1$};
    \node at (3.462,-0.2) {$v_{2l-1}$};
\end{scope}
    \end{tikzpicture}}} \quad 
    \vcenter{\hbox{\begin{tikzpicture}
        \begin{scope}[yscale=-1]
            \ThinLine[orange] (0:0) -- (-90:0.5);
    \ThinLine (0:0) -- (150:0.5);
    \ThinLine (0:0) -- (30:0.5);
    \draw[fill=black] (0:0) ellipse (0.05 and 0.05);
    \begin{scope}[shift={(0.866,0.5)}]
        \ThinLine[orange] (0:0) -- (90:0.5);
    \ThinLine (0:0) -- (-150:0.5);
    \ThinLine (0:0) -- (-30:0.5);
    \draw[fill=black] (0:0) ellipse (0.05 and 0.05);
    \node at (0,-0.3) {$v_2$};
    \end{scope}
    \node at (1.73,0.25) {$\cdots$};
    \begin{scope}[shift={(2.596,0)}]
        \ThinLine[orange] (0:0) -- (-90:0.5);
    \ThinLine (0:0) -- (150:0.5);
    \ThinLine (0:0) -- (30:0.5);
    \draw[fill=black] (0:0) ellipse (0.05 and 0.05);
    \end{scope}
    \end{scope}
    \begin{scope}[overlay]
    % Tight mask outside the height-adjusted trapezoid
    \fill[
        white,
        even odd rule,
        rounded corners=5pt
    ]
        (-0.50,-1.10) rectangle (3.10,0.56)
        ( 0.317,-1.05) --
        ( 2.279,-1.05) --
        ( 2.972, 0.15) --
        (-0.376, 0.15) -- cycle;
    % Rounded, densely dashed trapezoidal boundary
    \draw[
        densely dashed,
        rounded corners=5pt,
        line width=0.5pt
    ]
        ( 0.317,-1.05) --
        ( 2.279,-1.05) --
        ( 2.972, 0.15) --
        (-0.376, 0.15) -- cycle;
    % Restore the label covered by the mask
    \node at (0,-0.3) {$v_1$};
    \node at (0.866,-0.2) {$v_2$};
    \node at (2.596,-0.4) {$v_{2l-1}$};
\end{scope}
% Make the bounding box follow the visible trapezoid
\pgfresetboundingbox
\path[use as bounding box]
    (-0.50,-1.10) rectangle (3.10,0.20);
    \end{tikzpicture}}}
\end{equation}
so that the corresponding string-net subspace is identified with $\Hom(\rho^l \gamma,\gamma\rho^l)$, $\Hom(\gamma \rho^{l-1}\gamma,\rho^l)$, and $\Hom(\rho^l,\gamma\rho^{l-1}\gamma)$, respectively.
The proof of Lemma \ref{lemma:: eval in a single row} also shows that $\mathrm{eval}_{\Lambda}$ is a surjective partial isometry. 

The identifications of string-net subspaces in a row with morphism spaces in $\mathcal{C}$ allow us to lift operators in $\End(\rho^l)$ to local operators. 
To do so, let $\pi_L$ and $\pi_R$ denote the left and right actions of $\End(\rho^{l})$ on $\Hom(\gamma \rho^l,\rho^l\gamma)$ defined by
\begin{equation}
    \pi_L(f)\pi_R(g) h = (f\otimes \mathrm{Id}_\gamma) h(\mathrm{Id}_\gamma \otimes g),\quad h\in \Hom(\gamma \rho^l,\rho^l\gamma).
\end{equation}
These formulas define representations $\pi_L$ and $\pi_R$ of $\End(\rho^l)$ and $\End(\rho^l)^{\text{op}}$, respectively. 
They are not $*$-representations, since
\begin{equation}
    \begin{aligned}
        \pi_L(f)^\dagger &= \pi_L\left( \left( \mathrm{D}_{\rho^l} \right)^{-1} f^* \mathrm{D}_{\rho^l}\right)\\
        \pi_R(f)^\dagger &= \pi_R\left(  \mathrm{D}_{\rho^l}  f^* \left( \mathrm{D}_{\rho^l} \right)^{-1} \right)
    \end{aligned}
\end{equation}
These in turn define the following local operators in a region $\Lambda$ as in \eqref{fig:: 1st form of Lambda}:
\begin{equation}\label{eqn:: bimodule structure of Tree}
    \begin{aligned}
        \mathrm{L}_{\Lambda}(f) &= \mathrm{eval}_{\Lambda}^\dagger \circ \pi_L(f) \circ \mathrm{eval}_{\Lambda},\\
    \mathrm{R}_{\Lambda}(f) &= \mathrm{eval}_{\Lambda}^\dagger \circ \pi_R(f) \circ \mathrm{eval}_{\Lambda}.
    \end{aligned}
\end{equation}
where $\mathrm{eval}_{\Lambda}^\dagger$ is the adjoint with respect to the inner products defined in \eqref{eqn:: inner product on Hom}. 
By Lemma \ref{lemma:: eval in a single row}, $\mathrm{eval}_{\Lambda}\mathrm{eval}^\dagger_{\Lambda}$ is the identity map on $\Hom(\gamma\rho^l,\rho^l\gamma)$. 
So the maps $\mathrm{L}_{\Lambda}(\cdot )$ and $\mathrm{R}_{\Lambda}(\cdot )$ preserve/reverse the multiplication on $\End(\rho^l)$. 
For the three region types in \eqref{eqn:: 2nd form of Lambda}, we similarly define $\mathrm{L}_{\Lambda}$/$\mathrm{R}_{\Lambda}$ to be the left/right actions of $\End(\rho^l)$, $\End(\rho^l)$/$\End(\rho^{l-1})$, and $\End(\rho^{l-1})$/$\End(\rho^l)$ respectively. 
Note that the operators $\mathrm{L}_{\Lambda}(f)$ and $\mathrm{R}_{\Lambda}(f)$ act by $0$ on the complement of the string-net subspace. 

\begin{definition}\label{def:: generating set}
    Let $\mathcal{C}$ be a unitary fusion category and $\rho\in \mathcal{C}$ be a self-dual object. 
    For $k>1$, we call a finite linearly independent set of morphisms $\mathcal{G}\subset \End(\rho^k)$ a weight-$k$ generating set if the following conditions are satisfied:
    \begin{enumerate}
        \item for any $f\in \mathcal{G}$, $f^* \in \mathcal{G}$; 
        \item for any $f\in \mathcal{G}$, there are simple objects $a_1,\dots, a_k$ and $b_1,\dots ,b_k$ in $\Irr(\mathcal{C})$ such that 
        \begin{equation}
            f = (q_{b_1}\otimes \cdots\otimes q_{b_k}) f (q_{a_1}\otimes \cdots\otimes q_{a_k});
        \end{equation}
        \item for any $n\geq k$, the $C^*$-algebra $\End(\rho^n)$ is generated by the set
        \begin{equation}
            \left\{\mathrm{Id}^{l}_{\rho}\otimes f\otimes \mathrm{Id}^{n-l-k}_{\rho} | f\in \mathcal{G}, 0\leq l\leq n-k \right\}. 
        \end{equation}
    \end{enumerate}
\end{definition}

Fixing a generating set $\mathcal{G}$, we define the second type of interaction as follows. 
Let $S$ a set of $(k-1)$ consecutive plaquettes between two rows $\Lambda,\Lambda'$. 
We define the operator $F_S$ by
\begin{equation}\label{eqn:: F_S}
    F_S = \sum_{f\in \mathcal{G}} \left\vert \mathrm{R}_{\Lambda} \left( \mathrm{D}_{\rho^k} f \left( \mathrm{D}_{\rho^k} \right)^{-1} \right) -  \mathrm{L}_{\Lambda'}\left( \left( \mathrm{D}_{\rho^k} \right)^{-1} f \mathrm{D}_{\rho^k} \right) \right\vert^2, 
\end{equation}
where $\vert A\vert^2 = A^\dagger A$. 
The following figure illustrate the support of the operator $F_S$ made from a generating set of weight $3$. 
\begin{figure}[hb]
    \centering
    \definecolor{bondorange}{RGB}{245,177,67}
\definecolor{bondblack}{RGB}{25,25,25}
\definecolor{contourmagenta}{RGB}{205,35,135}

\begin{tikzpicture}[
    x=1cm,
    y=1cm,
    line cap=round,
    line join=round,
    orange bond/.style={
        draw=bondorange,
        line width=0.95pt
    },
    black bond/.style={
        draw=bondblack,
        line width=0.70pt
    },
    contour/.style={
        draw=contourmagenta,
        line width=1.25pt,
        % dash pattern=on 5pt off 3pt,
        rounded corners=12pt
    }
]

    % Landscape crop, similar to the supplied picture.
    % For the original square crop, replace both rectangles
    % below by (0,0) rectangle (8,8).
    \path[use as bounding box]
        (0,1.2) rectangle (8,6.5);
    \clip
        (0,1.2) rectangle (8,6.5);

    \begin{scope}[xshift=0.25cm]

        \def\dx{0.73}
        \def\dy{0.44}

        % -----------------------------------------------
        % Orange vertical bonds
        % Each shared vertical bond is drawn only once.
        % -----------------------------------------------
        \foreach \cy/\xoffset in {
            -0.10/0,
             1.22/0.73,
             2.54/0,
             3.86/0.73,
             5.18/0,
             6.50/0.73,
             7.82/0,
             9.14/0.73
        }{
            \foreach \xb in {
                -2.92,-1.46,0,1.46,2.92,4.38,
                5.84,7.30,8.76,10.22
            }{
                \pgfmathsetmacro{\cx}{\xb+\xoffset}

                \draw[orange bond]
                    ({\cx+\dx},{\cy-\dy})
                    --
                    ({\cx+\dx},{\cy+\dy});
            }
        }

        % -----------------------------------------------
        % Black diagonal bonds
        % -----------------------------------------------
        \foreach \cy/\xoffset in {
            -0.10/0,
             1.22/0.73,
             2.54/0,
             3.86/0.73,
             5.18/0,
             6.50/0.73,
             7.82/0,
             9.14/0.73
        }{
            \foreach \xb in {
                -2.92,-1.46,0,1.46,2.92,4.38,
                5.84,7.30,8.76,10.22
            }{
                \pgfmathsetmacro{\cx}{\xb+\xoffset}

                \draw[black bond]
                    ({\cx-\dx},{\cy+\dy})
                    --
                    (\cx,{\cy+2*\dy})
                    --
                    ({\cx+\dx},{\cy+\dy});

                \draw[black bond]
                    ({\cx-\dx},{\cy-\dy})
                    --
                    (\cx,{\cy-2*\dy})
                    --
                    ({\cx+\dx},{\cy-\dy});
            }
        }

        % -----------------------------------------------
        % Black vertices
        % These are the endpoints of the vertical bonds.
        % -----------------------------------------------
        \foreach \cy/\xoffset in {
            -0.10/0,
             1.22/0.73,
             2.54/0,
             3.86/0.73,
             5.18/0,
             6.50/0.73,
             7.82/0,
             9.14/0.73
        }{
            \foreach \xb in {
                -2.92,-1.46,0,1.46,2.92,4.38,
                5.84,7.30,8.76,10.22
            }{
                \pgfmathsetmacro{\cx}{\xb+\xoffset}

                \fill[bondblack]
                    ({\cx+\dx},{\cy-\dy})
                    circle[radius=2pt];

                \fill[bondblack]
                    ({\cx+\dx},{\cy+\dy})
                    circle[radius=2pt];
            }
        }

        % -----------------------------------------------
        % Rounded contour
        %
        % All six control vertices are hexagon centers.
        % The top and bottom sides each cross two orange
        % bonds. Each sloping side crosses one black bond.
        %
        % Every crossing is at the exact bond midpoint.
        % -----------------------------------------------
        \coordinate (Cleft)        at (2.19,3.86);
        \coordinate (Cupperleft)   at (2.92,5.18);
        \coordinate (Cupperright)  at (5.84,5.18);
        \coordinate (Cright)       at (6.57,3.86);
        \coordinate (Clowerright)  at (5.84,2.54);
        \coordinate (Clowerleft)   at (2.92,2.54);

        % Start on a straight section so that the closing
        % point does not interfere with corner rounding.
        \draw[contour]
            (4.38,5.18)
            -- (Cupperright)
            -- (Cright)
            -- (Clowerright)
            -- (Clowerleft)
            -- (Cleft)
            -- (Cupperleft)
            -- cycle;

         % -----------------------------------------------
        % Region and row labels
        % -----------------------------------------------

        % Label the entire enclosed region near its upper boundary.
        \node[
            anchor=north,
            font=\large,
            inner sep=2pt
        ] at (4.38,5.10) {$S$};

         % Upper-row label: shifted slightly upward.
        \node[
            anchor=west,
            font=\large,
            inner sep=1pt
        ] at (6.85,4.78) {$\Lambda'$};

        % Lower-row label: shifted slightly downward.
        \node[
            anchor=west,
            font=\large,
            inner sep=1pt
        ] at (6.85,2.94) {$\Lambda$};
    \end{scope}
\end{tikzpicture}
\caption{The support of the operator $F_S$ defined from a set of generating set of weight $k=3$, which occupies two consecutive plaquettes in a row.}
\end{figure}
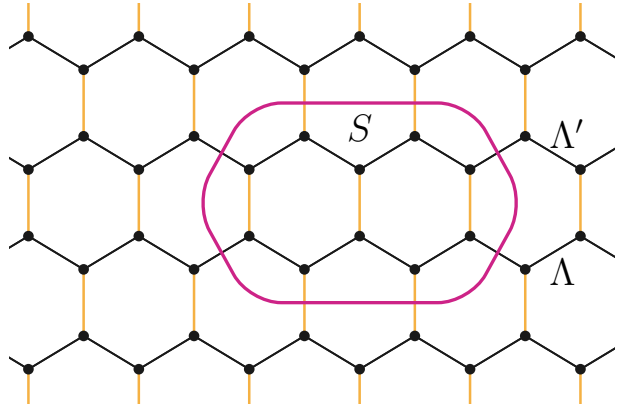
In particular, when $k=2$, the operator $F_S$ is supported in a single plaquette, so a genuine plaquette operator. 
For a region $\Lambda$, denote by $\mathcal{P}_{k-1}(\Lambda)$ the collection of $k-1$ consecutive plaquettes in its interior. 
Then the interaction of the model is given by
\begin{equation}\label{eqn: interaction of new model}
    \varPhi(D) =\begin{cases}
        \mathrm{I} - Q_e, & \text{if $D = e$}\\
        F_S, & \text{if $D \in \mathcal{P}_{k-1}(\Lambda)$}\\
        0, & \text{otherwise}
    \end{cases}
\end{equation}

By construction, the local Hamiltonian is a sum of positive local operators. 
Note that the edge projections $Q_e$ associated with uncolored edges commute with each other and with $F_S$. 
However, $F_S$ and $F_{S'}$ do not commute in general. 

\begin{remark}
    The choice of the generating set $\mathcal{G}$ is not unique, and different choices give rise to different Hamiltonians. 
    However, as we will see, local Hamiltonians defined by different sets of generators all share the same local ground state subspace for sufficiently large regions. 
\end{remark}

\section{Solving Local Ground States}\label{Sec:: Solving Local Ground States}

In this section, we solve the local ground states of the Hamiltonian defined in \eqref{eqn: interaction of new model}. 
We fix throughout the unitary fusion category $\mathcal{C}$ and the chosen object $\rho$. 
We show that the Hamiltonian is frustration-free, and the local ground state subspace can mapped to certain morphism spaces in $\mathcal{C}$. 

For a finite region $\Lambda\subset \Gamma$, we denote by $\mathfrak{A}(\Lambda) = \bigotimes_{i\in \Lambda}\mathcal{B}(\mathcal{H}_i)$ the algebra of local operators in $\Lambda$. 
The local Hamiltonian $H_{\Lambda}\in \mathfrak{A}(\Lambda)$ defined by the interaction as in \eqref{eqn: interaction of new model} is
\begin{equation}
    H_{\Lambda} = \sum_{D\subset \Lambda}\varPhi(D) = \sum_{e\in E(\Lambda)} (\mathrm{I} - Q_e) + \sum_{S\in \mathcal{P}_{k-1}(\Lambda)} F_S 
\end{equation}
the local Hamiltonian associated to $\Lambda$.
Denote by $\Pi(\Lambda)\in \mathfrak{A}(\Lambda)$ the ground state projection of $H_{\Lambda}$. 

We choose to solve ground states for a parameterized family of regions described below. 
Let $s$ be a vertex in the dual graph of $\Gamma$. 
For $l,h\geq 1$, a region $\Lambda$ in $\mathbb{R}^2$ is said to have shape $(l,h,s)$, if $\Lambda$ is bounded by two line segments emanating from $s$, such that the first goes through $l$ colored edges to the right and the second goes through $h$ uncolored edges to the upper right (Fig. \ref{fig:: based parallelogram}). 
The parallelogram region with shape $(l,h,s)$ contains $2lh$ vertices in its interior, and its boundary intersects with $\Gamma$ at $2l$ colored edges and $2h$ uncolored edges. 
We write $(l,h)$ when the choice of the base point is irrelevant. 
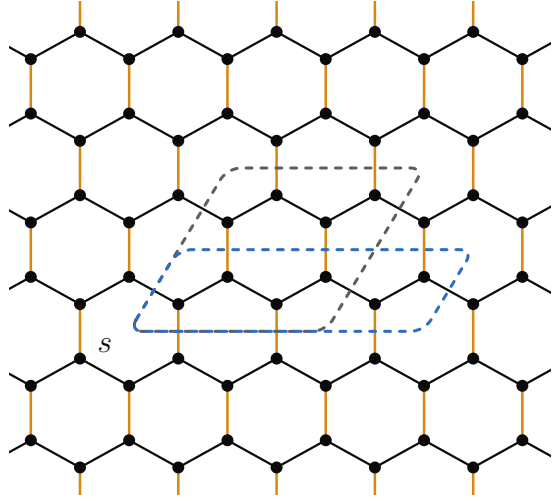
\begin{figure}[htbp]
    \centering
    % Darker versions of the colors in the original figure
    \definecolor{edgeorange}{RGB}{215,135,15}
    \definecolor{regionblue}{RGB}{45,110,195}
    \definecolor{regiongray}{RGB}{95,95,95}
    \definecolor{latticeblack}{RGB}{10,10,10}
    \begin{tikzpicture}[
        x=0.65cm,
        y=-0.72cm, % negative so coordinates increase downward
        line cap=round,
        line join=round
    ]

        % Fixed visible region
        \path[use as bounding box] (0,-0.55) rectangle (11,8.50);
        \clip (0,-0.55) rectangle (11,8.50);

        % Shift the complete picture to the right. The clipping window
        % remains fixed, exposing the leftmost vertical edges fully.
        \begin{scope}[xshift=0.28cm]

            % ---------------------------------------------------------
            % Orange vertical edges
            % ---------------------------------------------------------
            % Each honeycomb row has vertical period 3.
            \foreach \r in {-1,0,1,2,3} {
                \pgfmathsetmacro{\yy}{3*\r}

                % Vertical edges at even x-coordinates
                \foreach \x in {-2,0,2,4,6,8,10,12} {
                    \draw[
                        edgeorange,
                        line width=0.95pt
                    ]
                    (\x,{\yy+0.5}) -- (\x,{\yy+1.5});
                }

                % Vertical edges at odd x-coordinates
                \foreach \x in {-1,1,3,5,7,9,11,13} {
                    \draw[
                        edgeorange,
                        line width=0.95pt
                    ]
                    (\x,{\yy+2}) -- (\x,{\yy+3});
                }
            }

            % ---------------------------------------------------------
            % Black diagonal edges
            % ---------------------------------------------------------
            \foreach \r in {-1,0,1,2,3} {
                \pgfmathsetmacro{\yy}{3*\r}

                \foreach \x in {-3,-1,1,3,5,7,9,11} {

                    % Upper pair of sloping edges
                    \draw[
                        latticeblack,
                        line width=0.85pt
                    ]
                    (\x,\yy)
                    --
                    ({\x+1},{\yy+0.5})
                    --
                    ({\x+2},\yy);

                    % Lower pair of sloping edges
                    \draw[
                        latticeblack,
                        line width=0.85pt
                    ]
                    (\x,{\yy+2})
                    --
                    ({\x+1},{\yy+1.5})
                    --
                    ({\x+2},{\yy+2});
                }
            }

            % ---------------------------------------------------------
            % Black vertices
            % ---------------------------------------------------------
            \foreach \r in {-1,0,1,2,3} {
                \pgfmathsetmacro{\yy}{3*\r}

                % Vertices at odd x-coordinates
                \foreach \x in {-3,-1,1,3,5,7,9,11,13} {
                    \fill[latticeblack]
                        (\x,\yy) circle[radius=2.25pt];

                    \fill[latticeblack]
                        (\x,{\yy+2}) circle[radius=2.25pt];
                }

                % Vertices at even x-coordinates
                \foreach \x in {-2,0,2,4,6,8,10,12} {
                    \fill[latticeblack]
                        (\x,{\yy+0.5}) circle[radius=2.25pt];

                    \fill[latticeblack]
                        (\x,{\yy+1.5}) circle[radius=2.25pt];
                }
            }

            % ---------------------------------------------------------
            % Dashed parallelograms
            % ---------------------------------------------------------

            % Gray, more steeply tilted parallelogram
            \draw[
                regiongray,
                line width=1.15pt,
                dash pattern=on 2.2pt off 3.4pt,
                rounded corners=5pt
            ]
                (2,5.50)
                --
                (4,2.50)
                --
                (8,2.50)
                --
                (6,5.50)
                --
                cycle;

            % Blue, more horizontal parallelogram
            \draw[
                regionblue,
                line width=1.15pt,
                dash pattern=on 2.2pt off 3.4pt,
                rounded corners=5pt
            ]
                (2,5.50)
                --
                (3,4.00)
                --
                (9,4.00)
                --
                (8,5.50)
                --
                cycle;

            % Region label
            \node[
                anchor=east,
                inner sep=1pt,
                text=latticeblack
            ] at (1.72,5.75) {$s$};

        \end{scope}
    \end{tikzpicture}
    \caption{The parallelogram region in dashed gray has shape $(2,2,s)$, and the one in dashed blue has shape $(3,1,s)$.}
    \label{fig:: based parallelogram}
\end{figure}

The building block of parallelograms are those of shape $(l,1)$, which contain vertices in a single row. 
Since the local Hamiltonian $H_{\Lambda}$ contains no plaquette terms, the local ground state subspace of $H_{\Lambda}$ is just $\mathcal{H}^{\mathrm{st.n.}}_{\Lambda}$. 
By Lemma \ref{lemma:: eval in a single row}, we thus obtain that for such $\Lambda$: 
\begin{equation}\label{eqn:: ground state projection in a row}
    \Pi(\Lambda) = \prod_{e\subset \Lambda} Q_e = \mathrm{eval}^\dagger_{\Lambda}\mathrm{eval}_{\Lambda}.
\end{equation}
For a configuration $\ket{\xi}$ on some $\Lambda$ that has shape $(l,1)$, we write $\ket{\xi} = \ket{\xi_1,\xi_2,,\dots \xi_{2l}}$ where $\xi_i\in \mathcal{H}_{v_i}$ and $v_1,v_2,\dots v_{2l}$ are the vertices in $\Lambda$ numbered from left to right. 

\begin{lemma}\label{lemma:: propagate actions to bigger regions}
    Let $\Lambda$ be a region of shape $(l,1)$, and number the vertices in $\Lambda$ from left to right by $v_1,v_2,\dots v_{2l}$. 
    Let $\Lambda'\subseteq \Lambda$ be a region of shape $(k,1)$ with $k\leq l$, and suppose $\Lambda'$ contains the vertices $v_{2i+1},\dots, v_{2i+2k}$. 
    Then for any $f\in \End(\rho^k)$, we have 
    \begin{equation}\label{eqn:: propagate left actions}
        \mathrm{eval}_{\Lambda} \mathrm{L}_{\Lambda'}(f) = \pi_L(\mathrm{Id}^i_\rho\otimes f\otimes \mathrm{Id}^{l-i-k}_\rho)\mathrm{eval}_{\Lambda}.
    \end{equation}
    Similarly, we have 
    \begin{equation}\label{eqn:: propagate right actions}
        \mathrm{eval}_{\Lambda} \mathrm{R}_{\Lambda'}(f) = \pi_R(\mathrm{Id}^i_\rho\otimes f\otimes \mathrm{Id}^{l-i-k}_\rho)\mathrm{eval}_{\Lambda}.
    \end{equation}
\end{lemma}
\begin{proof}
    By linearity, it suffices to consider a configuration $\ket{\xi}\in \mathcal{H}_{\Lambda}$. 
    Partition $\Lambda$ into $\Lambda_-\cup \Lambda'\cup \Lambda_+$ from left to right, and write $\ket{\xi}=\ket{\xi^-}\otimes\ket{\xi'}\otimes\ket{\xi^+}$, where $\ket{\xi'}=\ket{\xi_{2i+1},\dots,\xi_{2i+2k}}$. 
    By \eqref{eqn:: bimodule structure of Tree}, 
    \begin{equation}\label{eqn:: expanding the LHS}
        \begin{aligned}
            \mathrm{eval}_{\Lambda} \left( \mathrm{L}_{\Lambda'}(f)\otimes \mathrm{I}_{\Lambda\backslash\Lambda'} \right)\ket{\xi} = \mathrm{eval}_{\Lambda}\left( \xi^{-}\otimes \mathrm{eval}^\dagger_{\Lambda'}\left( \pi_{L}(f)\mathrm{eval}_{\Lambda'}(\xi') \right) \otimes \xi^{+} \right).
        \end{aligned}
    \end{equation}
    For any $h\in \Hom(a\rho^k,\rho^k d)$, we have 
    \begin{equation}
        \begin{aligned}
            &\mathrm{eval}_{\Lambda}\left( \xi^{-}\otimes \mathrm{eval}^\dagger_{\Lambda'}(h) \otimes \xi^{+} \right)\\
            &= \left[ \mathrm{Id}^{\otimes i+k}_{\rho}\otimes \mathrm{eval}_{\Lambda_+}(\xi^{+}) \right] \left[ \mathrm{Id}^{\otimes i}_{\rho} \otimes h\otimes \mathrm{Id}^{\otimes l-k-i}_{\rho} \right] \left[ \mathrm{eval}_{\Lambda_-}(\xi^{-}) \otimes \mathrm{Id}^{\otimes l-i}_{\rho} \right].
        \end{aligned}
    \end{equation}
    Write $\widetilde f=\mathrm{Id}^{i}_{\rho}\otimes f\otimes \mathrm{Id}^{l-i-k}_{\rho}$.
    Substituting
    $h=\pi_L(f)\mathrm{eval}_{\Lambda'}(\xi')=(f\otimes\mathrm{Id}_d)\mathrm{eval}_{\Lambda'}(\xi')$
    into the preceding expansion gives
    \begin{equation}
        \begin{aligned}
            &\mathrm{eval}_{\Lambda}\left( \xi^{-}\otimes \mathrm{eval}^\dagger_{\Lambda'}\left( \pi_{L}(f)\mathrm{eval}_{\Lambda'}(\xi') \right) \otimes \xi^{+} \right) = \left(\widetilde f\otimes\mathrm{Id}_{\gamma}\right)\mathrm{eval}_{\Lambda}(\xi) = \pi_L\left(\widetilde f\right)\mathrm{eval}_{\Lambda}(\xi).
        \end{aligned}
    \end{equation}
    Together with \eqref{eqn:: expanding the LHS}, this proves \eqref{eqn:: propagate left actions}. 
    Similarly, using $\pi_R(f)\mathrm{eval}_{\Lambda'}(\xi') = \mathrm{eval}_{\Lambda'}\left( \xi'(\mathrm{Id}_a\otimes f) \right)$ in the same expansion gives
    \begin{equation}
        \begin{aligned}
            &\mathrm{eval}_{\Lambda}\left(\mathrm{R}_{\Lambda'}(f)\otimes\mathrm{I}_{\Lambda\backslash\Lambda'}\right)\ket{\xi} = \mathrm{eval}_{\Lambda}(\xi)\left(\mathrm{Id}_{\gamma}\otimes\widetilde f\right) = \pi_R\left(\widetilde f\right)\mathrm{eval}_{\Lambda}(\xi).
        \end{aligned}
    \end{equation}
    This proves \eqref{eqn:: propagate right actions}. 
\end{proof}

\begin{corollary}\label{corollary:: commutativity between left/right actions}
    Let $\Lambda$ be a region of shape $(l,1)$. 
    Suppose $\Lambda = \Lambda_1\cup \Lambda_2$ where $\Lambda_i$ is a region of shape $(k_i,1)$ for $i=1,2$. 
    Then for any $f\in \End(\rho^{k_1})$ and $g\in \End(\rho^{k_2})$, 
    \begin{equation}
        [\mathrm{L}_{\Lambda_1}(f), \mathrm{R}_{\Lambda_2}(g)] = 0. 
    \end{equation}
\end{corollary}
\begin{proof}
    Set $A=\mathrm{L}_{\Lambda_1}(f)\otimes\mathrm{I}_{\Lambda\backslash\Lambda_1}$ and $B=\mathrm{R}_{\Lambda_2}(g)\otimes\mathrm{I}_{\Lambda\backslash\Lambda_2}$. 
    If $\Lambda_1$ and $\Lambda_2$ have disjoint vertex sets, then $A$ and $B$ have disjoint supports and hence commute.
    We may therefore assume that $\Lambda_1$ and $\Lambda_2$ overlap.
    By \eqref{eqn:: bimodule structure of Tree}, each local operator annihilates a configuration with mismatched labels on an edge contained in the interior of the corresponding subrow and preserves the labels on the edges crossing its boundary.
    In particular, $A$ and $B$ preserve the string-net subspace of $\mathcal{H}_{\Lambda}$. 
    By Lemma \ref{lemma:: propagate actions to bigger regions}, 
    \begin{equation}
        \begin{aligned}
            \mathrm{eval}_{\Lambda}A
            &=\pi_L(\widetilde f)\mathrm{eval}_{\Lambda},\\
            \mathrm{eval}_{\Lambda}B
            &=\pi_R(\widetilde g)\mathrm{eval}_{\Lambda},
        \end{aligned}
    \end{equation}
    where $\widetilde f$ and $\widetilde g$ are obtained by tensoring $f$ and $g$ with identity morphisms.
    If $\ket{\xi}$ belongs to the string-net subspace, then
    \begin{equation}
        \begin{aligned}
            \mathrm{eval}_{\Lambda}AB\ket{\xi}
            &=\pi_L(\widetilde f)\pi_R(\widetilde g)\mathrm{eval}_{\Lambda}\ket{\xi}\\
            &=\pi_R(\widetilde g)\pi_L(\widetilde f)\mathrm{eval}_{\Lambda}\ket{\xi}
            =\mathrm{eval}_{\Lambda}BA\ket{\xi},
        \end{aligned}
    \end{equation}
    because the left and right actions commute.
    The restriction of $\mathrm{eval}_{\Lambda}$ to the string-net subspace is injective by Lemma \ref{lemma:: eval in a single row}; hence $AB\ket{\xi}=BA\ket{\xi}$.

    It remains to consider a configuration $\ket{\xi}$ with mismatched labels on some edge $e\in E(\Lambda)$.
    Since the two overlapping subrows cover $\Lambda$, we have $e\in E(\Lambda_1)\cup E(\Lambda_2)$.
    Suppose first that $e\in E(\Lambda_1)$.
    Then $A\ket{\xi}=0$, so $BA\ket{\xi}=0$.
    If $e\in E(\Lambda_2)$, then $B\ket{\xi}=0$; otherwise, $B$ preserves the mismatched labels on $e$, hence $AB\ket{\xi}=0$.
    If $e\notin E(\Lambda_1)$, then $e\in E(\Lambda_2)$, and the same argument with $A$ and $B$ interchanged shows that both compositions vanish.
    This proves the claim. 
\end{proof}

\subsection{Local ground state subspace}

We now solve the local ground states of $H_{\Lambda}$ for a region $\Lambda$ of shape $(l,h)$ with $h>1$. 
Consider a partition of $\Lambda$ into $h$ parallelograms of shape $(l,1)$, denoted by $\Lambda_1,\Lambda_2,\dots \Lambda_h$ from top to bottom. 
We shall see that local ground states in a parallelogram $\Lambda$ of shape $(l,h)$ can be obtained by sewing those in regions of shape $(l,1)$, which are known to correspond to morphisms in $\Hom(\gamma\rho^l,\rho^l\gamma)$. 

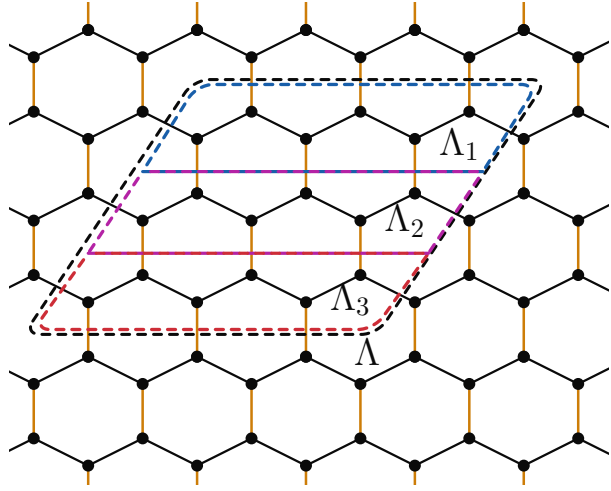
\begin{figure}[htbp]
    \centering
    % Darker colors
\definecolor{edgeorange}{RGB}{205,125,8}
\definecolor{boundaryblue}{RGB}{25,95,170}
\definecolor{boundarymagenta}{RGB}{180,35,165}
\definecolor{boundaryred}{RGB}{205,45,55}
\definecolor{latticeblack}{RGB}{10,10,10}

\begin{tikzpicture}[
    x=0.72cm,
    y=-0.72cm, % Coordinates increase downward
    line cap=round,
    line join=round
]

    % Fixed visible region
    \path[use as bounding box]
        (0,-0.50) rectangle (11,8.35);
    \clip
        (0,-0.50) rectangle (11,8.35);

    % Shift the complete drawing slightly to the right
    \begin{scope}[xshift=0.32cm]

        % =========================================================
        % Honeycomb lattice
        % =========================================================

        % Orange vertical bonds
        \foreach \r in {-1,0,1,2,3} {
            \pgfmathsetmacro{\yb}{3*\r}

            % Vertical bonds at even x-coordinates
            \foreach \x in {-2,0,2,4,6,8,10,12} {
                \draw[
                    edgeorange,
                    line width=0.95pt
                ]
                    (\x,{\yb+0.5})
                    --
                    (\x,{\yb+1.5});
            }

            % Vertical bonds at odd x-coordinates
            \foreach \x in {-3,-1,1,3,5,7,9,11,13} {
                \draw[
                    edgeorange,
                    line width=0.95pt
                ]
                    (\x,{\yb+2})
                    --
                    (\x,{\yb+3});
            }
        }

        % Black diagonal bonds
        \foreach \r in {-1,0,1,2,3} {
            \pgfmathsetmacro{\yb}{3*\r}

            % Diagonal bonds from the upper odd vertices
            \foreach \x in {-3,-1,1,3,5,7,9,11,13} {
                \draw[
                    latticeblack,
                    line width=0.85pt
                ]
                    (\x,\yb)
                    --
                    ({\x-1},{\yb+0.5});

                \draw[
                    latticeblack,
                    line width=0.85pt
                ]
                    (\x,\yb)
                    --
                    ({\x+1},{\yb+0.5});
            }

            % Diagonal bonds from the lower odd vertices
            \foreach \x in {-3,-1,1,3,5,7,9,11,13} {
                \draw[
                    latticeblack,
                    line width=0.85pt
                ]
                    (\x,{\yb+2})
                    --
                    ({\x-1},{\yb+1.5});

                \draw[
                    latticeblack,
                    line width=0.85pt
                ]
                    (\x,{\yb+2})
                    --
                    ({\x+1},{\yb+1.5});
            }
        }

        % Black lattice vertices
        \foreach \r in {-1,0,1,2,3} {
            \pgfmathsetmacro{\yb}{3*\r}

            % Odd-column vertices
            \foreach \x in {-3,-1,1,3,5,7,9,11,13} {
                \fill[latticeblack]
                    (\x,\yb) circle[radius=2.25pt];

                \fill[latticeblack]
                    (\x,{\yb+2}) circle[radius=2.25pt];

                \fill[latticeblack]
                    (\x,{\yb+3}) circle[radius=2.25pt];
            }

            % Even-column vertices
            \foreach \x in {-2,0,2,4,6,8,10,12} {
                \fill[latticeblack]
                    (\x,{\yb+0.5}) circle[radius=2.25pt];

                \fill[latticeblack]
                    (\x,{\yb+1.5}) circle[radius=2.25pt];
            }
        }

        % =========================================================
        % Coordinates of the dashed parallelogram
        % =========================================================

        % Outer corners
        \coordinate (A) at (0,5.50);     % lower-left
        \coordinate (B) at (3,1.00);     % upper-left
        \coordinate (C) at (9.25,1.00);  % upper-right, expanded rightward
        \coordinate (D) at (6.25,5.50);  % lower-right, expanded rightward

        % Intersections with the first internal line
        \coordinate (UL) at (2.00,2.60);
        \coordinate (UR) at (8.25,2.60);

        % Intersections with the second internal line
        \coordinate (LL) at (1.00,4.10);
        \coordinate (LR) at (7.25,4.10);

        % =========================================================
        % Slightly expanded black outer boundary
        % =========================================================

        % The black outline is kept slightly outside the colored
        % outer boundary so that it does not cover the colored dashes.
        \coordinate (Ablack) at (-0.18,5.59);
        \coordinate (Bblack) at (2.94,0.91);
        \coordinate (Cblack) at (9.43,0.91);
        \coordinate (Dblack) at (6.31,5.59);

        % Draw the black outline first so the colored dashed lines
        % remain visible above it.
        \draw[
            latticeblack,
            line width=1.15pt,
            dash pattern=on 3pt off 3pt,
            rounded corners=6pt
        ]
            (Ablack) --
            (Bblack) --
            (Cblack) --
            (Dblack) --
            cycle;

        % =========================================================
        % Colored portions of the outer dashed boundary
        % =========================================================

        % Upper blue part, including the two upper corners
        \draw[
            boundaryblue,
            line width=1.25pt,
            dash pattern=on 3pt off 3pt,
            dash phase=3pt,
            rounded corners=6pt
        ]
            (UL) -- (B) -- (C) -- (UR);

        % Middle magenta portions of the sloping sides
        \draw[
            boundarymagenta,
            line width=1.25pt,
            dash pattern=on 3pt off 3pt,
            dash phase=3pt
        ]
            (UL) -- (LL);

        \draw[
            boundarymagenta,
            line width=1.25pt,
            dash pattern=on 3pt off 3pt,
            dash phase=3pt
        ]
            (UR) -- (LR);

        % Lower red part, including both lower corners
        \draw[
            boundaryred,
            line width=1.25pt,
            dash pattern=on 3pt off 3pt,
            dash phase=3pt,
            rounded corners=6pt
        ]
            (LL) -- (A) -- (D) -- (LR);

        % =========================================================
        % Internal dashed boundaries
        % =========================================================

        % Boundary between Lambda_1 and Lambda_2:
        % alternating blue and magenta dashes
        \draw[
            boundaryblue,
            line width=1.20pt,
            dash pattern=on 3pt off 3pt
        ]
            (UL) -- (UR);

        \draw[
            boundarymagenta,
            line width=1.20pt,
            dash pattern=on 3pt off 3pt,
            dash phase=3pt
        ]
            (UL) -- (UR);

        % Boundary between Lambda_2 and Lambda_3:
        % alternating magenta and red dashes
        \draw[
            boundarymagenta,
            line width=1.20pt,
            dash pattern=on 3pt off 3pt
        ]
            (LL) -- (LR);

        \draw[
            boundaryred,
            line width=1.20pt,
            dash pattern=on 3pt off 3pt,
            dash phase=3pt
        ]
            (LL) -- (LR);

        % =========================================================
        % Labels
        % =========================================================

        \node[
            anchor=east,
            inner sep=1pt,
            text=latticeblack,
            font=\large
        ] at (8.25,2.10) {$\Lambda_{1}$};

        \node[
            anchor=east,
            inner sep=1pt,
            text=latticeblack,
            font=\large
        ] at (7.25,3.50) {$\Lambda_{2}$};

        \node[
            anchor=east,
            inner sep=1pt,
            text=latticeblack,
            font=\large
        ] at (6.25,4.95) {$\Lambda_{3}$};

        \node[
            anchor=north,
            inner sep=1pt,
            text=latticeblack,
            font=\large
        ] at (6.125,5.62) {$\Lambda$};

    \end{scope}
\end{tikzpicture}
    \caption{Example: a parallelogram $\Lambda$ of shape $(3,3)$ is partitioned into $3$ parallelograms of shape $(3,1)$, ordered from top to bottom.}
    \label{fig:: horizontal partition}
\end{figure}

We slightly extend the inner product defined in \eqref{eqn:: inner product on Hom} to morphism spaces of the form $\Hom(\gamma^m\rho^n,\rho^n \gamma^m)$ for $m,n\geq 0$. 
Namely, for $a_1,\dots,a_m$ and $b_1,\dots,b_m\in \mathrm{\Irr}(\mathcal{C})$, we define 
\begin{equation}
    \braket{\eta,\xi} = \prod^m_{i=1}\frac{1}{\sqrt{d_{a_i} d_{b_i}}}\, \tr_{\mathcal{C}}\left( \eta^* \left( \mathrm{D}_{\rho^n}\otimes \mathrm{Id}^{\otimes m}_{\gamma} \right)\xi \left( \mathrm{Id}^{\otimes m}_{\gamma}\otimes \mathrm{D}_{\rho^n} \right) \right),\quad \eta,\xi\in \Hom(\gamma^m\rho^n,\rho^n \gamma^m). 
\end{equation}
For $m=0$, this defines an inner product on $\End(\rho^n)$, with respect to which we shall fix an orthonormal basis $\mathrm{ONB}(\rho^n)$ of $\End(\rho^n)$. 
Hence we have 
\begin{equation}\label{eqn:: resolution with respect to ONB}
    f = \sum_{\eta\in \mathrm{ONB}(\rho^n)} \eta\tr_{\mathcal{C}}(\eta^* \mathrm{D}_{\rho^n} f \mathrm{D}_{\rho^n}), \quad f\in \End(\rho^n).
\end{equation}
The adjoint $f\mapsto f^*$ is anti-unitary with respect to $\braket{\cdot ,\cdot}$:
\begin{equation}
    \braket{\xi^*,\eta^*} = \tr_{\mathcal{C}}(\xi \mathrm{D}_{\rho^n}\eta^* \mathrm{D}_{\rho^n}) = \tr_{\mathcal{C}}( \eta^* \mathrm{D}_{\rho^n} \xi \mathrm{D}_{\rho^n}) = \braket{\eta,\xi},\quad \eta,\xi\in \End(\rho^n).
\end{equation}
Hence $\{f^*\vert f\in \mathrm{ONB}(\rho^n)\}$ is also a orthonormal basis, and \eqref{eqn:: resolution with respect to ONB} holds with $\eta$ and $\eta^*$ interchanged. 

\begin{definition}\label{def:: diamond composition}
    Let $n,m\geq 0$, for morphisms $f\in \Hom(\gamma^nx,y\gamma^n)$ and $g\in \Hom(\gamma^my,z\gamma^m)$, we define their graded composition $g\diamond f\in \Hom(\gamma^{m+n}x,z\gamma^{m+n})$ to be 
\begin{equation}\label{eqn:: the diamond product}
    \vcenter{\hbox{\begin{tikzpicture}
        \ThinLine[midarrow] (-0.25,0.75) -- (-0.25,0.25);
        \ThinLine[midarrow] (0.25,0.75) -- (0.25,0.25);
        \ThinLine[midarrow] (0.25,-0.25) -- (0.25,-0.75);
        \ThinLine[midarrow] (-0.25,-0.25) -- (-0.25,-0.75);
        \node [draw, minimum width=1cm, minimum height=0.25cm, fill = white] at (0,0) {$g$};
        \node at (-0.25,1) {$\gamma^m$};
        \node at (0.25,1) {$y$};
        \node at (-0.25,-1) {$z$};
        \node at (0.25,-1) {$\gamma^m$};
    \end{tikzpicture}}}\diamond \vcenter{\hbox{\begin{tikzpicture}
        \ThinLine[midarrow] (-0.25,0.75) -- (-0.25,0.25);
        \ThinLine[midarrow] (0.25,0.75) -- (0.25,0.25);
        \ThinLine[midarrow] (0.25,-0.25) -- (0.25,-0.75);
        \ThinLine[midarrow] (-0.25,-0.25) -- (-0.25,-0.75);
        \node [draw, minimum width=1cm, minimum height=0.25cm, fill = white] at (0,0) {$f$};
        \node at (-0.25,1) {$\gamma^n$};
        \node at (0.25,1) {$x$};
        \node at (-0.25,-1) {$y$};
        \node at (0.25,-1) {$\gamma^n$};
    \end{tikzpicture}}} := (g\otimes \mathrm{Id}_{\gamma^n})(\mathrm{Id}_{\gamma^m}\otimes f) = \vcenter{\hbox{\begin{tikzpicture}
        \ThinLine[midarrow] (-0.25,2.25) -- (-0.25,0.25);
        \ThinLine[midarrow] (0.25,-0.25) -- (0.25,-0.75);
        \ThinLine[midarrow] (-0.25,-0.25) -- (-0.25,-0.75);
        \node [draw, minimum width=1cm, minimum height=0.25cm, fill = white] at (0,0) {$g$};
        \node at (-0.25,2.5) {$\gamma^m$};
        \node at (-0.25,-1) {$z$};
        \node at (0.25,-1) {$\gamma^m$};
        \begin{scope}[shift={(0.5,1.5)}]
            \ThinLine[midarrow] (-0.25,0.75) -- (-0.25,0.25);
            \ThinLine[midarrow] (0.25,0.75) -- (0.25,0.25);
            \ThinLine[midarrow] (0.25,-0.25) -- (0.25,-2.25);
            \ThinLine[midarrow] (-0.25,-0.25) -- (-0.25,-1.25);
            \node [draw, minimum width=1cm, minimum height=0.25cm, fill = white] at (0,0) {$f$};
            \node at (-0.25,1) {$\gamma^n$};
            \node at (0.25,1) {$x$};
            \node at (0.25,-2.5) {$\gamma^n$};
        \end{scope}
    \end{tikzpicture}}}. 
\end{equation}
\end{definition}

Note that $\diamond$ is associative, and for $f,g\in \End(\rho^l)$, $\xi\in \Hom(\gamma\rho^l,\rho^l\gamma)$:
\begin{equation}
    \pi_L(f)\pi_R(g) \xi = f\diamond \xi\diamond g. 
\end{equation}
Using this, the inner product \eqref{eqn:: inner product on Hom} can then be expressed as
\begin{equation}\label{eqn:: skein module inner product with diamond}
    \braket{g,f} = \frac{1}{\sqrt{d_a d_b}}\tr_{\mathcal{C}}\left( g^* \mathrm{D}_{\rho^l}\diamond f \diamond \mathrm{D}_{\rho^l} \right),\quad f,g\in \Hom(a \rho^l,\rho^l b). 
\end{equation}

\begin{proposition}\label{prop::sewing two rows}
    Consider two adjacent regions $\Lambda_1$ and $\Lambda_2$ of shape $(l,1)$ with $\Lambda_1$ above $\Lambda_2$, as in Fig. \ref{fig:: horizontal partition}. 
    Define the operator $\chi(\Lambda_2,\Lambda_1)$ on $\mathcal{H}_{\Lambda_2}\otimes \mathcal{H}_{\Lambda_1}$ as 
    \begin{equation}\label{eqn:: def. of chi-operator}
        \chi(\Lambda_2,\Lambda_1)  = \sum_{\eta\in \mathrm{ONB}(\rho^l)} \mathrm{R}_{\Lambda_2} \left(\mathrm{D}_{\rho^l}\eta^* \left( \mathrm{D}_{\rho^l} \right)^{-1} \right)\otimes \mathrm{L}_{\Lambda_1}(\eta). 
    \end{equation}
    Then for any simple objects $a_1,a_2,b_1,b_2$ and morphisms $h_1,h_1'\in \Hom(a_1\rho^l,\rho^lb_1)$ and $h_2,h_2'\in \Hom(a_2\rho^l,\rho^lb_2)$, 
    \begin{align}
        &\braket{\mathrm{eval}^\dagger_{\Lambda_2}(h'_2)\otimes \mathrm{eval}^\dagger_{\Lambda_1}(h'_1)|\chi(\Lambda_2,\Lambda_1)|\mathrm{eval}^\dagger_{\Lambda_2}(h_2)\otimes \mathrm{eval}^\dagger_{\Lambda_1}(h_1)} = \braket{h'_2\diamond h'_1,h_2\diamond h_1}. 
    \end{align}
    Consequently, $\chi(\Lambda_2,\Lambda_1)\geq 0$ and is independent of the choice of the basis. 
    % \begin{equation}
    %     \left( \pi^{\Lambda_2}_R(h)\otimes \mathrm{I}_{\Lambda_1} -  \mathrm{I}_{\Lambda_2}\otimes \pi^{\Lambda_1}_L(h) \right)\sum_{\eta\in \mathrm{ONB}(\rho^l)} \pi^{\Lambda_2}_R(\eta^*)\otimes \pi^{\Lambda_1}_L(h\eta) = 0,\quad h\in \End(\rho^l).
    % \end{equation}
\end{proposition}
\begin{proof}
    By Lemma \ref{lemma:: eval in a single row} and \eqref{eqn:: bimodule structure of Tree}, we have 
    \begin{equation}
        \begin{aligned}
            &\braket{\mathrm{eval}^\dagger_{\Lambda_2}(h'_2)\otimes \mathrm{eval}^\dagger_{\Lambda_1}(h'_1)|\chi(\Lambda_2,\Lambda_1)|\mathrm{eval}^\dagger_{\Lambda_2}(h_2)\otimes \mathrm{eval}^\dagger_{\Lambda_1}(h_1)}\\
        &= \sum_{\eta \in \mathrm{ONB}(\rho^l)} \braket{\mathrm{eval}^\dagger_{\Lambda_2}(h'_2)|\mathrm{eval}^\dagger_{\Lambda_2}(\pi_R(\mathrm{D}_{\rho^l} \eta^* (\mathrm{D}_{\rho^l})^{-1} )h_2 )} \braket{\mathrm{eval}^\dagger_{\Lambda_1}(h'_1)|\mathrm{eval}^\dagger_{\Lambda_1}(\eta h_1)}\\
        &= \sum_{\eta \in \mathrm{ONB}(\rho^l)} \braket{h'_2, \pi_R(\mathrm{D}_{\rho^l} \eta^* (\mathrm{D}_{\rho^l})^{-1} )h_2} \braket{h'_1, \pi_L(\eta) h_1}. 
        \end{aligned}
    \end{equation}
    Denote by $\mathbb{E}: \End(\rho^l b_1)\rightarrow \End(\rho^l)$ the $\tr_{\mathcal{C}}$-preserving partial trace (adding a right cap). 
    Then 
    \begin{equation}
        \begin{aligned}
            \braket{h'_1, \pi_L(\eta) h_1} &= \frac{1}{\sqrt{d_{a_1} d_{b_1}}}\, \tr_{\mathcal{C}}\left( (h'_1)^*\left( \mathrm{D}_{\rho^l}\eta \otimes \mathrm{Id}_{b_1} \right) h_1 \left( \mathrm{Id}_{a_1}\otimes \mathrm{D}_{\rho^l} \right) \right)\\
            &= \frac{1}{\sqrt{d_{a_1} d_{b_1}}}\, \tr_{\mathcal{C}} \left(  \eta \mathbb{E}\left[ h_1 \left( \mathrm{Id}_{a_1}\otimes \mathrm{D}_{\rho^l} \right) (h'_1)^* \right] \mathrm{D}_{\rho^l}\right)\\
            &= \frac{1}{\sqrt{d_{a_1} d_{b_1}}} \braket{\eta^*, (\mathrm{D}_{\rho^l})^{-1}\mathbb{E}\left[ h_1 \left( \mathrm{Id}_{a_1}\otimes \mathrm{D}_{\rho^l} \right) (h'_1)^* \right]}.
        \end{aligned}
    \end{equation}
    Plugging in, we obtain
    \begin{equation}
        \begin{aligned}
            &\sum_{\eta \in \mathrm{ONB}(\rho^l)} \braket{h'_2, \pi_R(\mathrm{D}_{\rho^l} \eta^* (\mathrm{D}_{\rho^l})^{-1} )h_2} \braket{h'_1, \pi_L(\eta) h_1}\\
            &= \frac{1}{\sqrt{d_{a_1} d_{b_1}}} \braket{h'_2, h_2 \diamond \mathbb{E}\left[ h_1 \left( \mathrm{Id}_{a_1}\otimes \mathrm{D}_{\rho^l} \right) (h'_1)^* \right](\mathrm{D}_{\rho^l})^{-1}}\\
            &= \frac{1}{\sqrt{d_{a_1} d_{b_1}}}\frac{1}{\sqrt{d_{a_2} d_{b_2}}}\, \tr_{\mathcal{C}}\left( (h'_2)^* \left( \mathrm{D}_{\rho^l}\otimes \mathrm{Id}_{b_2} \right) h_2\diamond \mathbb{E}\left[ h_1 \left( \mathrm{Id}_{a_1}\otimes \mathrm{D}_{\rho^l} \right) (h'_1)^* \right] \right)\\
            &= \braket{h'_2\diamond h'_1,h_2\diamond h_1},
        \end{aligned}
    \end{equation}
    where in the last equality we used again that $\mathbb{E}$ preserves the trace. 
    Since $\chi(\Lambda_2,\Lambda_1)$ is zero on the complement of $\mathcal{H}^{\mathrm{st.n.}}_{\Lambda_2}\otimes \mathcal{H}^{\mathrm{st.n.}}_{\Lambda_1}$, $\chi(\Lambda_2,\Lambda_1)$ is a positive operator that is independent of the choice of basis. 
\end{proof}

\begin{lemma}\label{lemma:: key central element}
    For $n\geq 0$, the operator 
    \begin{equation}
        \Delta_n = \sum _{\eta\in \mathrm{ONB}(\rho^n)} \eta^* \mathrm{D}^{-1}_{\rho^n}\eta \mathrm{D}_{\rho^n} 
    \end{equation}
    is independent of the choice of $\mathrm{ONB}(\rho^n)$, strictly positive, and central in $\End(\rho^n)$. 
\end{lemma}
\begin{proof}
    Let $\mathcal{B}$ and $\mathcal{B}'$ be two orthonormal bases of $\End(\rho^n)$, and write $\Delta_n^{\mathcal{B}}$ and $\Delta_n^{\mathcal{B}'}$ for the corresponding elements.
    For any $f\in \End(\rho^n)$, by \eqref{eqn:: resolution with respect to ONB}, 
    \begin{equation}
        \begin{aligned}
            f\Delta_n^{\mathcal{B}}
            = \sum_{\substack{\eta\in\mathcal{B}\\ \xi\in\mathcal{B}'}} \xi^*\tr_{\mathcal{C}}\left(\xi \mathrm{D}_{\rho^n} f\eta^* \mathrm{D}_{\rho^n}\right) \mathrm{D}^{-1}_{\rho^n} \eta \mathrm{D}_{\rho^n}= \sum_{\xi\in\mathcal{B}'} \xi^*\mathrm{D}^{-1}_{\rho^n} \left( \xi \mathrm{D}_{\rho^n} f \mathrm{D}^{-1}_{\rho^n} \right) \mathrm{D}_{\rho^n} = \Delta_n^{\mathcal{B}'}f.
        \end{aligned}
    \end{equation}
    Taking $f=\mathrm{Id}_{\rho^n}$ shows that $\Delta_n^{\mathcal{B}}=\Delta_n^{\mathcal{B}'}$, so $\Delta_n$ is independent of the choice of orthonormal basis.
    The same identity for arbitrary $f$ then gives $f\Delta_n=\Delta_n f$, proving that $\Delta_n$ is central.
    In particular, $\Delta_n$ commutes with $\mathrm{D}_{\rho^n}^{1/2}$, and hence
    \begin{equation}
        \begin{aligned}
            \Delta_n
            &= \mathrm{D}_{\rho^n}^{1/2}\Delta_n\mathrm{D}_{\rho^n}^{-1/2}= \sum_{\eta\in\mathrm{ONB}(\rho^n)}
            \left(\mathrm{D}_{\rho^n}^{-1/2}\eta\mathrm{D}_{\rho^n}^{1/2}\right)^*
            \left(\mathrm{D}_{\rho^n}^{-1/2}\eta\mathrm{D}_{\rho^n}^{1/2}\right),
        \end{aligned}
    \end{equation}
    so $\Delta_n$ is positive.
    Let $p$ be the kernel projection of $\Delta_n$.
    Since $\Delta_n$ is central, so is $p$, and the preceding sum of squares implies $\mathrm{D}_{\rho^n}^{-1/2}\eta\mathrm{D}_{\rho^n}^{1/2}p=0$ for any $\eta\in\mathrm{ONB}(\rho^n)$. 
    The elements $\mathrm{D}_{\rho^n}^{-1/2}\eta\mathrm{D}_{\rho^n}^{1/2}$ span $\End(\rho^n)$ because conjugation by $\mathrm{D}_{\rho^n}^{1/2}$ is an invertible linear map.
    Therefore $p=0$,and $\Delta_n$ is strictly positive.
\end{proof}

\begin{lemma}\label{lemma:: Key identity for local ground states}
    Let $\Lambda_1$ and $\Lambda_2$ be as in Proposition \ref{prop::sewing two rows}. 
    Then for any $\ket{\psi}$ in the $\mathcal{H}^{\mathrm{st.n.}}_{\Lambda_2}\otimes \mathcal{H}^{\mathrm{st.n.}}_{\Lambda_1}$, we have
    \begin{equation}\label{eqn:: pulling through}
         \mathrm{L}_{\Lambda_1}(h)\ket{\psi} = \mathrm{R}_{\Lambda_2}(\mathrm{D}^2_{\rho^l} h \mathrm{D}^{-2}_{\rho^l})\ket{\psi},\quad h\in \End(\rho^l) 
    \end{equation}
    if and only if $\ket{\psi}$ is in the range of $\chi(\Lambda_2,\Lambda_1)$. 
\end{lemma}
\begin{proof}
    For any $h\in \End(\rho^l)$, we have 
    \begin{equation}
        \begin{aligned}
           \mathrm{L}_{\Lambda_1}(h) \chi(\Lambda_2,\Lambda_1) &= \sum_{\eta\in \mathrm{ONB}(\rho^l)} \mathrm{R}_{\Lambda_2} \left(\mathrm{D}_{\rho^l}\eta^*  \mathrm{D}^{-1}_{\rho^l} \right)\otimes \mathrm{L}_{\Lambda_1}(h\eta) \\
        &= \sum_{\eta,\xi\in \mathrm{ONB}(\rho^l)} \tr_{\mathcal{C}}(\xi^* \mathrm{D}_{\rho^l}h\eta \mathrm{D}_{\rho^l}) \mathrm{R}_{\Lambda_2} \left(\mathrm{D}_{\rho^l}\eta^*  \mathrm{D}^{-1}_{\rho^l} \right)\otimes \mathrm{L}_{\Lambda_1} \left(\xi \right)\\
        &= \sum_{\eta,\xi\in \mathrm{ONB}(\rho^l)} \braket{\eta^*,\xi^*\mathrm{D}_{\rho^l} h \mathrm{D}^{-1}_{\rho^l}} \mathrm{R}_{\Lambda_2} \left(\mathrm{D}_{\rho^l}\eta^*  \mathrm{D}^{-1}_{\rho^l}  \right)\otimes \mathrm{L}_{\Lambda_1} \left(\xi \right)\\
        &= \sum_{\xi\in \mathrm{ONB}(\rho^l)} \mathrm{R}_{\Lambda_2} \left(\mathrm{D}_{\rho^l}\left( \xi^*\mathrm{D}_{\rho^l} h \mathrm{D}^{-1}_{\rho^l} \right) \mathrm{D}^{-1}_{\rho^l} \right)\otimes \mathrm{L}_{\Lambda_1} \left(\xi \right)\\
        &=  \mathrm{R}_{\Lambda_2}(\mathrm{D}^2_{\rho^l} h \mathrm{D}^{-2}_{\rho^l}) \chi(\Lambda_2,\Lambda_1). 
        \end{aligned}
    \end{equation}
    So every state in the range of $\chi(\Lambda_2,\Lambda_1)$ satisfies the desired equality. 
    Conversely, for any $\ket{\psi}$ in the string-net subspace of $\mathcal{H}^{\mathrm{st.n.}}_{\Lambda_2}\otimes \mathcal{H}^{\mathrm{st.n.}}_{\Lambda_1}$ that satisfies \eqref{eqn:: pulling through}, we have
    \begin{equation}
        \begin{aligned}
        \chi(\Lambda_2,\Lambda_1)\ket{\psi} &= \sum_{\eta\in \mathrm{ONB}(\rho^l)} \mathrm{R}_{\Lambda_2} \left(\mathrm{D}_{\rho^l}\eta^*  \mathrm{D}^{-1}_{\rho^l} \right) \mathrm{R}_{\Lambda_2}(\mathrm{D}^2_{\rho^l} \eta \mathrm{D}^{-2}_{\rho^l})\otimes \mathrm{I}_{\Lambda_1}\ket{\psi} \\
        &= \sum_{\eta\in \mathrm{ONB}(\rho^l)}  \mathrm{R}_{\Lambda_2}\left(\mathrm{D}^2_{\rho^l} \eta \mathrm{D}^{-1}_{\rho^l}\eta^*  \mathrm{D}^{-1}_{\rho^l} \right)\ket{\psi}\\
        &= \mathrm{R}_{\Lambda_2}(\Delta_l)\ket{\psi}. 
        \end{aligned}
    \end{equation}
    In the last equality, we used that $\{\eta^*\}_{\eta\in \mathrm{ONB}(\rho^l)}$ is a basis, and that $\mathrm{D}^2_{\rho^l} \Delta_l\mathrm{D}^{-2}_{\rho^l} = \Delta_l$. 
    By Lemma \ref{lemma:: key central element}, $\Delta_l$ is a strictly positive operator in the center of $\End(\rho^l)$, $\mathrm{R}_{\Lambda_2}(\Delta_l)\otimes \mathrm{I}_{\Lambda_1}$ commutes with $\chi(\Lambda_2,\Lambda_1)$. 
    Thus $\ket{\psi}$ is in the range of $\chi(\Lambda_2,\Lambda_1)$, and the claim follows.  
\end{proof}

Given a Hamiltonian decomposed as a sum of operators $H = \sum_j h_j$, we call $H$ \emph{frustration-free}, if there exists a (and therefore every) ground state $\ket{\psi}$ of $H$ that is also a ground state of $h_j$ for every $j$. 
We say an interaction is frustration free, if all local Hamiltonians defined by it are. 

\begin{proposition}\label{prop:: local ground states in a row}
    Let $\Lambda_1$, $\Lambda_2$ be two adjacent parallelograms of shape $(l,1)$, $l\geq k$, with $\Lambda_1$ on top of $\Lambda_2$. 
    Then the local Hamiltonian $H_{\Lambda_2\cup \Lambda_1}$ is frustration-free, and its ground state projection is given by:
    \begin{equation}
        \Pi(\Lambda_2\cup \Lambda_1) = \mathrm{R}_{\Lambda_2}(\Delta^{-1}_l)\chi(\Lambda_2,\Lambda_1). 
    \end{equation}
    Consequently, the map 
    \begin{equation}\label{eqn:: isometry from local ground states to morphism space in a row}
        \mathcal{U}: \chi(\Lambda_2,\Lambda_1)^{1/2}\ket{\xi_2\otimes \xi_1}\mapsto \mathrm{eval}_{\Lambda_2}(\xi_2)\diamond \mathrm{eval}_{\Lambda_1}(\xi_1),
    \end{equation}
     where $\xi_i\in \mathcal{H}_{\Lambda_i}$ for $i=1,2$, extends to an isometry from the ground state subspace of $H_{\Lambda_2\cup \Lambda_1}$ to $\Hom(\gamma^2\rho^l,\rho^l\gamma^2)$. 
\end{proposition}
\begin{proof}
    Since every term in 
    \begin{equation}
        H_{\Lambda_2\cup\Lambda_1} = \sum_{e\subset \Lambda_1}\mathrm{I} - Q_e + \sum_{e\subset \Lambda_2}\mathrm{I} - Q_e + \sum_{S\in \mathcal{P}_{k-1}(\Lambda_2\cup \Lambda_1)} F_S
    \end{equation}
    is positive, its ground state space is the intersection of the kernels of its edge and plaquette-region terms.
    The edge terms first restrict a ground state $\ket{\psi}$ to $\mathcal{H}^{\mathrm{st.n.}}_{\Lambda_2}\otimes\mathcal{H}^{\mathrm{st.n.}}_{\Lambda_1}$.
    For $0\leq i\leq l-k$ and $f\in\mathcal{G}$, set
    \begin{equation}
        \widetilde{f_i}=\mathrm{Id}_{\rho}^{\otimes i}\otimes f\otimes \mathrm{Id}_{\rho}^{\otimes(l-i-k)}\in\End(\rho^l).
    \end{equation}
    Now let $S$ be a sequence of $k-1$ plaquettes in $\Lambda_2\cup \Lambda_1$ beginning at vertex $v_{2i+1}$ in each row. 
    By \eqref{eqn:: F_S} and Lemma \ref{lemma:: propagate actions to bigger regions}, a ground state $\ket{\psi}$ of $\sum_{e\in E(\Lambda_2\cup \Lambda_1)} (\mathrm{I} - Q_e)$ is also a ground state of $F_S$ if and only if 
    \begin{equation}\label{eqn:: identity imposed by absence of frustration}
        \mathrm{R}_{\Lambda_2}\left(\mathrm{D}_{\rho^l}\widetilde{f_i}\mathrm{D}_{\rho^l}^{-1}\right) \ket{\psi} = \mathrm{L}_{\Lambda_1}\left(\mathrm{D}_{\rho^l}^{-1}\widetilde{f_i}\mathrm{D}_{\rho^l}\right)\ket{\psi}, \quad f\in \mathcal{G}. 
    \end{equation}
    Observe that if the above identity holds for $f,g\in\End(\rho^l)$, then it also holds $fg$, as
    \begin{equation}
        \begin{aligned}
            &\mathrm{R}_{\Lambda_2}\left(\mathrm{D}_{\rho^l}fg\mathrm{D}_{\rho^l}^{-1}\right) \ket{\psi} = \mathrm{R}_{\Lambda_2}\left(\mathrm{D}_{\rho^l}g\mathrm{D}_{\rho^l}^{-1}\right)\mathrm{R}_{\Lambda_2}\left(\mathrm{D}_{\rho^l}f\mathrm{D}_{\rho^l}^{-1}\right) \ket{\psi}\\
            &= \mathrm{R}_{\Lambda_2}\left(\mathrm{D}_{\rho^l}g\mathrm{D}_{\rho^l}^{-1}\right) \mathrm{L}_{\Lambda_1}\left( \mathrm{D}^{-1}_{\rho^l} f \mathrm{D}_{\rho^l} \right) \ket{\psi} \\
            &= \mathrm{L}_{\Lambda_1}\left( \mathrm{D}^{-1}_{\rho^l} f \mathrm{D}_{\rho^l} \right) \mathrm{R}_{\Lambda_2}\left(\mathrm{D}_{\rho^l}g\mathrm{D}_{\rho^l}^{-1}\right)\ket{\psi} \\
            &= \mathrm{L}_{\Lambda_1}\left( \mathrm{D}^{-1}_{\rho^l} f \mathrm{D}_{\rho^l} \right) \mathrm{L}_{\Lambda_1}\left( \mathrm{D}^{-1}_{\rho^l} g \mathrm{D}_{\rho^l} \right)\ket{\psi} = \mathrm{L}_{\Lambda_1}\left( \mathrm{D}^{-1}_{\rho^l} fg \mathrm{D}_{\rho^l} \right)\ket{\psi}. 
        \end{aligned}
    \end{equation} 
    Since $\mathcal{G}$ is a generating set, $\ket{\psi}$ is a frustration-free ground state of $H_{\Lambda_2\cup \Lambda_1}$ if and only if \eqref{eqn:: identity imposed by absence of frustration} holds with $\widetilde{f_i}$ replaced by any $h\in \End(\rho^l)$. 
    Note that for $h\in \End(\rho^l)$, we have $\mathrm{D}_{\rho^l} h \mathrm{D}^{-1}_{\rho^l} = \mathrm{D}^2_{\rho^l} \left( \mathrm{D}^{-1}_{\rho^l} h \mathrm{D}_{\rho^l} \right) \mathrm{D}^{-2}_{\rho^l}$. 
    Thus, $\ket{\psi}$ is a frustration-free ground state of $H_{\Lambda_2\cup \Lambda_1}$ if and only if
    \begin{equation}
        \left(\mathrm{R}_{\Lambda_2}\left(\mathrm{D}^2_{\rho^l} h\mathrm{D}_{\rho^l}^{-2}\right)- \mathrm{L}_{\Lambda_1}\left( h \right)\right)\ket{\psi}=0, \quad h\in \End(\rho^l). 
    \end{equation}
    By Lemma \ref{lemma:: Key identity for local ground states}, the preceding identity identifies the ground state space with the range of $\chi(\Lambda_2,\Lambda_1)$ and gives
    \begin{equation}
        \chi(\Lambda_2,\Lambda_1)^2 = \left(\mathrm{R}_{\Lambda_2}(\Delta_l)\otimes \mathrm{I}_{\Lambda_1} \right)\chi(\Lambda_2,\Lambda_1),
    \end{equation}
    By Lemma \ref{lemma:: key central element}, $\mathrm{R}_{\Lambda_2}(\Delta_l)\otimes \mathrm{I}_{\Lambda_1}$ is strictly positive on the product string-net subspace and commutes with $\chi(\Lambda_2,\Lambda_1)$.
    Therefore, $\Pi(\Lambda_2\cup \Lambda_1) = \left(\mathrm{R}_{\Lambda_2}(\Delta^{-1}_l)\otimes \mathrm{I}_{\Lambda_1} \right)\chi(\Lambda_2,\Lambda_1)$ is the ground state projection of $H_{\Lambda_2\cup \Lambda_1}$.
    Finally, that $\mathcal{U}$ is an isometry follows from Proposition \ref{prop::sewing two rows}. 
\end{proof}

\begin{theorem}\label{thm:: local ground states in a parallelogram}
    Let $\Lambda$ be a parallelogram of shape $(l,h)$ with $l\geq k$ and $h\geq 2$. Decompose $\Lambda$ into regions  $\Lambda_1,\Lambda_2,\dots, \Lambda_h$ of shape $(l,1)$ such that $\Lambda = \bigcup^h_{i=1} \Lambda_i$, and $\Lambda_{i}$ is on top of $\Lambda_{i+1}$ for $1\leq i\leq h-1$. 
    Then the local Hamiltonian $H_{\Lambda}$ is frustration-free, with ground state projection $\Pi(\Lambda)$ given by
    \begin{equation}
        \Pi(\Lambda) = \prod^{h-1}_{i=1} \Pi(\Lambda_{i+1}\cup\Lambda_i). 
    \end{equation}
    Moreover, the map
    \begin{equation}
        \mathcal{U}_{h:1}: \prod^{h-1}_{i=1}\chi(\Lambda_{i+1},\Lambda_i)^{1/2} \ket{\xi_{h}\otimes \cdots\otimes \xi_1}\mapsto \mathrm{eval}_{\Lambda_h}(\xi_h)\diamond \cdots \diamond \mathrm{eval}_{\Lambda_1}(\xi_1)
    \end{equation}
    extends to an isometry from the ground state subspace of $H_{\Lambda}$ to $\Hom(\gamma^h \rho^l,\rho^l\gamma^h)$. 
\end{theorem}
\begin{proof}
    By Proposition \ref{prop:: local ground states in a row} and Corollary \ref{corollary:: commutativity between left/right actions}, projections associated with adjacent pairs of rows commute, while those associated with disjoint pairs act on disjoint tensor factors and hence also commute.
    Therefore, $\displaystyle \Pi_0 = \prod^{h-1}_{i=1} \Pi(\Lambda_{i+1}\cup\Lambda_i)$ is an orthogonal projection onto the common ground-state subspace of the local Hamiltonians $H_{\Lambda_{i+1}\cup \Lambda_i}$.
    Every interaction term supported in $\Lambda$ is contained in some adjacent pair of rows, so positivity of the interaction terms gives
    \begin{equation}
        \ker H_{\Lambda}=\bigcap_{i=1}^{h-1}\ker H_{\Lambda_{i+1}\cup\Lambda_i}=\range(\Pi_0).
    \end{equation}
    Moreover, Proposition \ref{prop:: local ground states in a row} identifies each $\Pi(\Lambda_{i+1}\cup\Lambda_i)$ as the support projection of $\chi(\Lambda_{i+1},\Lambda_i)$; these positive operators commute by the same argument, so the range of $\prod_{i=1}^{h-1}\chi(\Lambda_{i+1},\Lambda_i)^{1/2}$ is $\range(\Pi_0)$.

    We prove the second statement by induction on $h$. 
    The base case, for which $h=2$, is proved in Proposition \ref{prop::sewing two rows}. 
    Now take $h>2$, and assume $\mathcal{U}_{m:1}$ is an isometry for some $2\leq m\leq h-1$. 
    By Corollary \ref{corollary:: commutativity between left/right actions} and the expression of $\chi(\Lambda_m,\Lambda_{m-1})$ in \eqref{eqn:: def. of chi-operator}, for $f\in\End(\rho^l)$ we have 
    \begin{equation}
        \begin{aligned}
            &\mathcal{U}_{m:1}\mathrm{L}_{\Lambda_m}(f)\left( \prod^{m-1}_{i=1}\chi(\Lambda_{i+1},\Lambda_i)^{1/2} \ket{\xi_{m}\otimes \cdots\otimes \xi_1} \right)\\
            &= \mathcal{U}_{m:1}\left( \prod^{m-1}_{i=1}\chi(\Lambda_{i+1},\Lambda_i)^{1/2} \ket{\mathrm{L}_{\Lambda_m}(f)\xi_{m}\otimes \cdots\otimes \xi_1} \right)\\
            &= f\diamond \mathcal{U}_{m:1}\left( \prod^{m-1}_{i=1}\chi(\Lambda_{i+1},\Lambda_i)^{1/2} \ket{\xi_{m}\otimes \cdots\otimes \xi_1} \right). 
        \end{aligned}
    \end{equation}
    Since the $\chi$-operators commute, for $\xi_i,\eta_i\in \mathcal{H}^{\mathrm{st.n.}}_{\Lambda_i}$ with $1\leq i\leq m+1$, the inner product of the corresponding vectors in the domain of $\mathcal{U}_{m+1:1}$ is
    \begin{equation}
        \begin{aligned}
            &\Braket{\eta_{m+1}\otimes \eta_{m}\otimes \cdots \otimes \eta_1|\chi(\Lambda_{m+1},\Lambda_m) \prod^{m-1}_{i=1}\chi(\Lambda_{i+1},\Lambda_i)|\xi_{m+1}\otimes \xi_{m}\otimes \cdots \otimes \xi_1} \\
            &= \sum_{\zeta\in \mathrm{ONB}(\rho^l)} \Braket{\mathrm{eval}_{\Lambda_m}(\eta_{m})\diamond \cdots \diamond \mathrm{eval}_{\Lambda_1}(\eta_1),\zeta\diamond\mathrm{eval}_{\Lambda_m}(\xi_m)\diamond \cdots \diamond \mathrm{eval}_{\Lambda_1}(\xi_1)}\\
            &\times \Braket{\mathrm{eval}_{\Lambda_{m+1}}(\eta_{m+1}),\mathrm{eval}_{\Lambda_{m+1}}(\xi_{m+1})\diamond \left( \mathrm{D}_{\rho^l} \zeta^* \mathrm{D}^{-1}_{\rho^l}\right)}. 
        \end{aligned}
    \end{equation}
    By the same calculation as in the proof of Proposition \ref{prop::sewing two rows}, with $\mathbb{E}$ the categorical partial trace over $\gamma^m$ from $\End(\rho^l \gamma^m)$ to $\End(\rho^l)$, the above sum reduces to 
    \begin{equation}
        \Braket{\mathrm{eval}_{\Lambda_{m+1}}(\eta_{m+1})\diamond \mathrm{eval}_{\Lambda_m}(\eta_{m})\diamond \cdots \diamond \mathrm{eval}_{\Lambda_1}(\eta_1),\mathrm{eval}_{\Lambda_{m+1}}(\xi_{m+1}) \diamond\mathrm{eval}_{\Lambda_m}(\xi_m)\diamond \cdots \diamond \mathrm{eval}_{\Lambda_1}(\xi_1)}.
    \end{equation}
    Thus $\mathcal{U}_{m+1:1}$ is an isometry. 
    Finally, surjectivity of each $\mathrm{eval}_{\Lambda_i}$ allows us to choose $\xi_i$ such that $\mathrm{eval}_{\Lambda_i}(\xi_i) = \mathrm{Id}_{\rho^l}$. 
    As their composition is nonzero, we obtain $\Pi_0\neq 0$. 
    Hence $H_{\Lambda}$ is frustration-free and $\Pi(\Lambda)=\Pi_0$. 
\end{proof}

\begin{corollary}\label{corollary:: frustration-freeness}
    The interaction $\varPhi$ defined in \eqref{eqn: interaction of new model} is frustration-free.
\end{corollary}
\begin{proof}
    By Theorem \ref{thm:: local ground states in a parallelogram}, for any parallelogram $\Lambda$, the local Hamiltonian $H_{\Lambda}$ is frustration-free. 
    For any finite region $\Lambda'$, we can find a parallelogram $\Lambda$ such that $\Lambda' \subseteq \Lambda$. 
    Take an arbitrary non-zero ground state $\ket{\psi}$ of $H_{\Lambda}$, then $\ket{\psi}$ minimizes the energy of the local terms in $H_{\Lambda'}$ as well. 
    Consider the Schmidt decomposition 
    \begin{equation}
        \ket{\psi} = \sum_i \lambda_i \ket{\psi_i}\otimes \ket{\phi_i}
    \end{equation}
    where $\ket{\psi_i}\in \mathcal{H}_{\Lambda'}$. 
    Since operators localized in $\Lambda'$ act as identity on $\ket{\phi_i}$, 
    the state $\ket{\psi_i}$ minimizes the energy of the local terms in $H_{\Lambda'}$, hence are frustration-free ground states of $H_{\Lambda'}$. 
\end{proof}

\section{Fusion Spin Chains from OS Reconstruction}\label{Sec:: Fusion Spin Chains from OS Reconstruction}

In this section, we use the framework developed in \cite{LiuZhao2025RPTO} to construct the boundary algebras of the model defined in \eqref{eqn: interaction of new model}. 
We then show that they are isomorphic to the anyon chain algebras determined by $\rho$. 
The intuition behind our construction originates from the path integral formulation. 
A quantum spin system defined on a $2d$ spatial lattice can be viewed as a (2+1)D Euclidean field theory via imaginary time evolution. When the Hamiltonian exhibits spatial reflection positivity, its correlation functions are positive with respect to spatial reflection across a chosen line. 
The Osterwalder-Schrader (OS) reconstruction applied to such a system acts effectively as a spatial path integral, yielding a Hilbert space that encodes the physical degrees of freedom at the boundary defined by the reflection line. 

\subsection{Reflection positivity}

Here we briefly recall the definition of reflection positivity (RP) in the context of quantum spin systems. 
Let $\theta:\mathbb{R}^2\rightarrow \mathbb{R}^2$ be the reflection with respect to the $x$-axis, and denote by $\mathbb{R}^2_+$/$\mathbb{R}^2_-$ the upper-/lower- half plane. 
Let $\Gamma$ be a lattice in $\mathbb{R}^2$. 
For $\Lambda\subseteq \Gamma$, we denote $\Lambda_{\pm} = \Lambda \cap \mathbb{R}^2_{\pm}$. 
We assume that $\Gamma = \Gamma_{-}\sqcup \Gamma_{+}$, and $\theta$ maps $\Gamma_\pm$ to $\Gamma_\mp$ bijectively. 
Consequently, we have $\theta(\Gamma_{+}) = \Gamma_{-}$ and $\theta(\Gamma_{-}) = \Gamma_{+}$. 

We consider a quantum spin systems with finite dimensional Hilbert spaces $\mathcal{H}_v$ assigned to each $v\in \Gamma$. 
We assume that the local Hilbert space has a tensor product structure: $\mathcal{H}_{\Lambda} = \bigotimes_{v\in \Lambda}\mathcal{H}_{v}$ for every finite $\Lambda\in \Gamma$. 
For a subset $\Gamma'\subseteq \Gamma$, we denote $\mathfrak{A}_{\text{loc}}(\Gamma')$ the $*$-algebra of observables supported in $\Gamma'$, which is a finite-dimensional matrix algebra when $\Gamma'$ is finite. 
We fix a family of anti-unitary operators $\widehat{\theta}_v: \mathcal{H}_v \rightarrow \mathcal{H}_{\theta(v)}, v\in \Gamma_{+}.$
For a local observable $X \in \mathfrak{A}_{\text{loc}}(\Gamma_{+})$ supported on $\Gamma_{+}$, we define its reflection (time reversal) in $\Gamma_{-}$ as
\begin{equation}
    \Theta(X) = \widehat{\theta}_{\mathrm{supp} X} X \widehat{\theta}_{\mathrm{supp} X}^{-1},\quad \widehat{\theta}_{\mathrm{supp} X} = \bigotimes_{i\in \mathrm{supp}X} \widehat{\theta}_i.
\end{equation}
Since the operators $\widehat{\theta}_i$ are anti-unitary, $\Theta$ is norm-preserving and preserves both the multiplicative and the $*$-structures. Therefore, it can be uniquely extended to an anti-linear $*$-isomorphism from $\mathfrak{A}(\Gamma_{+})$ to $\mathfrak{A}(\Gamma_{-})$. 

\begin{definition}
    Let $\mathcal{H}_-, \mathcal{H}_+$ be finite-dimensional Hilbert spaces, and let $\widehat{\theta}: \mathcal{H}_+\rightarrow \mathcal{H}_-$ be an anti-unitary. 
    A self-adjoint operator $H\in \mathcal{B}(\mathcal{H}_-\otimes \mathcal{H}_+)$ is called reflection positive (with respect to $\widehat{\theta}$), if 
    \begin{equation}
        \Tr(\mathrm{e}^{-\tau H}\Theta(X)\otimes X) \geq 0, \quad \tau \geq 0
    \end{equation}
    for every $X\in \mathcal{B}(\mathcal{H}_+)$. 
    Here $\Theta(X) = \widehat{\theta}X\widehat{\theta}^{-1}$.  
\end{definition}

The following is a sufficient condition for $H$ to be reflection positive: 

\begin{proposition}[\cite{FILS1978}]\label{prop:: creterion for RP}
    Let $H$ be a self-adjoint operator on $\mathcal{H}_- \otimes \mathcal{H}_+$.
    Suppose that there exist a self-adjoint operator $H_+$ on $\mathcal{H}_{+}$ and a finite set $\{O_j\}^m_{j=1}$ in $\mathcal{B}(\mathcal{H}_{+})$ such that
    \begin{equation}
        H = H_+ + \Theta(H_+) - \sum^m_{j=1}\Theta(O_j)\otimes O_j.
    \end{equation}
    Then $H$ is reflection positive with respect to $\Theta$.
\end{proposition}

\begin{definition}
    Let $\varPhi$ be an interaction on $\Gamma$. 
    We say that $\varPhi$ is reflection positive with respect to $\Theta$ if the following conditions hold:
    \begin{itemize}
        \item For any finite subset $\Lambda \subseteq \Gamma_{+}$, we have $\Theta(\varPhi(\Lambda)) = \varPhi(\theta(\Lambda))$.
        \item For any finite subset $\Lambda$ such that $\Lambda_{+},\Lambda_{-}\neq\emptyset$, $\varPhi(\Lambda) = 0$ if $\theta(\Lambda)\neq \Lambda$. 
        Otherwise, there exist a self-adjoint operator $K_{\Lambda}\in \mathfrak{A}_{\Lambda_+}$ and a finite linearly independent set $\{O_{\Lambda,j}\}_{j=1}^{m_{\Lambda}}\subseteq \mathfrak{A}_{\Lambda_+}$ that is invariant under the adjoint operation such that
        \begin{equation}
            \varPhi(\Lambda)=\Theta(K_{\Lambda})+  K_{\Lambda}-\sum_{j=1}^{m_{\Lambda}}\Theta(O_{\Lambda,j})\otimes O_{\Lambda,j}.
        \end{equation}
    \end{itemize}
    Note that the second condition ensures that for $\Lambda\cap \Gamma_\pm\neq \emptyset$, $\varPhi(\Lambda)\neq 0$ only if $\Lambda$ is reflection symmetric. 
\end{definition}

We now demonstrate that the interaction $\varPhi$ defined in \eqref{eqn: interaction of new model} is reflection positive with respect to every horizontal reflection. 
We choose $\theta$ to be the horizontal reflection with respect to a line that passes through the middle of two rows of $\Gamma$. 
For each vertex $v\in \Gamma_{+}$, we define $\widehat{\theta}_v$ as 
\begin{equation}
    \widehat{\theta}_v\Ket{\vcenter{\hbox{\begin{tikzpicture}
    \ThinLine[orange] (0:0) -- (-90:0.5);
    \ThinLine (0:0) -- (150:0.5);
    \ThinLine (0:0) -- (30:0.5);
    \draw[fill=black] (0:0) ellipse (0.05 and 0.05);
    \node at (-0.2,-0.25) {$\xi$};
\end{tikzpicture}}}} = \Ket{\vcenter{\hbox{\begin{tikzpicture}
    \ThinLine[orange] (90:0.5) -- (0:0);
    \ThinLine (0:0) -- (-150:0.5);
    \ThinLine (0:0) -- (-30:0.5);
    \draw[fill=black] (0:0) ellipse (0.05 and 0.05);
    \node at (-0.2,0.25) {$\xi^*$};
\end{tikzpicture}}}},\quad \widehat{\theta}_v\Ket{\vcenter{\hbox{\begin{tikzpicture}
    \ThinLine[orange] (90:0.5) -- (0:0);
    \ThinLine (0:0) -- (-150:0.5);
    \ThinLine (0:0) -- (-30:0.5);
    \draw[fill=black] (0:0) ellipse (0.05 and 0.05);
    \node at (-0.2,0.25) {$\eta$};
\end{tikzpicture}}}} = \Ket{\vcenter{\hbox{\begin{tikzpicture}
    \ThinLine[orange] (0:0) -- (-90:0.5);
    \ThinLine (0:0) -- (150:0.5);
    \ThinLine (0:0) -- (30:0.5);
    \draw[fill=black] (0:0) ellipse (0.05 and 0.05);
    \node at (-0.2,-0.25) {$\eta^*$};
\end{tikzpicture}}}}.
\end{equation}
This choice of $\widehat{\theta}_v$ defines an anti-linear $*$-isomorphism $\Theta$ from $\mathfrak{A}(\Gamma_{+})$ to $\mathfrak{A}(\Gamma_{-})$ as above. 

\begin{lemma}\label{lemma:: reflection exchanges L and R}
    Let $\Lambda_+ \subset \Gamma_+$ be a region with shape $(l,1)$ and let $\Lambda_- = \theta(\Lambda_+)$ be its reflection. 
    Then we have 
    \begin{equation}
        \Theta(\mathrm{L}_{\Lambda_+}(f)) = \mathrm{R}_{\Lambda_-}(f^*),\quad f\in \End(\rho^l). 
    \end{equation}
\end{lemma}
\begin{proof}
    Let $S$ denote the adjoint from $\Hom(\rho^l\gamma,\gamma\rho^l)$ to $\Hom(\gamma\rho^l,\rho^l\gamma)$, so $S$ is anti-unitary. 
    By the definition of $\widehat{\theta}_{\Lambda_+}$, for any $\ket{\xi}\in \mathcal{H}_{\Lambda_-}$, we have
    \begin{equation}\label{eqn:: reflection intertwines evaluation}
        \mathrm{eval}_{\Lambda_+}\widehat{\theta}^{-1}_{\Lambda_+}\ket{\xi} = \left(\mathrm{eval}_{\Lambda_-}\ket{\xi}\right)^* = S\, \mathrm{eval}_{\Lambda_-}\ket{\xi}.
    \end{equation}
    Since $\widehat{\theta}_{\Lambda_+}$ is also anti-unitary, taking adjoints in \eqref{eqn:: reflection intertwines evaluation} gives $\widehat{\theta}_{\Lambda_+}\mathrm{eval}_{\Lambda_+}^\dagger=\mathrm{eval}_{\Lambda_-}^\dagger S^{-1}$. 
    Moreover, for $h\in\Hom(\rho^l\gamma,\gamma\rho^l)$, we have $S^{-1}\pi_L(f)S(h)=\pi_R(f^*)h$. 
    Therefore, 
    \begin{equation}
        \begin{aligned}
            \Theta(\mathrm{L}_{\Lambda_+}(f))\ket{\xi}
            &= \widehat{\theta}_{\Lambda_+}\mathrm{eval}^\dagger_{\Lambda_+}\pi_L(f)\mathrm{eval}_{\Lambda_+}\widehat{\theta}^{-1}_{\Lambda_+}\ket{\xi}\\
            &= \mathrm{eval}^\dagger_{\Lambda_-}S^{-1}\pi_L(f)S\mathrm{eval}_{\Lambda_-}\ket{\xi}\\
            &= \mathrm{eval}^\dagger_{\Lambda_-}\pi_R(f^*)\mathrm{eval}_{\Lambda_-}\ket{\xi}\\
            &= \mathrm{R}_{\Lambda_-}(f^*)\ket{\xi}.
        \end{aligned}
    \end{equation}
    Since $\ket{\xi}$ is arbitrary, the result follows. 
\end{proof}
By symmetry, we also have 
\begin{equation}
    \Theta(\mathrm{R}_{\Lambda_+}(f)) = \mathrm{L}_{\Lambda_-}(f^*),\quad f\in \End(\rho^l). 
\end{equation}
\begin{proposition}
    The interaction $\varPhi$ defined in \eqref{eqn: interaction of new model} is reflection positive with respect to $\Theta$.
\end{proposition}
\begin{proof}
    For any finite reflection-symmetric region $\Lambda$, the local Hamiltonian decomposes as
    \begin{equation*}
        H_{\Lambda} = \mathrm{I}_{\Lambda_-}\otimes H_{\Lambda_+} + H_{\Lambda_-}\otimes \mathrm{I}_{\Lambda_+} + \sum_{\substack{ S\in \mathcal{P}_{k-1}(\Lambda) \\
        S\cap \Lambda_{\pm}\neq \emptyset }} F_S. 
    \end{equation*}
    For each $S\in \mathcal{P}_{k-1}(\Lambda_+)$, Lemma \ref{lemma:: reflection exchanges L and R} and the adjoint invariance of $\mathcal{G}$ give $\Theta(F_S) = F_{\theta(S)}$. 
    Every edge is either fully contained in $\Gamma_{+}$ or $\Gamma_{-}$, and we also have $\Theta(Q_e) = Q_{\theta(e)}$. 
    This proves that $H_{\Lambda_-} = \Theta(H_{\Lambda_+})$. 
    By Proposition \ref{prop:: creterion for RP}, it remains to show that $F_S$ has the required form for each $S\in \mathcal{P}_{k-1}(\Lambda)$ overlapping both $\Lambda_+$ and $\Lambda_-$. 
    For such $S$, denote by $S_\pm = S\cap \Lambda_\pm$. 
    Then for each $f\in \mathcal{G}$, we have 
    \begin{equation}\label{eqn: expand F_S}
        \begin{aligned}
            &\left\vert \mathrm{R}_{S_-} \left( \mathrm{D}_{\rho^k} f \left( \mathrm{D}_{\rho^k} \right)^{-1} \right) - \mathrm{L}_{S_+}\left( \left( \mathrm{D}_{\rho^k} \right)^{-1} f \mathrm{D}_{\rho^k} \right) \right\vert^2\\
            &= - \mathrm{R}_{S_-}(f^*)\otimes \mathrm{L}_{S_+}\left( \mathrm{D}^{-1}_{\rho^k}  f \mathrm{D}_{\rho^k} \right) +  \mathrm{L}_{S_+}\left( f^*  \mathrm{D}^{-1}_{\rho^k} f \mathrm{D}_{\rho^k}\right)\\
            & - \mathrm{R}_{S_-} \left(\mathrm{D}_{\rho^k} f \mathrm{D}^{-1}_{\rho^k} \right)\otimes \mathrm{L}_{S_+}\left( f^* \right) +  \mathrm{R}_{S_-}\left( \mathrm{D}_{\rho^k} f \mathrm{D}^{-1}_{\rho^k}f^*  \right)
        \end{aligned}
    \end{equation}
The second condition in Definition \ref{def:: generating set} ensures that every $f\in \mathcal{G}$ is an eigenvector under the conjugation by $\mathrm{D}_{\rho^k}$ with positive eigenvalue, hence 
\begin{equation}
    \begin{aligned}
        \mathrm{R}_{S_-}(f^*)\otimes \mathrm{L}_{S_+}\left( \mathrm{D}^{-1}_{\rho^k}  f \mathrm{D}_{\rho^k} \right) &= \mathrm{R}_{S_-} \left(\mathrm{D}^{1/2}_{\rho^k} f^* \mathrm{D}^{-1/2}_{\rho^k} \right)\otimes \mathrm{L}_{S_+}\left( \mathrm{D}^{-1/2}_{\rho^k}  f \mathrm{D}^{1/2}_{\rho^k} \right)\\
        &= \Theta\left( \mathrm{L}_{S_+}\left( \mathrm{D}^{-1/2}_{\rho^k}  f \mathrm{D}^{1/2}_{\rho^k} \right) \right)\otimes \mathrm{L}_{S_+}\left( \mathrm{D}^{-1/2}_{\rho^k}  f \mathrm{D}^{1/2}_{\rho^k} \right),
    \end{aligned}
\end{equation}
where the second equality follows from Lemma \ref{lemma:: reflection exchanges L and R}. 
Similarly, 
\begin{equation}
    \mathrm{R}_{S_-} \left(\mathrm{D}_{\rho^k} f \mathrm{D}^{-1}_{\rho^k} \right)\otimes \mathrm{L}_{S_+}\left( f^* \right) = \Theta\left( \mathrm{L}_{S_+}\left( \mathrm{D}^{-1/2}_{\rho^k}  f^* \mathrm{D}^{1/2}_{\rho^k} \right) \right)\otimes \mathrm{L}_{S_+}\left( \mathrm{D}^{-1/2}_{\rho^k}  f^* \mathrm{D}^{1/2}_{\rho^k} \right). 
\end{equation}
Since $\mathcal{G}$ is invariant under taking adjoint, 
\begin{equation}
    \begin{aligned}
        &- \sum_{f\in \mathcal{G}}\mathrm{R}_{S_-}(f^*)\otimes \mathrm{L}_{S_+}\left( \mathrm{D}^{-1}_{\rho^k}  f \mathrm{D}_{\rho^k} \right) - \sum_{f\in \mathcal{G}}\mathrm{R}_{S_-} \left(\mathrm{D}_{\rho^k} f \mathrm{D}^{-1}_{\rho^k} \right)\otimes \mathrm{L}_{S_+}\left( f^* \right)\\
        &= - 2 \sum_{f\in \mathcal{G}}\Theta\left( \mathrm{L}_{S_+}\left( \mathrm{D}^{-1/2}_{\rho^k}  f \mathrm{D}^{1/2}_{\rho^k} \right) \right)\otimes \mathrm{L}_{S_+}\left( \mathrm{D}^{-1/2}_{\rho^k}  f \mathrm{D}^{1/2}_{\rho^k} \right). 
    \end{aligned}
\end{equation}
Again by the invariance of $\mathcal{G}$ under adjoint:
\begin{equation}
    \Theta\left( \sum_{f\in \mathcal{G}}\mathrm{L}_{S_+}\left( f^*  \mathrm{D}^{-1}_{\rho^k} f \mathrm{D}_{\rho^k}\right)  \right) =  \sum_{f\in \mathcal{G}}\mathrm{R}_{S_-}\left( \mathrm{D}_{\rho^k}f \mathrm{D}^{-1}_{\rho^k}f^*\right).
\end{equation}
Thus, we finally obtain:
\begin{equation}\label{eqn:: reflection positive expansion of F_S}
    \begin{aligned}
        &\sum_{f\in \mathcal{G}}\left\vert \mathrm{R}_{S_-} \left( \mathrm{D}_{\rho^k} f \left( \mathrm{D}_{\rho^k} \right)^{-1} \right) - \mathrm{L}_{S_+}\left( \left( \mathrm{D}_{\rho^k} \right)^{-1} f \mathrm{D}_{\rho^k} \right) \right\vert^2 \\
        &= \sum_{f\in \mathcal{G}}\mathrm{L}_{S_+}\left( f^*  \mathrm{D}^{-1}_{\rho^k} f \mathrm{D}_{\rho^k} \right) + \Theta\left( \sum_{f\in \mathcal{G}}\mathrm{L}_{S_+}\left( f^*  \mathrm{D}^{-1}_{\rho^k} f \mathrm{D}_{\rho^k} \right) \right) \\
        & -2\sum_{f\in \mathcal{G}}\Theta\left( \mathrm{L}_{S_+}\left( \mathrm{D}^{-1/2}_{\rho^k}  f \mathrm{D}^{1/2}_{\rho^k} \right) \right)\otimes \mathrm{L}_{S_+}\left( \mathrm{D}^{-1/2}_{\rho^k}  f \mathrm{D}^{1/2}_{\rho^k} \right),
    \end{aligned}
\end{equation}
This has the form required by Proposition \ref{prop:: creterion for RP}, completing the proof. 
\end{proof}

\subsection{Realizing fusion spin chains}

To construct the boundary algebra, we recall the Osterwalder-Schrader (OS) reconstruction for reflection positive, frustration-free models from \cite[Section 5]{LiuZhao2025RPTO}. 
Consider a finite, reflection-symmetric region $\Lambda = \Lambda_{-} \sqcup \Lambda_{+}$. 
We define the entanglement support of $H_{\Lambda}$ as the range projection of the reduced density matrix on $\Lambda_{+}$:
\begin{equation}
    \widehat{\Pi}(\Lambda) = \range(\Tr_{\mathcal{H}_{\Lambda_{-}}}(\Pi(\Lambda))) \in \mathfrak{A}(\Lambda_{+}).
\end{equation}
The reflection positivity of $\Pi(\Lambda)$ ensures that the sesquilinear form on $\mathfrak{A}(\Lambda_{+})$ defined by
\begin{equation}
    \braket{Y, X}_{0} = \frac{1}{\Tr(\Pi(\Lambda))} \Tr(\Pi(\Lambda)(\Theta(Y) \otimes X))
\end{equation}
is positive semi-definite on $\mathfrak{A}(\Lambda_{+})$. The quotient space $\mathfrak{H}_{\Lambda} = \mathfrak{A}(\Lambda_{+}) / \ker\braket{\cdot, \cdot}_{0}$ forms a Hilbert space. Left multiplication by an operator $T \in \mathfrak{A}(\Lambda_{+})$ that preserves the kernel induces a well-defined operator $\Phi(T)$ on $\mathfrak{H}_{\Lambda}$ via $\Phi(T)\psi(X) = \psi(TX)$, where $\psi$ is the quotient map. The OS reconstruction yields a finite-dimensional $C^*$-algebra $\mathcal{M}_{\Lambda}$ generated by all such well-defined $\Phi(T)$. 
% Note that since only the ground state projection is involved, this construction does not not require the local Hamiltonians to be reflection positive. 

In \cite{LiuZhao2025RPTO} it was proved that $\mathcal{M}_{\Lambda}$ is canonically isomorphic to a corner of the interaction algebra \cite{zanardi2000StabilizingQuantumInformation} of the ground state projection $\Pi(\Lambda)$. 
Given a local operator $O$, and a subset $D\subseteq\Gamma$, the interaction algebra $\mathcal{A}_{D}(O)$ is defined as the $C^*$-algebra generated by all operators supported on $D$ that is of the form $\Tr_{\overline{D}}(QO)$, where $Q$ is an operator supported on $\overline{D}$. 
Alternatively, defining the local symmetry algebra as
\begin{equation}
    \mathrm{Sym}_{\Lambda_{+}}(H_{\Lambda}) = \{X \in \mathfrak{A}(\Lambda_{+}) : [I_{\Lambda_{-}} \otimes X, H_{\Lambda}] = [I_{\Lambda_{-}} \otimes X^\dagger, H_{\Lambda}] = 0 \},
\end{equation}
then the interaction algebra is given by the relative commutant $\mathcal{A}_{+}(H_{\Lambda}) = \mathrm{Sym}_{\Lambda_{+}}(H_{\Lambda})' \cap \mathfrak{A}(\Lambda_{+})$. 
In \cite[Proposition 3.17]{LiuZhao2025RPTO}, it is shown that the entanglement support $\widehat{\Pi}(\Lambda)$ is central in both $\mathrm{Sym}_{\Lambda_+}(H_{\Lambda})$ and $\mathrm{Sym}_{\Lambda_+}(\Pi(\Lambda))$, and moreover 
\begin{equation}
    \mathrm{Sym}_{\Lambda_+}(H_{\Lambda}) \widehat{\Pi}(\Lambda)= \mathrm{Sym}_{\Lambda_+}(\Pi(\Lambda))\widehat{\Pi}(\Lambda). 
\end{equation}
Taking commutant, we obtain:
\begin{equation}\label{eqn:: reduce interaction algebra to entanglement support}
    \mathcal{A}_{+}(H_{\Lambda})\widehat{\Pi}(\Lambda) = \mathcal{A}_{+}(\Pi(\Lambda))\widehat{\Pi}(\Lambda)
\end{equation}
Moreover, there is a $*$-isomorphism from $\mathcal{A}_{+}(H_{\Lambda})\widehat{\Pi}(\Lambda)$ to $\mathcal{M}_{\Lambda}$. 
Due to this result, we will base our discussions around $\mathcal{A}_{+}(\Pi(\Lambda))\widehat{\Pi}(\Lambda)$. 

Consider $\mathcal{X}$ to be a upward-directed set of finite symmetric regions in $\Gamma$. 
By a local net of $C^*$-algebras over $\mathcal{X}$, we mean a family of $C^*$-algebras $\{\mathcal{A}_{\Lambda}\}_{\Lambda\in \mathcal{X}}$, equipped with unital $*$-inclusions $\iota_{\Lambda_2,\Lambda_1}: \mathcal{A}_{\Lambda_1}\rightarrow \mathcal{A}_{\Lambda_2}$ for each pair $\Lambda_1\subseteq\Lambda_2$, such that 
\begin{description}
    \item[Isotony] For each $\Lambda_1\subseteq \Lambda_2\subseteq \Lambda_3$, $\iota_{\Lambda_3,\Lambda_1} = \iota_{\Lambda_3,\Lambda_2}\circ \iota_{\Lambda_2 ,\Lambda_1}$;
    \item[Locality] For each $\Lambda_1,\Lambda_2\subseteq \Lambda_3$ with $\Lambda_1\cap \Lambda_2 = \emptyset$, $[\iota_{\Lambda_3,\Lambda_1}(\mathcal{A}_{\Lambda_1}),\iota_{\Lambda_3,\Lambda_2}(\mathcal{A}_{\Lambda_2})] = 0$. 
\end{description}

With these data, one can define a quasi local $C^*$-algebra $\mathcal{A}$ is defined as the norm-closure of the inductive limit: 
\begin{equation}
    \mathcal{A} = \overline{\varinjlim_{\Lambda\in \mathcal{X}}\mathcal{A}_{\Lambda}}^{\|\cdot \|}. 
\end{equation}
which fits into the definition of discrete net of $C^*$-algebra in \cite{Jones2024DHRBimodules}. 

In \cite[Section 5]{LiuZhao2025RPTO}, it was shown that using the frustration-free property of the interaction, one can organize the family $\{\mathcal{A}_{+}(H_{\Lambda})\widehat{\Pi}(\Lambda)\}_{\Lambda\in \mathcal{X}}$ into a local net of $C^*$-algebras over $\mathcal{X}$. 
For any pair $\Lambda_1 \subseteq \Lambda_2\in \mathcal{X}$, it was shown that  
\begin{equation}
    [\mathcal{A}_{+}(H_{\Lambda_1})\widehat{\Pi}(\Lambda_1), \widehat{\Pi}(\Lambda_2)] = 0.
\end{equation}
This allows us to define the $*$-homomorphisms $\iota_{\Lambda_2,\Lambda_1}: \mathcal{A}_{+}(H_{\Lambda_1})\widehat{\Pi}(\Lambda_1) \rightarrow \mathcal{A}_{+}(H_{\Lambda_2})\widehat{\Pi}(\Lambda_2)$ via:
\begin{equation}\label{eqn: inclusion map for the boundary algebra}
    \iota_{\Lambda_2,\Lambda_1}(x) = (x \otimes \mathrm{I}_{\Lambda_{2,+} \setminus \Lambda_{1,+}})\widehat{\Pi}(\Lambda_2),\quad x \in \mathcal{A}_{+}(H_{\Lambda_1})\widehat{\Pi}(\Lambda_1). 
\end{equation}
The maps $\iota_{\Lambda_2,\Lambda_1}$ satisfy the isotony property: $\iota_{\Lambda_3,\Lambda_2} \circ \iota_{\Lambda_2,\Lambda_1} = \iota_{\Lambda_3,\Lambda_1}$. 
The map $\iota_{\Lambda_2,\Lambda_1}$ is injective when the family of ground state projections $\{\Pi(\Lambda)\}_{\Lambda\in \mathcal{X}}$ satisfies the following \emph{extendability condition}:
\begin{align}
    \range(\Tr_{\overline{\Lambda'_+}}\widehat{\Pi}(\Lambda)) = \widehat{\Pi}(\Lambda'),\quad \Lambda'\subseteq \Lambda.
\end{align}
Moreover, this net satisfies locality in the sense that 
\begin{equation}
    [\iota_{\Lambda_3,\Lambda_1}(\mathcal{A}_{\Lambda_1}),\iota_{\Lambda_3,\Lambda_2}(\mathcal{A}_{\Lambda_2})],
\end{equation} 
for any symmetric regions $\Lambda_1,\Lambda_2 \subseteq \Lambda_3$ with $\Lambda_1 \cap \Lambda_2 = \emptyset$. 

Another feature of the above construction is that the local modular flow extends consistently to a global one on the quasi-local algebra. 
For $\Lambda\in \mathcal{X}$, maximally mixed state in the ground state subspace of $H_{\Lambda}$ defines a the vacuum state on $\mathcal{A}_{+}(\Pi(\Lambda))\widehat{\Pi}(\Lambda)$ as
\begin{equation}
    \omega_{\Lambda}(x) = \frac{1}{\Tr(\Pi(\Lambda))}\Tr(\left( x\otimes \mathrm{I}_{\Lambda_-} \right)\Pi(\Lambda)),\quad x\in \mathcal{A}_{+}(\Pi(\Lambda))\widehat{\Pi}(\Lambda). 
\end{equation}

\begin{theorem}\cite[Theorem 5.21]{LiuZhao2025RPTO}\label{thm:: dynamics in the inductive limit}
    For each $\Lambda\in \mathcal{X}$, let $\sigma^{\Lambda}_t$ denote the modular automorphism group associated with the vacuum state $\omega_{\Lambda}$ on $\mathcal{A}_{+}(\Pi(\Lambda))\widehat{\Pi}(\Lambda)$.
    Then for all $\Lambda_1,\Lambda_2\in \mathcal{X}$ satisfying $\Lambda_1\subseteq \Lambda_2$, we have 
    \begin{equation}
        \sigma^{\Lambda_2}_t\circ \iota_{\Lambda_2,\Lambda_1} = \iota_{\Lambda_2,\Lambda_1}\circ \sigma^{\Lambda_1}_t,\quad \forall t\in\mathbb{R}. 
    \end{equation}
\end{theorem}

Back to our model, we now explicitly compute the boundary net produced by OS reconstruction. 
We consider a symmetric regions $\Lambda = \Lambda_-\cup \Lambda_+$ depicted in Fig. \ref{fig:: reflection symmetric region}, such that $\Lambda_+$ is a region of shape $(l,h)$ with the bottom boundary lying on the reflection line: 
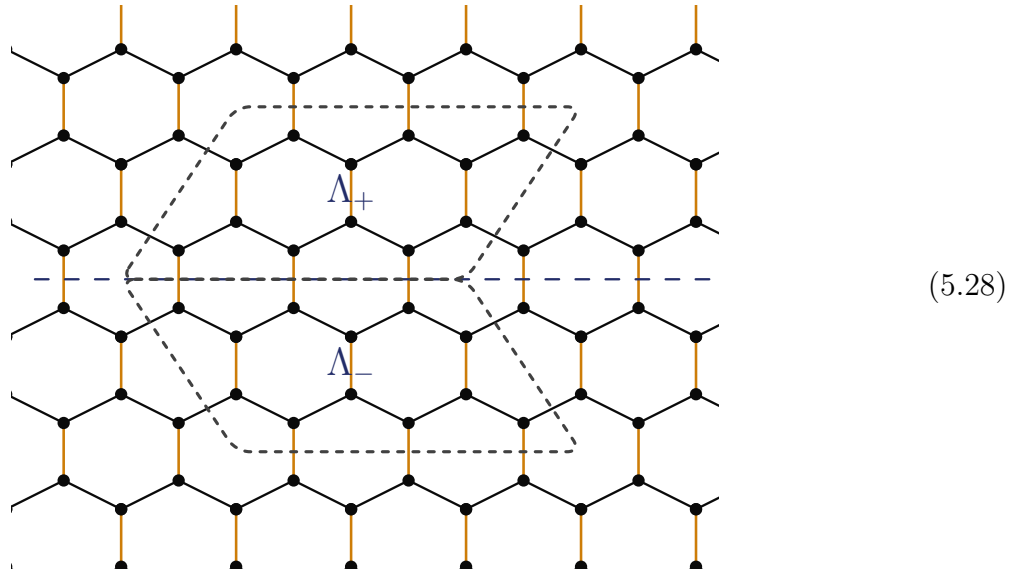
\begin{align}\label{fig:: reflection symmetric region}
    \vcenter{\hbox{%
\resizebox{0.55\columnwidth}{!}{%
\begin{tikzpicture}[
    x=0.74cm,
    y=0.74cm,
    line cap=round,
    line join=round
]
\definecolor{latticeblack}{RGB}{10,10,10}
\definecolor{edgeorange}{RGB}{205,125,8}
\definecolor{boundarygray}{RGB}{65,65,65}
\definecolor{separatorblue}{RGB}{35,45,105}
% Fixed clipping and bounding window.
\path[use as bounding box]
    (-0.30,-0.70) rectangle (12.00,9.10);
\clip
    (-0.30,-0.70) rectangle (12.00,9.10);
% Shift the complete picture to the right and upward.
% The leftmost and bottom orange edges are thereby exposed further.
\begin{scope}[xshift=0.45cm,yshift=0.25cm]
% =============================================================
% Orange vertical bonds
% =============================================================
\foreach \r in {-1,0,1,2,3} {
    \pgfmathsetmacro{\yb}{3*\r}
    % Vertical bonds in even columns.
    \foreach \x in {-2,0,2,4,6,8,10,12,14} {
        \draw[
            edgeorange,
            line width=0.95pt
        ]
            (\x,{\yb+0.5})
            --
            (\x,{\yb+1.5});
    }
    % Vertical bonds in odd columns.
    \foreach \x in {-3,-1,1,3,5,7,9,11,13,15} {
        \draw[
            edgeorange,
            line width=0.95pt
        ]
            (\x,{\yb+2})
            --
            (\x,{\yb+3});
    }
}
% =============================================================
% Black diagonal bonds
% =============================================================
\foreach \r in {-1,0,1,2,3} {
    \pgfmathsetmacro{\yb}{3*\r}
    \foreach \x in {-3,-1,1,3,5,7,9,11,13,15} {
        \draw[
            latticeblack,
            line width=0.85pt
        ]
            (\x,\yb)
            --
            ({\x-1},{\yb+0.5});
        \draw[
            latticeblack,
            line width=0.85pt
        ]
            (\x,\yb)
            --
            ({\x+1},{\yb+0.5});
        \draw[
            latticeblack,
            line width=0.85pt
        ]
            (\x,{\yb+2})
            --
            ({\x-1},{\yb+1.5});
        \draw[
            latticeblack,
            line width=0.85pt
        ]
            (\x,{\yb+2})
            --
            ({\x+1},{\yb+1.5});
    }
}
% =============================================================
% Black lattice vertices
% =============================================================
\foreach \r in {-1,0,1,2,3} {
    \pgfmathsetmacro{\yb}{3*\r}
    % Odd-column vertices.
    \foreach \x in {-3,-1,1,3,5,7,9,11,13,15} {
        \fill[latticeblack]
            (\x,\yb)
            circle[radius=2.25pt];
        \fill[latticeblack]
            (\x,{\yb+2})
            circle[radius=2.25pt];
        \fill[latticeblack]
            (\x,{\yb+3})
            circle[radius=2.25pt];
    }
    % Even-column vertices.
    \foreach \x in {-2,0,2,4,6,8,10,12,14} {
        \fill[latticeblack]
            (\x,{\yb+0.5})
            circle[radius=2.25pt];
        \fill[latticeblack]
            (\x,{\yb+1.5})
            circle[radius=2.25pt];
    }
}
% =============================================================
% Horizontal reflection line
% =============================================================
% The line y=4 bisects the orange vertical bonds running from
% y=3.5 to y=4.5.
\draw[
    separatorblue,
    line width=0.85pt,
    dash pattern=on 5pt off 6pt
]
    (-0.50,4.00)
    --
    (12.50,4.00);
% =============================================================
% Parallelogram coordinates
% =============================================================
% Every boundary passes through the midpoints of the lattice
% edges that it intersects.
\coordinate (L)  at (1.00,4.00);
\coordinate (R)  at (7.00,4.00);
\coordinate (UL) at (3.00,7.00);
\coordinate (UR) at (9.00,7.00);
\coordinate (LL) at (3.00,1.00);
\coordinate (LR) at (9.00,1.00);
% =============================================================
% Upper dark-gray dashed parallelogram
% =============================================================
\draw[
    boundarygray,
    line width=1.15pt,
    dash pattern=on 2.2pt off 3.4pt,
    rounded corners=5pt
]
    (L)
    --
    (UL)
    --
    (UR)
    --
    (R)
    --
    cycle;
% =============================================================
% Lower dark-gray dashed parallelogram
% =============================================================
\draw[
    boundarygray,
    line width=1.15pt,
    dash pattern=on 2.2pt off 3.4pt,
    rounded corners=5pt
]
    (L)
    --
    (R)
    --
    (LR)
    --
    (LL)
    --
    cycle;
% =============================================================
% Region labels
% =============================================================
\node[
    text=separatorblue,
    font=\large
] at (5.00,5.50)
    {$\Lambda_{+}$};
\node[
    text=separatorblue,
    font=\large
] at (5.00,2.50)
    {$\Lambda_{-}$};
\end{scope}
\end{tikzpicture}%
}%
}}
\end{align}

\begin{lemma}\label{lemma:: ground state projection of a reflection symmetric region}
    Consider a parallelogram $\Lambda_+$ of shape $(l,h)$, with the base point on the reflection line, define $\Lambda_- = \theta(\Lambda_+)$ and $\Lambda = \Lambda_-\cup \Lambda_+$ as shown in \eqref{fig:: reflection symmetric region}. 
    Denote by $\Lambda_h$ the region with shape $(l,1)$ above the reflection line. 
    Then for $l\geq k$, we have: 
    \begin{equation}
        \Pi(\Lambda) = \sum_{\eta\in \mathrm{ONB}(\rho^l)} \Pi(\Lambda_-)\Theta\left( \mathrm{L}_{\Lambda_h}\left( \Delta^{-1/2}_l\mathrm{D}^{-1/2}_{\rho^l} \eta \mathrm{D}^{1/2}_{\rho^l} \right) \right)\otimes \mathrm{L}_{\Lambda_h}\left( \Delta^{-1/2}_l\mathrm{D}^{-1/2}_{\rho^l} \eta \mathrm{D}^{1/2}_{\rho^l} \right)\Pi(\Lambda_+). 
    \end{equation}
    For $l<k$, we have $\Pi(\Lambda) = \Pi(\Lambda_-)\otimes \Pi(\Lambda_+)$. 
\end{lemma}
\begin{proof}
    Suppose first $l\geq k$. 
    Set $\Lambda_0=\theta(\Lambda_h)\cup\Lambda_h$, and write $\mathrm{D}=\mathrm{D}_{\rho^l}$.
    By Lemma \ref{lemma:: key central element}, $\Delta_l$ is strictly positive and central in $\End(\rho^l)$.
    Define $T:\End(\rho^l)\to\End(\rho^l)$ by
    \begin{equation}
        T(\eta)=\Delta_l^{-1/2}\mathrm{D}^{-1/2}\eta\mathrm{D}^{1/2}.
    \end{equation}
    With respect to the inner product on $\End(\rho^l)$ in \eqref{eqn:: resolution with respect to ONB}, cyclicity of $\tr_{\mathcal C}$ and centrality of $\Delta_l$ give
    \begin{equation}
        \braket{\xi,T(\eta)}
        =\tr_{\mathcal C}\bigl(\xi^*\Delta_l^{-1/2}\mathrm{D}^{1/2}\eta\mathrm{D}^{3/2}\bigr)
        =\braket{T(\xi),\eta},\qquad \xi,\eta\in\End(\rho^l).
    \end{equation}
    Thus $T$ is self-adjoint. Applying the orthonormal-basis resolution twice yields
    \begin{equation}
        \begin{aligned}
            \sum_{\eta\in\mathrm{ONB}(\rho^l)}T(\eta)^*\otimes T(\eta)
            &=\sum_{\eta\in\mathrm{ONB}(\rho^l)}T^2(\eta)^*\otimes\eta\\
            &=\sum_{\eta\in\mathrm{ONB}(\rho^l)}\mathrm{D}\eta^*\mathrm{D}^{-1}\Delta_l^{-1}\otimes\eta.
        \end{aligned}
    \end{equation}
    Proposition \ref{prop:: local ground states in a row}, \eqref{eqn:: def. of chi-operator}, and the fact that $\Delta_l$ is central in $\End(\rho^l)$ therefore give
    \begin{equation}
        \begin{aligned}
            \Pi(\Lambda_0)
            &=\sum_{\eta\in\mathrm{ONB}(\rho^l)}
                \mathrm{R}_{\theta(\Lambda_h)}\bigl(\mathrm{D}\eta^*\mathrm{D}^{-1}\Delta_l^{-1}\bigr)
                \otimes\mathrm{L}_{\Lambda_h}(\eta)\\
            &=\sum_{\eta\in\mathrm{ONB}(\rho^l)}
                \mathrm{R}_{\theta(\Lambda_h)}\bigl(T(\eta)^*\bigr)
                \otimes\mathrm{L}_{\Lambda_h}\bigl(T(\eta)\bigr)\\
            &=\sum_{\eta\in\mathrm{ONB}(\rho^l)}
                \Theta\bigl(\mathrm{L}_{\Lambda_h}(T(\eta))\bigr)
                \otimes\mathrm{L}_{\Lambda_h}(T(\eta)),
        \end{aligned}
    \end{equation}
    where we uses Lemma \ref{lemma:: reflection exchanges L and R} the last equality.

    Now partition $\Lambda$ into its $2h$ rows, with $\Lambda_0$ the adjacent pair crossing the reflection line.
    By the expression of ground state projection in a row \eqref{eqn:: ground state projection in a row} when $h=1$, and by Theorem \ref{thm:: local ground states in a parallelogram} when $h\geq 2$, the products of the adjacent-pair ground-state projections within $\Lambda_-$ and $\Lambda_+$ are $\Pi(\Lambda_-)$ and $\Pi(\Lambda_+)$, respectively, and hence
    \begin{equation}
        \Pi(\Lambda)=\Pi(\Lambda_-)\wedge \Pi(\Lambda_0)\wedge \Pi(\Lambda_+).
    \end{equation}
    By Proposition \ref{prop:: local ground states in a row}, the three projections commute mutually. 
    Thus $\Pi(\Lambda) = \Pi(\Lambda_-)\Pi(\Lambda_0)\Pi(\Lambda_+)$. 
    Substituting the preceding expression for $\Pi(\Lambda_0)$ proves the claim. 

    For $l<k$, no operator $F_S$ is supported in $\Lambda$.
    Hence $H_{\Lambda}=H_{\Lambda_-}\otimes \mathrm{I}_{\Lambda_+}+\mathrm{I}_{\Lambda_-}\otimes H_{\Lambda_+}$, and taking the ground-state projection gives $\Pi(\Lambda)=\Pi(\Lambda_-)\otimes\Pi(\Lambda_+)$.
\end{proof}

\begin{corollary}\label{corollary:: entanglement support as local ground state projection}
    For a reflection symmetric region $\Lambda = \Lambda_-\cup \Lambda_+$ as described in Lemma \ref{lemma:: ground state projection of a reflection symmetric region}, we have 
    \begin{equation}
        \widehat{\Pi}(\Lambda) = \Pi(\Lambda_+).
    \end{equation}
\end{corollary}
\begin{proof}
    By Corollary \ref{corollary:: frustration-freeness}, the Hamiltonian $H_{\Lambda}$ is frustration-free.
    Thus $\Pi(\Lambda)$ is a subprojection of $\mathrm{I}_{\Lambda_-}\otimes \Pi(\Lambda_+)$, so $\Pi(\Lambda_+)\Tr_{\Lambda_-}(\Pi(\Lambda))\Pi(\Lambda_+) = \Tr_{\Lambda_-}(\Pi(\Lambda))$.
    Hence $\widehat{\Pi}(\Lambda)\leq \Pi(\Lambda_+)$. 
    For the converse, if $l<k$, the preceding lemma gives $\Pi(\Lambda)=\Pi(\Lambda_-)\otimes\Pi(\Lambda_+)$. 
    Since $\Pi(\Lambda_-)$ is nonzero, the partial trace of this product has support $\Pi(\Lambda_+)$.

    Now suppose $l\geq k$, and let $\ket{\psi}\neq 0$ be a ground state of $H_{\Lambda_+}$.
    By Lemma \ref{lemma:: ground state projection of a reflection symmetric region},
    \begin{equation}
        \Braket{\widehat{\theta}_{\Lambda_+}\psi\otimes\psi|\Pi(\Lambda)|\widehat{\theta}_{\Lambda_+}\psi\otimes\psi}
        =\sum_{\eta\in\mathrm{ONB}(\rho^l)}\left|\Braket{\psi|\mathrm{L}_{\Lambda_h}(T(\eta))|\psi}\right|^2,
    \end{equation}
    where $T$ is defined in the proof of the preceding lemma.
    The right-hand side is strictly positive, since the elements $T(\eta)$ span $\End(\rho^l)$ and $\Braket{\psi|\mathrm{L}_{\Lambda_h}(\mathrm{Id}_{\rho^l})|\psi}=\|\psi\|^2>0$.
    Thus $\Tr_{\Lambda_-}(\Pi(\Lambda))$ is strictly positive on $\Pi(\Lambda_+)$, and we have $\widehat{\Pi}(\Lambda) = \Pi(\Lambda_+)$.
\end{proof}

\begin{proposition}\label{prop:: endomorphism algebras as boundary algebra}
    Let $\Lambda$ be as in Lemma \ref{lemma:: ground state projection of a reflection symmetric region}. 
    Then for $l\geq k$, the $C*$-algebra $\mathcal{A}_{+}(\Pi(\Lambda))\widehat{\Pi}(\Lambda)$ associated with $\Lambda$ is isomorphic to the $\End(\rho^l)$, with the $*$-isomorphism given by
    \begin{equation}
        \pi_0: f\mapsto \mathrm{L}_{\Lambda_h}\left(\mathrm{D}_{\rho^l}^{-1/2}f\mathrm{D}_{\rho^l}^{1/2}\right)\Pi(\Lambda_+),\quad f\in \End(\rho^l). 
    \end{equation}
    Whereas for $l<k$, we have $\mathcal{A}_{+}(\Pi(\Lambda))\widehat{\Pi}(\Lambda)=\mathbb{C}\Pi(\Lambda_+)$. 
\end{proposition}
\begin{proof}
    Suppose first that $l\geq k$, and write $\mathrm{D}=\mathrm{D}_{\rho^l}$.
    For every $f\in\End(\rho^l)$, the operator $\mathrm{L}_{\Lambda_h}(f)$ commutes with $\Pi(\Lambda_+)$.
    When $h=1$, this follows from $\Pi(\Lambda_+)=\mathrm{L}_{\Lambda_h}(\mathrm{Id}_{\rho^l})$; when $h\geq2$, the adjacent-pair factorization in Theorem \ref{thm:: local ground states in a parallelogram} acts on $\Lambda_h$ only through right actions, which commute with left actions by Corollary \ref{corollary:: commutativity between left/right actions}.
    Moreover, $\mathrm{L}_{\Lambda_h}(\mathrm{Id}_{\rho^l})\Pi(\Lambda_+)=\Pi(\Lambda_+)$ because every upper-half ground state belongs to the bottom-row string-net subspace.
    Thus, for $f,g\in\End(\rho^l)$, the row-action formulas give
    \begin{equation}
        \begin{aligned}
            \mathrm{L}_{\Lambda_h}(f)\Pi(\Lambda_+)\mathrm{L}_{\Lambda_h}(g)\Pi(\Lambda_+)&=\mathrm{L}_{\Lambda_h}(fg)\Pi(\Lambda_+),\\
            \bigl(\mathrm{L}_{\Lambda_h}(f)\Pi(\Lambda_+)\bigr)^\dagger&=\mathrm{L}_{\Lambda_h}(\mathrm{D}^{-1}f^*\mathrm{D})\Pi(\Lambda_+).
        \end{aligned}
    \end{equation}

    The compressed action is faithful.
    For $h=1$, the surjective row evaluation intertwines it with left multiplication on $\Hom(\gamma\rho^l,\rho^l\gamma)$, whose tensor-unit summand contains $\mathrm{Id}_{\rho^l}$.
    For $h\geq2$, the isometry $\mathcal{U}_{h:1}$ in Theorem \ref{thm:: local ground states in a parallelogram} has the same intertwining property, since $\mathrm{L}_{\Lambda_h}(f)$ commutes with the adjacent-pair operators meeting $\Lambda_h$.
    Its image contains $\mathrm{Id}_{\rho^l}$ in the tensor-unit summand of $\gamma^h$: choose every row evaluation to be the identity there.
    In either case, applying the compressed action of $f$ to a preimage of this identity gives $f$ in the same summand, so the action is faithful.

    Let $T(f)=\Delta_l^{-1/2}\mathrm{D}^{-1/2}f\mathrm{D}^{1/2}$, as in the proof of Lemma \ref{lemma:: ground state projection of a reflection symmetric region}.
    Since reflection sends $\Pi(\Lambda_+)$ to $\Pi(\Lambda_-)$, commutation with the row action rewrites that lemma as
    \begin{equation}
        \Pi(\Lambda)=\sum_{\eta\in\mathrm{ONB}(\rho^l)}\Theta\bigl(\mathrm{L}_{\Lambda_h}(T(\eta))\Pi(\Lambda_+)\bigr)\otimes\mathrm{L}_{\Lambda_h}(T(\eta))\Pi(\Lambda_+).
    \end{equation}
    The map $T$ is invertible, so faithfulness shows that the reflected operators $\Theta\bigl(\mathrm{L}_{\Lambda_h}(T(\eta))\Pi(\Lambda_+)\bigr)$ are linearly independent.
    Hence the nondegenerate trace pairing lets $Q\in\mathfrak{A}(\Lambda_-)$ prescribe their coefficients independently in
    \begin{equation}
        \Tr_{\Lambda_-}\bigl((Q\otimes\mathrm{I}_{\Lambda_+})\Pi(\Lambda)\bigr)
        =\sum_{\eta\in\mathrm{ONB}(\rho^l)}\Tr\bigl(Q\Theta(\mathrm{L}_{\Lambda_h}(T(\eta))\Pi(\Lambda_+))\bigr)\mathrm{L}_{\Lambda_h}(T(\eta))\Pi(\Lambda_+).
    \end{equation}
    These partial traces therefore span $\{\mathrm{L}_{\Lambda_h}(f)\Pi(\Lambda_+):f\in\End(\rho^l)\}$.
    This span is a $C^*$-algebra with unit $\Pi(\Lambda_+)$ by the formulas above, and Corollary \ref{corollary:: entanglement support as local ground state projection} gives
    \begin{equation}
        \mathcal{A}_+(\Pi(\Lambda))\widehat{\Pi}(\Lambda)=\{\mathrm{L}_{\Lambda_h}(f)\Pi(\Lambda_+):f\in\End(\rho^l)\}.
    \end{equation}
    Equation \eqref{eqn:: reduce interaction algebra to entanglement support} identifies this with the boundary algebra of $H_\Lambda$.
    Finally, the map $\pi_0$ in the statement is multiplicative, faithful, and onto this algebra, and taking $\mathrm{D}^{-1/2}f\mathrm{D}^{1/2}$ in the adjoint formula above shows that $\pi_0(f)^\dagger=\pi_0(f^*)$.

    For $l<k$, Lemma \ref{lemma:: ground state projection of a reflection symmetric region} gives $\Pi(\Lambda)=\Pi(\Lambda_-)\otimes\Pi(\Lambda_+)$. Hence, for every $Q\in\mathfrak{A}(\Lambda_-)$,
    \begin{equation}
        \Tr_{\Lambda_-}\bigl((Q\otimes\mathrm{I}_{\Lambda_+})\Pi(\Lambda)\bigr)
        =\Tr\bigl(Q\Pi(\Lambda_-)\bigr)\Pi(\Lambda_+).
    \end{equation}
    Taking $Q=\mathrm{I}_{\Lambda_-}$ shows that $\widehat{\Pi}(\Lambda)=\Pi(\Lambda_+)$. 
    Thus, by the definition of the interaction algebra, $\mathcal{A}_+(\Pi(\Lambda))\widehat{\Pi}(\Lambda)=\mathbb{C}\Pi(\Lambda_+)$.
\end{proof}
\begin{remark}
    Define $\Lambda_0 = \theta(\Lambda_h)\cup \Lambda_h$. 
    Notice that the above Proposition implies that the map $\mathcal{A}_+(\Pi(\Lambda_0))\widehat{\Pi}(\Lambda_0)\rightarrow \mathcal{A}_+(\Pi(\Lambda))\widehat{\Pi}(\Lambda) $ defined by compressing using $\widehat{\Pi}(\Lambda)=\Pi(\Lambda_+)$ is a $*$-\emph{isomorphism}. 
    This doesn't happen by chance, and is  rooted in the locality of the interaction, as shown in \cite[Theorem 5.25]{LiuZhao2025RPTO}. 
    In this sense, the non-trivial operators in the reconstructed local net appear only near the boundary, hence the name "boundary algebra".  
\end{remark}

Now we recall fusion spin chains. 
Denote by the set of finite intervals in $\mathbb{Z}$ by $\mathcal{I}(\mathbb{Z})$. 
Given an object $\rho$ in a unitary fusion category $\mathcal{C}$, 
the fusion spin chain is a local net of $C^*$-algebras $\{\mathcal{A}(\mathcal{C}, \rho)_I\}_{I\in \mathcal{I}(\mathbb{Z})}$ over $\mathcal{I}(\mathbb{Z})$. 
For the interval $I=[a,b]$, $a\leq b$, the algebra $\mathcal{A}(\mathcal{C}, \rho)_I$ is defined as 
\begin{equation}
    \mathcal{A}(\mathcal{C}, \rho)_I = \End(\rho^{ (b-a+1)}). 
\end{equation}
For $I=[a,b]\subseteq J=[c,d]$, the inclusion is given by tensoring with the identity morphism of tensor powers of $\rho$ out side $I$:
\begin{equation}
    \End(\rho^{ (b-a+1)}) \ni x \mapsto \mathrm{Id}_{\rho^{a-c}}\otimes x \otimes \mathrm{Id}_{\rho^{d-b}} \in \End(\rho^{ (d-c+1)}). 
\end{equation}

Finally, we need the notion of locality-preserving isomorphism. 
For a interval $I = [a,b]\subset \mathbb{Z}$ and $r\geq 0$, we define $I^{+r} = [a-r,b+r]$. 

\begin{definition}[cf. \cite{Jones2024DHRBimodules}]
    Given $\{\mathcal{A}_I\}_I$ and $\{\mathcal{B}_I\}_I$ net of $C^*$-algebras over $\mathbb{Z}$ with inclusion maps $\iota^{\mathcal{A}}_{J,I}$ and $\iota^{\mathcal{B}}_{J,I}$ respectively. 
    A bounded spread isomorphism is defined by two systems of $*$-homomorphisms $\alpha_I: \mathcal{A}_I\rightarrow \mathcal{B}_{I^{+r}}$, $\beta_I: \mathcal{B}_I\rightarrow \mathcal{A}_{I^{+r}}$ such that 
    \begin{enumerate}
        \item for any $I\subseteq J$, $\iota^{\mathcal{B}}_{J^{+r},I^{+r}}\circ \alpha_I = \alpha_J\circ \iota^{\mathcal{A}}_{J,I}$;
        \item for any $I\subseteq J$, $\iota^{\mathcal{A}}_{J^{+r},I^{+r}}\circ \beta_I = \beta_J\circ \iota^{\mathcal{B}}_{J,I}$;
        \item for each $I$, we have
        \begin{equation}\label{isomorphism between inductive systems}
            \beta_{I^{+r}}\circ\alpha_I = \iota^{\mathcal{A}}_{I^{+2r},I},\quad \alpha_{I^{+r}}\circ\beta_I = \iota^{\mathcal{B}}_{I^{+2r},I}. 
        \end{equation}
    \end{enumerate}
\end{definition}

Now we show the equivalence between the boundary algebra derived from the model and the fusion spin chain $\mathcal{A}(\mathcal{C},\rho)$. 
We order the vertices of the dual graph of $\Gamma$ on the reflection line from left to right and label them by $\mathbb{Z}$, and fix $h>0$. 
For each interval $I=[a,b]\in \mathcal{I}(\mathbb{Z})$, let $\Lambda^+_I$ be a parallelogram of shape $(b-a+1,h)$ with the base point on the reflection line being $a$. 
Define $\Lambda_I = \Lambda^+_I \cup \theta(\Lambda^+_I)$, then $\{\Lambda_I\}_{I\in \mathcal{I}(\mathbb{Z})}$ is a up-ward directed family of finite symmetric regions in $\Gamma$. 
\begin{theorem}\label{thm:: fusion spin chain as boundary algebra}
    Let $\varPhi$ in \ref{eqn: interaction of new model} be defined by $\rho\in \mathcal{C}$, with a weight $k$ generating set $\mathcal{G}$. 
    Then the local net of $C^*$-algebras $\{\mathcal{A}_+(\Pi(\Lambda_I))\widehat{\Pi}(\Lambda_I)\}_{I\in \mathcal{I}(\mathbb{Z})}$ is isomorphic to the fusion spin chain $\mathcal{A}(\mathcal{C},\rho)$ through a bounded spread isomorphism with $r = k$. 
\end{theorem}
\begin{proof}
    Write $\mathcal{A}_I=\mathcal{A}_+(\Pi(\Lambda_I))\widehat{\Pi}(\Lambda_I)$ and $\mathcal{B}_I=\mathcal{A}(\mathcal{C},\rho)_I$.
    For $|I|\geq k$, denote the isomorphism in Proposition \ref{prop:: endomorphism algebras as boundary algebra} by $\pi_{I}:\mathcal{B}_I\rightarrow\mathcal{A}_I$.
    We first verify that these isomorphisms respect inclusions.
    Let $I=[a,b]\subseteq J=[c,d]$, with $|I|\geq k$, and let $f\in\End(\rho^{|I|})$.
    Since $\mathrm{D}_{\rho^n}=\mathrm{D}_{\rho}^{\otimes n}$, we have
    \begin{equation}
        \mathrm{D}_{\rho^{|J|}}^{-1/2}\iota^{\mathcal{B}}_{J,I}(f)\mathrm{D}_{\rho^{|J|}}^{1/2}
        =\mathrm{Id}_{\rho^{a-c}}\otimes\left(\mathrm{D}_{\rho^{|I|}}^{-1/2}f\mathrm{D}_{\rho^{|I|}}^{1/2}\right)\otimes\mathrm{Id}_{\rho^{d-b}}.
    \end{equation}
    Frustration-freeness of the interaction gives $(\Pi(\Lambda_I^+)\otimes\mathrm{I})\Pi(\Lambda_J^+)=\Pi(\Lambda_J^+)$.
    Every ground state on $\Lambda_J^+$ lies in the string-net subspace of its bottom row.
    Thus, the inclusion formula \eqref{eqn: inclusion map for the boundary algebra}, Lemma \ref{lemma:: propagate actions to bigger regions}, and Corollary \ref{corollary:: entanglement support as local ground state projection} give
    \begin{equation}\label{eqn:: boundary isomorphisms respect inclusions}
        \iota^{\mathcal{A}}_{J,I}\circ\pi_{I}
        =\pi_{J}\circ\iota^{\mathcal{B}}_{J,I}.
    \end{equation}
    For every interval $I$, we have $|I^{+k}|\geq k$, so we may define unital $*$-homomorphisms
    \begin{equation}
        \alpha_I=\pi_{I^{+k}}^{-1}\circ\iota^{\mathcal{A}}_{I^{+k},I},\quad
        \beta_I=\pi_{I^{+k}}\circ\iota^{\mathcal{B}}_{I^{+k},I}.
    \end{equation}
    When $|I|\geq k$, $\alpha_I$ and $\beta_I$ reduce to $\pi_{I}^{-1}$ and $\pi_{I}$ followed by the respective inclusions. 
    When $|I|<k$, Proposition \ref{prop:: endomorphism algebras as boundary algebra} gives $\mathcal{A}_I=\mathbb{C}\Pi(\Lambda_I^+)$, so $\alpha_I$ is the scalar inclusion, while $\beta_I$ is the restriction of $\pi_{I^{+k}}$ to the embedded copy of $\End(\rho^{|I|})$.

    For $I\subseteq J$, applying \eqref{eqn:: boundary isomorphisms respect inclusions} to $I^{+k}\subseteq J^{+k}$ and using isotony gives
    \begin{equation}
        \begin{aligned}
            \iota^{\mathcal{B}}_{J^{+k},I^{+k}}\alpha_I
            &=\pi_{J^{+k}}^{-1}\iota^{\mathcal{A}}_{J^{+k},I}
            =\alpha_J\iota^{\mathcal{A}}_{J,I},\\
            \iota^{\mathcal{A}}_{J^{+k},I^{+k}}\beta_I
            &=\pi_{J^{+k}}\iota^{\mathcal{B}}_{J^{+k},I}
            =\beta_J\iota^{\mathcal{B}}_{J,I}.
        \end{aligned}
    \end{equation}
    Finally, the same identity for $I^{+k}\subseteq I^{+2k}$ gives
    \begin{equation}
        \beta_{I^{+k}}\alpha_I=\iota^{\mathcal{A}}_{I^{+2k},I},\quad
        \alpha_{I^{+k}}\beta_I=\iota^{\mathcal{B}}_{I^{+2k},I}.
    \end{equation}
    Therefore $(\alpha_I,\beta_I)_I$ is a bounded spread isomorphism with $r=k$.
\end{proof}

\begin{lemma}
    Let $\mathcal{H}_{\pm}$ be finite-dimensional Hilbert spaces, $\hat{\theta}: \mathcal{H}_+\rightarrow \mathcal{H}_-$ an anti-unitary, and $\Pi$ a projection on $\mathcal{H}_-\otimes\mathcal{H}_+$ that is reflection positive. 
    Suppose that there exists an invertble positive operator $D\in \mathcal{A}_+(\Pi)\widehat{\Pi}$, such that 
    \begin{equation}
        \Theta(DA^\dagger D^{-1})\Pi = A\Pi,\quad A\in \mathcal{A}_+(\Pi)\widehat{\Pi}.
    \end{equation}
    Denote by $\omega$ the state on $\mathcal{A}_+(\Pi)\widehat{\Pi}$ defined by the maximally mixed state in the range of $\Pi$. 
    Then the modular automorphism group $\sigma^\omega_t$ of $\omega$ is given by
    \begin{equation}
        \sigma^\omega_t(A) = D^{2it}AD^{-2it},\quad t\in\mathbb{R}. 
    \end{equation}
\end{lemma}

\begin{proof}
    By \cite[Theorem 3.18]{LiuZhao2025RPTO}, every $\ket{\psi}\in\range(\Pi)$ has the form
    \begin{equation}
        \ket{\psi}=(\mathrm{I}\otimes W)\ket{\phi_{\mathrm{PF}}},\qquad W\in\mathrm{Sym}_{\mathcal{H}_+}(\Pi)\widehat{\Pi},
    \end{equation}
    where $\ket{\phi_{\mathrm{PF}}}$ is the canonical Perron--Frobenius ground state.
    Set $\Xi=\mathcal{O}(\phi_{\mathrm{PF}})$, with $\mathcal{O}(\hat{\theta}\eta\otimes\xi)=\ket{\xi}\bra{\eta}$.
    By \cite[Proposition 3.13 and Theorem 3.18]{LiuZhao2025RPTO}, $\Xi$ is strictly positive on $\widehat{\Pi}\mathcal{H}_+$ and commutes with $\mathrm{Sym}_{\mathcal{H}_+}(\Pi)\widehat{\Pi}$.
    Since $\mathcal{A}_+(\Pi)=\mathrm{Sym}_{\mathcal{H}_+}(\Pi)'$, $\Xi\in\mathcal{A}_+(\Pi)\widehat{\Pi}$.
    All inverses and powers below are taken in this corner.

    Apply the assumed identity to $\ket{\phi_{\mathrm{PF}}}$.
    Since $\mathcal{O}((\Theta(X)\otimes\mathrm{I})\phi)=\mathcal{O}(\phi)X^\dagger$, positivity of $D$ gives
    \begin{equation}
        A\Xi=\Xi D^{-1}AD,\qquad A\in\mathcal{A}_+(\Pi)\widehat{\Pi}.
    \end{equation}
    Taking adjoints and replacing $A^\dagger$ by $A$, we obtain
    \begin{equation}
        DAD^{-1}=\Xi A\Xi^{-1},\qquad A\in\mathcal{A}_+(\Pi)\widehat{\Pi}.
    \end{equation}
    Therefore, $\Xi^{-1}D$ is strictly positive and central in $\mathcal{A}_+(\Pi)\widehat{\Pi}$.
    By \cite[Equations (5.23)--(5.24)]{LiuZhao2025RPTO}, the state induced by the maximally mixed state on $\range(\Pi)$ is faithful and has modular group $A\mapsto\Xi^{2it}A\Xi^{-2it}$.
    The central factor $\Xi^{-1}D$ cancels under conjugation, so
    \begin{equation}
        \sigma_t^\omega(A)=\Xi^{2it}A\Xi^{-2it}=D^{2it}AD^{-2it},\qquad t\in\mathbb{R},
    \end{equation}
    as claimed.
\end{proof}

\begin{corollary}
    Let $I\subset \mathbb{Z}$ be a finite interval of length $l\geq k$, and let $\Lambda_I$ be as in Theorem \ref{thm:: fusion spin chain as boundary algebra}. 
    Then the modular automorphism of the state $\omega_{\Lambda_I}$ on $\mathcal{A}_+(\Pi(\Lambda_I))\widehat{\Pi}(\Lambda_I)$ is given by
    \begin{equation}
        \sigma_t^{\Lambda_I}\left( \pi_I(f) \right) = \pi_I\left( \mathrm{D}^{-2it}_{\rho^l}f \mathrm{D}^{2it}_{\rho^l} \right),\quad f\in \End(\rho^l),\quad t\in\mathbb{R}.
    \end{equation}
\end{corollary}

\begin{proof}
    Write $\mathrm{D}=\mathrm{D}_{\rho^l}$.
    By frustration-freeness and \eqref{eqn:: pulling through}, we have
    \begin{equation}
        \left(\mathrm{R}_{\theta(\Lambda_h)}(\mathrm{D}^2g\mathrm{D}^{-2})-\mathrm{L}_{\Lambda_h}(g)\right)\Pi(\Lambda_I)=0,\qquad g\in\End(\rho^l).
    \end{equation}
    Taking $g=\mathrm{D}^{-1/2}f\mathrm{D}^{1/2}$ and using the definition of $\pi_I$ in Proposition \ref{prop:: endomorphism algebras as boundary algebra}, we obtain
    \begin{equation}
        \mathrm{R}_{\theta(\Lambda_h)}(\mathrm{D}^{3/2}f\mathrm{D}^{-3/2})\Pi(\Lambda_I)=\pi_I(f)\Pi(\Lambda_I).
    \end{equation}
    On the other hand, Lemma \ref{lemma:: reflection exchanges L and R} gives
    \begin{equation}
        \Theta\left(\pi_I(\mathrm{D}^{-1}f^*\mathrm{D})\right)\Pi(\Lambda_I)
        =\mathrm{R}_{\theta(\Lambda_h)}(\mathrm{D}^{3/2}f\mathrm{D}^{-3/2})\Pi(\Lambda_I),
    \end{equation}
    where the reflected half-region projection is absorbed by $\Pi(\Lambda_I)$ by frustration-freeness.
    Since $\pi_I$ is a $*$-isomorphism, $\pi_I(\mathrm{D}^{-1})$ is positive and invertible in the boundary algebra, and these identities verify the hypothesis of the preceding lemma with $D=\pi_I(\mathrm{D}^{-1})$.
    Therefore
    \begin{equation}
        \sigma_t^{\Lambda_I}(\pi_I(f))
        =\pi_I(\mathrm{D}^{-2it})\pi_I(f)\pi_I(\mathrm{D}^{2it})
        =\pi_I(\mathrm{D}^{-2it}f\mathrm{D}^{2it}),\qquad t\in\mathbb{R}.
    \end{equation}
\end{proof}
\begin{remark}
    Notice that when $\rho$ consists of simple objects with equal quantum dimensions, which is the case when $\rho$ is simple, the one-parameter group of automorphisms on the boundary algebra is trivial. 
    Concerning the results in \cite{Ogata2024TypeOfConeAlgebras} and \cite[Corollary 5.18]{JNPW2025}, this suggests that the cone von Neumann algebras in these models are of type $II_{\infty}$. 
\end{remark}
\section{Examples}\label{Sec:: Examples}

\subsection{Fibonacci fusion category}

Here we derive the qubit model in Section \ref{Sec:: two examples} from $\mathcal{C} = \mathrm{Fib}$, the Fibonacci category with two simple objects $\{\mathbb{1},\tau\}$, and take $\rho = \tau$. 
The fusion rules of $\mathrm{Fib}$ are given by 
\begin{equation}
    \tau^2 = \mathbb{1}\oplus \tau
\end{equation}
which determines the quantum dimension $d_{\tau} = \phi$, where $\phi = \frac{1+\sqrt{5}}{2}$. 
The absence of fusion multiplicities allows us to redistribute the local degrees of freedom in the model to the edges, as in the setup in \cite{LevinWen2005}. 
This results in a system of qubits assigned to un-colored edges, with basis $\ket{0}$ for $\tau$ and $\ket{1}$ for $\mathbb{1}$. 
The constraint from the fusion rule requires no two $\ket{1}$ appear consecutively in a row, which is represented by the projector $P_{ii+1} = \ket{11}\bra{11}$ for the vertex between edges $i,i+1$. 
The allowed configurations around each vertex correspond to vectors in $\mathbb{C}^2\otimes \mathbb{C}^2$ as: 
\begin{equation}
    \begin{aligned}
        \vcenter{\hbox{\begin{tikzpicture}
        \ThinLine[orange] (0:0) -- (150:0.5);
        \ThinLine[dashed] (0:0) -- (30:0.5);
        \ThinLine[orange] (0:0) -- (-90:0.5);
    \end{tikzpicture}}}/ \vcenter{\hbox{\begin{tikzpicture}
        \ThinLine[orange] (0:0) -- (-150:0.5);
        \ThinLine[dashed] (0:0) -- (-30:0.5);
        \ThinLine[orange] (0:0) -- (90:0.5);
    \end{tikzpicture}}}\mapsto \ket{01},\quad \vcenter{\hbox{\begin{tikzpicture}
        \ThinLine[orange] (0:0) -- (30:0.5);
        \ThinLine[dashed] (0:0) -- (150:0.5);
        \ThinLine[orange] (0:0) -- (-90:0.5);
    \end{tikzpicture}}}/ \vcenter{\hbox{\begin{tikzpicture}
        \ThinLine[orange] (0:0) -- (-30:0.5);
        \ThinLine[dashed] (0:0) -- (-150:0.5);
        \ThinLine[orange] (0:0) -- (90:0.5);
    \end{tikzpicture}}}\mapsto \ket{10},\quad \vcenter{\hbox{\begin{tikzpicture}
        \ThinLine[orange] (0:0) -- (30:0.5);
        \ThinLine[orange] (0:0) -- (150:0.5);
        \ThinLine[orange] (0:0) -- (-90:0.5);
    \end{tikzpicture}}} / \vcenter{\hbox{\begin{tikzpicture}
        \ThinLine[orange] (0:0) -- (-30:0.5);
        \ThinLine[orange] (0:0) -- (-150:0.5);
        \ThinLine[orange] (0:0) -- (90:0.5);
    \end{tikzpicture}}}\mapsto \ket{00}.
    \end{aligned}
\end{equation}
Mapping the un-colored edges to vertices then removing the colored edges, we obtain the square lattice as in Fig. \ref{figure:: 2-qubit Fib}. 

It is well known that for each $n\geq 1$, the $C^*$-algebra $\End(\tau^n)$ is the semi-simple quotient of the Temperley--Lieb algebra on $n$ strands \cite{TemperleyLieb1971,Jones1983}, which is generated by elements $E_1,E_2,\dots,E_{n-1}$ satisfying relations
\begin{equation}
    \begin{aligned}
        &E^*_i = E_i,\quad  E^2_i = \varphi E_i\\
        &E_iE_jE_i = E_i,\quad \vert i-j\vert = 1\\
        &E_iE_j = E_jE_i,\quad \vert i-j\vert > 1. 
    \end{aligned}
\end{equation}
Hence a natural choice of a weight-$2$ generating set is just $\mathcal{G} = \{E_1\}$. 
An explicit form of the plaquette interaction can be obtained from the diagrammatic calculus, in which $E_1$ is represented by the morphism $\vcenter{\hbox{\begin{tikzpicture}
    \draw[orange,line width=0.8pt] (0,0) arc (180:360:0.25);
    \draw[orange,line width=0.8pt] (0,-0.75) arc (180:0:0.25);
\end{tikzpicture}}}$. 
The local operator $\mathrm{L}_{\Lambda'}(E_1)$ acts on $4$ consecutive qubits (which correspond to the un-colored edges in the figure below) in a row:
\begin{equation}
    \vcenter{\hbox{\begin{tikzpicture}
    \ThinLine (150:1) -- (90:1) -- (30:1);
    \ThinLine (150:1) -- (150:2);
    \ThinLine[orange] (150:1) -- (180:0.866);
    \ThinLine[orange] (30:1) -- (0:0.866);
    \ThinLine[orange] (90:1.5) -- (90:1);
    \ThinLine (30:1) -- (30:2);
    \draw[fill = black] (150:1) ellipse (0.05 and 0.05);
    \draw[fill = black] (90:1) ellipse (0.05 and 0.05);
    \draw[fill = black] (30:1) ellipse (0.05 and 0.05);
    \node at (2.5,0.75) {$\Lambda'$};
    \node[below=2pt] at (150:1.5) {$1$};
    \node[above=2pt] at (125:0.866) {$2$};
    \node[above=2pt] at (55:0.866) {$3$};
    \node[below=2pt] at (30:1.5) {$4$};
\end{tikzpicture}}}
\end{equation}
Then by expanding the action of $E$ on the $\Hom(\gamma\tau\gamma,\tau^2)$ using the $F$-symbols of $\mathrm{Fib}$, we obtain the expression in \eqref{eqn:: Jones projections on qubits}. 
Finally, note that $\mathrm{D}_{\tau}$ is a multiple of identity since $\tau$ is simple. 
Hence by \eqref{eqn:: F_S}, we have
\begin{equation}
    F_S = \left\vert \mathrm{R}_{\Lambda} \left(  E_1 \right) -  \mathrm{L}_{\Lambda'}\left( E_1 \right) \right\vert^2 = \phi \left( \mathrm{R}_{\Lambda}(E_1) + \mathrm{L}_{\Lambda'}(E_1) \right) - 2\mathrm{R}_{\Lambda}(E_1)\otimes \mathrm{L}_{\Lambda'}(E_1). 
\end{equation}

We remark that the above construction generalizes to the family of $A_{k+1}$ fusion categories with simple objects $0,1,\dots, k$, with $1$ being the chosen object. 
This is because the algebra $\End(1^n)$ is given by the semisimple quotient of Temperley-Lieb algebra on $n$ strands with $\delta = 2\cos \frac{\pi}{k+2}$. 
The construction results in a system of qudits with dimention $k+1$. 

\subsection{Levin--Wen string-net}

To recover the Levin--Wen string-net models \cite{LevinWen2005}, we choose $\rho = \bigoplus_{c\in \Irr(\mathcal{C})} c$. 
We take the generating set $\mathcal{G}_{\mathrm{LW}}\subset \End(\rho^2)$ to be the set of morphisms of the following form: 
\begin{equation}
    f(\xi,\eta): = \vcenter{\hbox{%
    \begin{tikzpicture}[
        scale=0.6,
        line cap=round,
        % every node/.style={font=\large},
        mid arrow/.style={
            postaction={
                decorate,
                decoration={
                    markings,
                    mark=at position 0.50 with {
                        \arrow{
                            Stealth[
                                length=4.125pt,
                                width=3.75pt
                            ]
                        }
                    }
                }
            }
        }
    ]

        % Main coordinates
        \coordinate (L) at (0,1);
        \coordinate (R) at (2,1);
        \draw[fill=black] (L) ellipse (0.1 and 0.1);
        \draw[fill=black] (R) ellipse (0.1 and 0.1);
        % Vertical lines
        \draw[line width=0.675pt] (0,0) -- (0,2);
        \draw[line width=0.675pt] (2,0) -- (2,2);

        % Curved arrow, directed from right to left
        \draw[
            line width=0.675pt,
            mid arrow
        ]
            (R)
            .. controls (1.50,1.52) and (0.50,1.52) ..
            (L);

        % Left-hand labels: text size remains unchanged
        \node[above]    at (0,2) {$a_{1}$};
        \node[below]    at (0,0) {$b_{1}$};
        \node[left=5pt] at (L)   {$\xi$};

        % Right-hand labels: text size remains unchanged
        \node[above]     at (2,2) {$\overline{a_{2}}$};
        \node[below]     at (2,0) {$\overline{b_{2}}$};
        \node[right=5pt] at (R)   {$\overline{\eta}$};

        % Arrow label
        \node[below=4pt] at (1,0.82) {$j$};

    \end{tikzpicture}%
}}
\end{equation}
where $a_1,a_2,b_1,b_2,j\in \Irr(\mathcal{C})$, and $\xi,\eta$ are taken from an orthonormal basis of corresponding morphism spaces, with $\overline{\eta}$ is the modular conjugation. 
It is easy to verify that these morphisms indeed form a weight-$2$ generating set. 
In particular, that $\End(\rho^n)$ is generated by $f(\xi,\eta)$ and their translations, since every morphism in $\End(\rho^n)$ admits a ladder decomposition:
\begin{equation}
    \vcenter{\hbox{
\begin{tikzpicture}[
    x=1cm,
    y=1cm,
    baseline=(current bounding box.center),
    strand/.style={line width=0.8pt, line cap=round},
    vertex/.style={
        circle, fill=black, inner sep=0pt, minimum size=4pt
    },
    scale = 0.75
]
    % Vertical strands, ordered from left to right.
    \foreach \x in {0,1,2,4.2}{
        \draw[strand] (\x,0) -- (\x,2.5);
    }

    % Horizontal rungs between adjacent strands.
    \draw[strand] (0,2.1) -- (1,2.1);
    \draw[strand] (1,1.6) -- (2,1.6);

    % Partial rungs adjoining the omitted portion.
    \draw[strand] (2,0.8)   -- (2.4,0.8);
    \draw[strand] (3.8,0.4) -- (4.2,0.4);

    % Junction vertices.
    \foreach \x/\y in {
        0/2.1,
        1/2.1,
        1/1.6,
        2/1.6,
        2/0.8,
        4.2/0.4
    }{
        \node[vertex] at (\x,\y) {};
    }

    % Omitted intermediate strands and rungs.
    \node at (3.1,1.25) {$\cdots$};

    % Optional labels on the boundary strands:
    % \foreach \x in {0,1,2,4.2}{
    %     \node[below] at (\x,0) {$\rho$};
    %     \node[above] at (\x,2.5) {$\rho$};
    % }
\end{tikzpicture}
}}
\end{equation}
Since the generating set has weight $2$, the operator $F_S$ is supported on a single plaquette $S$. 
Note that for $\xi\in\Hom(a_1\otimes j,b_1)$ and $\eta\in \Hom(a_2\otimes j,b_2)$, we have 
\begin{equation}
    \mathrm{D}^{-1}_{\rho^2}f(\xi,\eta)\mathrm{D}_{\rho^2} = \sqrt{\frac{d_{b_1} d_{b_2}}{d_{a_1} d_{a_2}}}\, f(\xi,\eta). 
\end{equation} 
Put $\mu=(\sum_{j\in\Irr(\mathcal{C})}d_j^2)^{1/2}$. Then
\begin{equation}\label{eqn:: LW scalar identity}
    \sum_{f\in \mathcal{G}_{\mathrm{LW}}} f^*  \mathrm{D}^{-1}_{\rho^2} f \mathrm{D}_{\rho^2}
    = \mu^2\mathrm{Id}_{\rho^2}.
\end{equation}
Indeed, for fixed $a,j\in\Irr(\mathcal{C})$, the normalization of the trivalent morphisms gives
\begin{equation}
    \sum_{b\in\Irr(\mathcal{C})}\sum_{\xi\in\mathrm{ONB}(\Hom(a\otimes j,b))}
    \sqrt{\frac{d_b}{d_a d_j}}\,\xi^*\xi
    =\mathrm{Id}_{a\otimes j}.
\end{equation}
For fixed $a_1,a_2,j$, the coefficient multiplying $f(\xi,\eta)^*f(\xi,\eta)$ is
\begin{equation}
    \sqrt{\frac{d_{b_1}d_{b_2}}{d_{a_1}d_{a_2}}}
    =d_j\sqrt{\frac{d_{b_1}}{d_{a_1}d_j}}\sqrt{\frac{d_{b_2}}{d_{a_2}d_j}}.
\end{equation}
Applying the resolution of identity at both vertices leaves the identity on $a_1\otimes\overline{a_2}$ and a closed $j$-loop, whose value is $d_j$.
Thus the sum over $b_1,b_2,\xi,\eta$ contributes $d_j^2\mathrm{Id}_{a_1\otimes\overline{a_2}}$.
Summing over $j$ and all summands $a_1\otimes\overline{a_2}$ of $\rho^2$ proves the claimed identity.

Let $B_p^j$ denote insertion of a $j$-loop into the plaquette $p$ in $S$, so that $B_p=\mu^{-2}\sum_j d_j B_p^j$.
Resolving the loop across the reflection axis gives the coefficient $d_j^{-1}\sqrt{d_{b_1}d_{b_2}/(d_{a_1}d_{a_2})}$ for each pair $\xi,\eta$ as in \cite{JL2020}.
Distributing the square-root factor evenly across the reflection line, we obtain
\begin{equation}
    \begin{aligned}
        &\sum_{f\in\mathcal{G}_{\mathrm{LW}}}
        \Theta\left(\mathrm{L}_{S_+}\left(\mathrm{D}_{\rho^2}^{-1/2}f\mathrm{D}_{\rho^2}^{1/2}\right)\right)
        \otimes\mathrm{L}_{S_+}\left(\mathrm{D}_{\rho^2}^{-1/2}f\mathrm{D}_{\rho^2}^{1/2}\right) = \sum_j d_j B_p^j=\mu^2 B_p.
    \end{aligned}
\end{equation}
Therefore, \eqref{eqn:: LW scalar identity} and the last line of \eqref{eqn:: reflection positive expansion of F_S} give
\begin{equation}
    F_S=2\mu^2(\mathrm{I}-B_p)
\end{equation}
on the string-net subspace. 
This shows that this model exhibits the same set of local ground states as Levin--Wen string-net model. 

\subsection{Strong tensor generating object}

In general, one can take $\rho = X$ to be any self-dual, strongly tensor generating object in $\mathcal{C}$. 
This means that there exists $L\geq 1$, such that $X^{L}$ contains every simple object of $\mathcal{C}$ as a summand. 
The depth of $X$ is then defined to be the minimum $L$ that the above property holds.  
Given a self-dual object $X$ of depth $L$, a generating set of weight $L+1$ can be any basis $\mathcal{G} = \{\eta\}$ of $\End(X^{L+1})$ that satisfies the following two conditions:
\begin{enumerate}\label{condition for a generating set}
    \item for any $\eta\in \mathcal{G}$, we have $\eta^*\in \mathcal{G}$;
    \item for any $\eta\in \mathcal{G}$, there are simple objects $a_1,\dots, a_{L+1}$ and $b_1,\dots ,b_{L+1}$ such that 
    \begin{equation}
        \eta = (q_{b_1}\otimes\cdots\otimes q_{b_{L+1}}) \eta (q_{a_1}\otimes\cdots\otimes q_{a_{L+1}}). 
    \end{equation}
\end{enumerate}
That the set $\mathcal{G}$ also satisfies the third condition in Definition \ref{def:: generating set} is given by the following lemma. 

\begin{lemma}
    Let $X$ be a self-dual object of depth $L$ in a unitary fusion category $\mathcal{C}$. 
    Then for any $m\geq L+1$, $\End(X^m)$ is generated by the subalgebras 
    \begin{equation}
        \mathrm{Id}_{X^{k}}\otimes \End(X^{L+1})\otimes \mathrm{Id}_{X^{m-k-L-1}},\quad 0\leq k\leq m-L-1. 
    \end{equation}
\end{lemma}
\begin{proof}
    By induction, it suffices to prove that for any $k\geq L+1$, we have 
    \begin{equation}
        \End(X^{k+1}) = (\End(X^{k})\otimes\mathrm{Id}_X)\left( \mathrm{Id}_{X^{k-L}}\otimes \End(X^{L+1}) \right)(\End(X^{k})\otimes\mathrm{Id}_X). 
    \end{equation}
    Since $X$ is self-dual, we can choose a duality morphism $c:\mathbb{1}\to X\otimes X$.
    Put $u=\mathrm{Id}_{X^{L-1}}\otimes c$ and define the $L$-th unnormalized Jones element $e_L=uu^*\in\End(X^{L+1})$.
    Since $k\geq L+1$, $X^{k-1}$ contains every simple object as a summand.
    Semisimplicity therefore implies that $\End(X^{k+1})$ is spanned by $fg^*$ with $f,g\in\Hom(X^{k-1},X^{k+1})$.
    Consequently, we obtain (see also  \cite[Lemma 5.7(a)]{KodiyalamLandauSunder2003})
    \begin{equation}
        \begin{aligned}
            \End(X^{k+1})
        =\operatorname{span}\{(a\otimes\mathrm{Id}_X)\left( \mathrm{Id}_{X^{k-L}}\otimes e_L \right)(b\otimes\mathrm{Id}_X):a,b\in\End(X^k)\}.
        \end{aligned}
    \end{equation}
    This proves the claim. 
\end{proof}

\begin{theorem}\label{thm:: boundary algebra determined by tensor generator}
    Let $\mathcal{C}$ be a unitary fusion category, and $X$ a self-dual strongly tensor generating object with depth $L$. 
    Then for any basis $\mathcal{G}$ of $\End(X^{L+1})$ that satisfies the conditions \ref{condition for a generating set}, the interaction \eqref{eqn: interaction of new model} has boundary algebra that is bounded spread isomorphic to the fusion spin chain $\mathcal{A}(\mathcal{C},X)$. 
\end{theorem}
\begin{proof}
    This is a direct consequence of Theorem \ref{thm:: fusion spin chain as boundary algebra}. 
\end{proof}
\begin{conflictofinterest}
There is no conflict of interest. 
\end{conflictofinterest}

\begin{dataavailability}
Data sharing is not applicable to this article as no new data were created or analyzed in this study. 
\end{dataavailability}

\bibliographystyle{alpha}
\bibliography{beyondRGFP}
\end{document}